\documentclass[preprint]{imsart}

\usepackage{xr-hyper}  

\RequirePackage[OT1]{fontenc}
\RequirePackage{amsthm,amsmath}
\usepackage[round, longnamesfirst]{natbib}
\RequirePackage[colorlinks,citecolor=blue,urlcolor=blue]{hyperref}

\arxiv{arXiv:1711.00572}

\usepackage[utf8]{inputenc}
\usepackage{makeidx}
\usepackage{amsfonts}
\usepackage{amssymb}
\usepackage{gensymb}
\usepackage{caption}
\usepackage{float}
\usepackage{subfig} 
\usepackage[final]{graphicx}
\usepackage{bbm}
\usepackage{lipsum}
\usepackage[linesnumbered,ruled,vlined,resetcount,algosection]{algorithm2e} 
\usepackage{xspace} 
\usepackage{enumitem}
\usepackage[english]{babel} 
\usepackage{bm} 
\usepackage{bbm} 
\usepackage[margin = 1in]{geometry} 
\usepackage[T1]{fontenc}	 
\usepackage{lmodern}
\usepackage{amsthm} 
\usepackage{graphicx} 
\usepackage{booktabs} 
\usepackage{multicol} 
\usepackage{tocstyle} 
\usepackage[toc,page]{appendix}
\usepackage{mathrsfs} 
\usepackage{scalerel,stackengine} 
\usepackage{placeins}
\usepackage{filecontents} 
\usepackage{pdfpages} 

\startlocaldefs

\stackMath

\newcommand{\numbereqn}{\addtocounter{equation}{1}\tag{\theequation}} 

\renewcommand{\epsilon}{\varepsilon}

\newtheorem{theorem}{Theorem}[section]
\newtheorem{lemma}{Lemma}[section]
\newtheorem{prop}{Proposition}[subsection]
\theoremstyle{remark}
\newtheorem{remark}{Remark}[section]

\newcommand\widecheck[1]{%
	\savestack{\tmpbox}{\stretchto{%
			\scaleto{%
				\scalerel*[\widthof{\ensuremath{#1}}]{\kern-.6pt\bigwedge\kern-.6pt}%
				{\rule[-\textheight/2]{1ex}{\textheight}}
			}{\textheight}%
		}{0.5ex}}%
	\stackon[1pt]{#1}{\scalebox{-1}{\tmpbox}}%
}

\renewcommand{\tilde}{\widetilde}
\renewcommand{\hat}{\widehat}
\newcommand{\R}{\mathbb{R}}
\newcommand{\X}{\mathscr{X}}
\newcommand{\B}{\mathscr{B}}
\newcommand{\Z}{\mathcal{Z}}

\newcommand{\mcx}{\widetilde{X}}
\renewcommand{\sp}{\lambda}
\newcommand{\one}{\mathbbm{1}}
\newcommand{\hs}{\text{HS}}
\newcommand{\N}{\text{N}}
\newcommand{\pg}{\text{Polya-Gamma}}
\newcommand{\fb}{\text{F}}
\DeclareMathOperator{\tr}{trace}
\DeclareMathOperator{\E}{E}
\DeclareMathOperator{\var}{var}
\newcommand{\diag}{\text{diag}}
\newcommand{\psw}{Polson, Scott and Windle\xspace}
\newcommand{\betab}{\bm{\beta}}
\newcommand{\thetab}{\bm{\theta}}
\newcommand{\betat}{\tilde{\bm{\beta}}}
\newcommand{\wb}{\bm{w}}
\newcommand{\yb}{\bm{y}}
\newcommand{\bbo}{\bm{b}}

\newcommand{\ut}{\tilde{\bm{u}}}
\newcommand{\mub}{\bm{\mu}}
\newcommand{\gt}{\tilde{g}}
\newcommand{\Ct}{\tilde{C}}

\newcommand{\zb}{\bm{z}}
\newcommand{\It}{\tilde{I}}
\newcommand{\ub}{\bm{u}}
\newcommand{\nut}{\tilde{\nu}}
\newcommand{\Ut}{\tilde{U}}
\newcommand{\betas}{{\betab^*}}
\newcommand{\mustar}{{\mub^*}}
\newcommand{\Hmn}{{\hat{H}^{(N)}_m}}
\newcommand{\Hmnc}{{\check{H}^{(N)}_m}}

\newcommand{\Smn}{{\hat{S}^{(N)}_m}}

\newcommand{\pt}{\tilde{p}}
\newcommand{\pit}{\tilde{\pi}}
\newcommand{\fs}{Fr{\"u}hwirth-Schnatter }

\DeclareMathOperator{\spmax}{\lambda_{\max}}

\allowdisplaybreaks
\newcommand*{\supplproofname}{Proof}
\newenvironment{supplproof}[1][\supplproofname]{\begin{proof}[#1]}{\end{proof}}

\endlocaldefs

\begin{document}
	
	\begin{frontmatter}
		\title{Consistent estimation of the spectrum of trace class data augmentation algorithms}
		\runtitle{Spectrum estimation for trace class DA algorithms}
		
		\begin{aug}
		\author{\fnms{Saptarshi} \snm{Chakraborty}\ead[label=e1]{c7rishi@ufl.edu}}
		\and
		\author{\fnms{Kshitij} \snm{Khare}\ead[label=e2]{kdkhare@stat.ufl.edu}}
		\address{}
		
		\thankstext{t2}{Kshitij Khare (\printead{e2}) is Associate Professor, Department of Statistics, University of Florida. Saptarshi Chakraborty (\printead{e1}) is PhD candidate, Department of Statistics, University of Florida.}
		
		\affiliation{University of Florida}

		\runauthor{Chakraborty and Khare}
		
		\end{aug}
	
	\begin{abstract}
		\quad Markov chain Monte Carlo is widely used in a variety of scientific applications to generate approximate samples from intractable distributions. A thorough understanding of the convergence and mixing properties of these Markov chains can be obtained by studying the spectrum of the associated Markov operator. While several methods to bound/estimate the second largest eigenvalue are available in the literature, very few general techniques for consistent estimation of the entire spectrum have been proposed. Existing methods for this purpose require the Markov transition density to be available in closed form, which is often not true in practice, especially in modern statistical applications.  In this paper, we propose a novel method to consistently estimate the entire spectrum of a general class of Markov chains arising from a popular and widely used statistical approach known as Data Augmentation. The transition densities of these Markov chains can often only be expressed as intractable integrals. We illustrate the applicability of our method using real and simulated data. 
	\end{abstract}

		\begin{keyword}[class=MSC]
			\kwd[Primary ]{60J22}
			\kwd[; secondary ]{62F15}
		\end{keyword}
		
		\begin{keyword}
			\kwd{MCMC convergence}
			\kwd{trace class Markov operators}
			\kwd{eigenvalues of Markov operators}
			\kwd{Data Augmentation algorithms}
		\end{keyword}

\end{frontmatter}

\section{Introduction}

\noindent
Markov chain Monte Carlo (MCMC) techniques have become an indispensable tool in modern computations. With major applications in high dimensional settings, MCMC methods are routinely applied in various scientific disciplines. A major application of MCMC is to evaluate intractable integrals. To elaborate, let $(\X, \B, \nu)$ be an arbitrary measure space and let $\Pi$ be a probability measure on $\X$, with associated probability density $\pi(\cdot)$ (with respect to the measure $\nu$).  The quantity of interest is the integral
\[
\pi g := \int_{\X} g(x) \:d\Pi(x)
\]
where $g$ is a well-behaved function. In many modern applications, the above integral is highly intractable. In particular, it is not available in closed form, a (deterministic) numerical integration is extremely inefficient (often due to the high dimensionality of $\X$), and it can not be estimated by classical Monte Carlo techniques, as random (IID) generation from $\pi$ is not feasible. In such cases, one typically resorts to Markov Chain Monte Carlo (MCMC) methods. Here, a Markov chain $\mcx = (X_n)_{n\geq0}$  with equilibrium probability distribution $\Pi$ is generated (using any standard MCMC strategies such as Metropolis Hastings, Gibbs sampler etc.) and then a Monte Carlo average based on those Markov chain realizations is used to estimate 
$\pi g$. 

If the Markov chain $\mcx = (X_n)_{n \geq 0}$ is Harris ergodic (which is the case if the corresponding Markov transition density is strictly positive everywhere), then the cumulative averages based on the Markov chain realizations consistently estimate the integral of interest (see \citet{asmussen:glynn:2011}). The accuracy of the estimate depends on two factors: (a) the convergence behavior of the Markov chain to its stationary distribution, and (b) the dependence between the successive 
realizations of the chain at stationarity. An operator theoretic framework provides a 
unified way of analyzing these two related factors. Let us consider the Hilbert space $L^2 (\pi)$ of real valued functions $f$ with finite second moment with respect to $\pi$.  This is a Hilbert space where the inner product of $f, h \in L^2(\pi)$ is defined as
\[
\langle f, h \rangle = \int_{\X} f(x)\: {h(x)}\: \pi(x)\: d\nu(x) =  \int_{\X} f(x)\: {h(x)}\: d\Pi(x)
\]
and the corresponding norm is defined by $\|f\|_{L^2(\pi)} = \sqrt{\langle f,f \rangle}$. Then the Markov transition density $k(\cdot, \cdot)$ corresponding to the Markov chain $\mcx$ defines an operator $K: L^2(\pi) \rightarrow L^2(\pi)$  that maps $f$ to
\begin{equation} \label{K_defn}
(Kf)(x) = \int_{\X} \: k(x, x')  \: f(x') \:d\nu(x') = \int_{\X} \: \frac{k(x, x')}{\pi(x')} \: f(x') \:d\Pi(x').
\end{equation}

We will assume that the Markov chain $\mcx$ is reversible. In terms of the associated operator $K$, this means that $K$ is self-adjoint. The spectrum of the self-adjoint operator $K$,  denoted by $\sp(K)$, is the set of $\lambda$ for which $K -\lambda I$  is non-invertible (here $I$ denotes the identity operator that leaves a function unchanged). It is known that if $K$ is positive, i.e., if $\langle Kf, f \rangle \geq 0$ for all $f \in L^2(\pi)$, (which is the case when $K$ is the operator corresponding to a Data Augmentation (DA) Markov chain, see Section~\ref{sec_mcrma}), then $\sp(K) \subseteq  [0, 1]$  (see, e.g., \citet{retherford:1993}).  

In this paper, we will focus on situations when the (positive, self-adjoint) operator $K$ is trace class, i.e., 
$\sp(K)$ is countable and its elements are summable (\citet[p. 214]{conway:1990}). All finite state space Markov chains trivially correspond to trace class operators. Also, in recent years, an increasingly large class of continuous state space Markov chains from statistical applications have been shown to correspond to trace class operators  (see, e.g., \citet*{choi:roman:2017, chakraborty:khare:2017, pal:khare:hobert:2017,  qin:hobert:2016, hobert:jung:khare:qin:2015, rajarantnam:sparks:khare:zhang:2017}). Let $\sp(K) = \{\lambda_i\}_{i=0}^\infty$, where $(\lambda_i)_{i=0}^\infty$ are the decreasingly ordered eigenvalues of $K$. Then $\lambda_0 = 1$ and the difference $\gamma = \lambda_0 - \lambda_1 = 1 - \lambda_1$ is called the spectral gap for the compact Markov operator $K$. The spectral gap plays a major role in determining the convergence behavior of the Markov chain. In particular, any $g \in L^2(\pi)$ can be expressed as $g = \sum_{i=0}^\infty \eta_i \phi_i $  where $(\phi_i)_{i=0}^\infty$ is the sequence of eigenfunctions corresponding to $K$, and 
\begin{equation} \label{K_lambda1_reln}
\| K^m g - \pi g \|_{L^2 (\pi)}  =  \left(\sum_{i=1}^\infty \eta_i^2 \lambda_i^{2m}\right)^{1/2} \leq  \|g\| \lambda_1^m = \|g\| (1-\gamma)^m
\end{equation}
for any positive interger $m$. Hence, $\gamma$ determines the asymptotic rate of convergence of $\mcx$ to the stationary distribution. Furthermore, $(1-\gamma)^m$ provides maximal absolute correlation between $X_j$ and $X_{j+m}$ when $j$ is large (i.e., $X_j$ is sufficiently close to the target), and enables us to compute upper bounds of the asymptotic variance of MCMC estimators based on ergodic averages. 

There is a substantial literature devoted to finding a theoretical bound for the second largest eigenvalue $\lambda_1 = 1 - \gamma$ of a Markov operator. For finite state 
space Markov chains, see \citet{lawler:sokal:1988, sinclair:jerrum:1989, diaconis:stroock:1991, saloff:laurent:2004, yuen:kong:2000, diaconis:saloff:1996, franccois:2000, diaconis:saloff:1993} to name just a few. In many statistical applications, the Markov chains move on large continuous state spaces, and 
techniques based on drift and minorization (see \cite{rosenthal:1995, jones:hobert:2001}) have been used to get bounds on $\lambda_1$ for some of these Markov chains. However, these bounds can in many cases be way off. Techniques to estimate the spectral gap based on simulation have been developed in \citet{garren:smith:2000, raftery:lewis:1992}, and more recently in \citet{qin:hobert:khare:2017} for trace class data augmentation Markov chains. 

While bounding or estimating the spectral gap is clearly useful, a much more detailed and accurate picture of the convergence can be obtained by analyzing the entire spectrum of the Markov operator, as explained below. 
\begin{enumerate}[label=(\roman*)]
\item If we have two competing Markov chains to sample from the same stationary density, having knowledge of their respective spectra allows for a detailed and careful comparison (see Section~\ref{sec_illus_finite} for an illustration).

\item For positive integer $m$, let $k^m(\cdot, \cdot)$ denote the $m$-step transition density of the associated Markov chain $\mcx$. The chi-square distance to stationarity after $m$  steps, starting at state $x$ is defined as:
\[
\chi^2_x(m) := \int_{\X} \frac{|k^m(x, x') - \pi(x')|^2}{\pi(x')}\:d\nu(x').
\] 
Since $K$ is assumed to be trace class (and hence Hilbert Schmidt), it follows that \citep{diaconis:khare:saloff:2008} $\chi^2_x(m) = \sum_{i=1}^\infty \lambda_i^{2m} \phi_i(x)^2$. The average or expected chi-square distance to stationarity after $m$  steps is therefore $\pi \chi^2(m) := \int_x \chi^2_m \: d\Pi = \sum_{i=1}^\infty \lambda_i^{2m}$ (since $\pi \phi_i^2 = 1$ for all $i$). Thus, having knowledge of the entire spectrum enables one to compute these average or expected chi-square distances.

\item From (\ref{K_lambda1_reln}), it is apparent that if $\eta_i$'s are known, then the knowledge of the entire spectrum enables us to compute the exact $L^2$ distance to stationarity. While finding the exact 
$\eta_ i$'s in general will be difficult, specific  examples can be found  in \citet{diaconis:khare:saloff:2008, hobert:roy:robert:2011,khare:zhou:2009}.    
\end{enumerate}

The literature for general methods to evaluate/estimate the entire spectrum (all the eigenvalues) of a Markov operator is, however, rather sparse. \citet{adamczak:bednorz:2015} provide an elegant and simple way of consistently estimating the spectrum of a general Hilbert-Schmidt integral operator with symmetric kernel using approximations based on random matrices simulated from a Markov chain. The approach in \citet{adamczak:bednorz:2015} can in particular be adapted for estimating the spectra of Markov operators. In fact, as we show in Section~\ref{sec_rma_exact}, in this context, the regularity condition needed for their method is exactly equivalent to the underlying Markov operator being trace class. 

However, in order for the approach (and the technical consistency results) in \cite{adamczak:bednorz:2015} to be applicable, the Markov transition density $k(\cdot, \cdot)$ and the stationary density $\pi(\cdot)$ are required to be available in closed form. These assumptions are not satisfied by an overwhelming majority of Markov chains arising in modern statistical applications. This is particularly true for the so-called Data Augmentation (DA) algorithm, which is a widely used technique for constructing Markov chains by introducing unobserved/latent random variables. In this context, often, (a) the transition density can only be expressed as an intractable high-dimensional integral, and/or (b) the stationary density is only available up to an unknown normalizing constant\footnote{one would typically need to evaluate a complicated high-dimensional integral to obtain this constant, which is often infeasible}, see \citet{Albert:Chib:1993, hobert:roy:robert:2011, roy:2012, polson:scott:windle:2013, choi:hobert:2013,  hobert:jung:khare:qin:2015, qin:hobert:2016,  pal:khare:hobert:2017} to name just a few. 

The main objective of this paper is to develop a random matrix approximation method to consistently estimate the spectrum of DA Markov operators for situations where (a) and/or (b) holds. In particular, we show that if the transition densities in the method of \cite{adamczak:bednorz:2015} are replaced by appropriate Monte Carlo based approximations, the spectrum of the resulting random matrix consistently estimates the spectrum of the underlying Markov operator (Theorem~\ref{thm_mcrma_consist}). More generally, we show that the method and the result can be easily adapted to situations where the stationary density is known only up to a normalizing constant (Theorem~\ref{thm_mcrma_ext_consist}). 

No regularity conditions are needed for our results if the state space $\X$, or the latent variable space $\Z$ is finite. We would like to mention that in many statistical applications with finite state spaces,  the state space can be extremely large, with millions/billions of states. The intractability of the transition density and the size of the state space often make numerical techniques for eigenvalue estimation completely infeasible. However, as we show in the context of the example in Section~\ref{sec_illus_finite}, our method can provide reasonable answers in less than 5 minutes using modern parallel processing machinery. If both the state space $\X$ and the latent variable space $\Z$ are infinite, two regularity conditions need to be verified in order to use our results. One of them requires the Markov operator to be trace class, and the other one is a variance condition; each require checking that an appropriate integral is finite. An illustration is provided in Section~\ref{sec_illus_infinite} for the Gibbs sampler of \citet*{polson:scott:windle:2013}. 

The remainder of the article is organized as follows. In Section~\ref{sec_rma_exact} we first review the approach developed by \citet{adamczak:bednorz:2015}, which is applicable  when the Markov transition densities have closed form expressions. Then we show that in the context of Markov operators, their regularity condition for consistency is equivalent to assuming that the operator is trace class.  In Section~\ref{sec_mcrma} we introduce our approach for estimating the spectrum of DA Markov operators with intractable Markov transition densities and establish weak and strong consistency of the resulting estimates under a mild regularity assumption. In Section~\ref{sec_toy_normal} we consider a toy normal-normal DA Markov chain \citep{diaconis:khare:saloff:2008}, where all the eigenvalues are known, and examine the accuracy of the eigenvalue estimates provided by our algorithm. We then compare the convergence rates of the estimated spectrum  to those of an estimated functional of interest (mean second Hermite polynomial), and also make a comparative analysis of the performances of our method to the method of \citet{qin:hobert:khare:2017} in estimating the second largest eigenvalue. We then move on to real applications. In Section~\ref{sec_illus_infinite} we illustrate our method on the Polya Gamma Markov chain of \citet{polson:scott:windle:2013}. We verify that this Markov chain satisfies the regularity condition needed for consistency and work out the first few eigenvalue estimates for the \texttt{nodal} dataset provided in the \texttt{boot} (\citet{canty:ripley:2017}) R package. In Section~\ref{sec_illus_finite} we consider a Bayesian analysis of the two component normal mixture model and examine two competing DA Markov chains proposed in \citet{hobert:roy:robert:2011} to sample from the resulting posterior distribution. We illustrate the usefulness and applicability of our method by estimating and comparing the first few eigenvalues of the two DA chains for simulated data. We end with a discussion in Section~\ref{sec_discuss}. Proofs of all theorems and lemmas introduced in this paper are provided in the Appendix.

\section{Random Matrix Approximation method of \citet{adamczak:bednorz:2015}} 
\label{sec_rma_exact}
The objective of this section is to describe the method of operator spectra estimation via random matrices, first proposed in \citet{koltchinskii:gine:2000} and then in \citet{adamczak:bednorz:2015} in the context of Markov operators. We begin this section with a brief description of the general method, and then discuss how one can  potentially use it to estimate spectra of trace class Markov operators. This discussion is followed by a short lemma that establishes an equivalence between the regularity condition used in \citet{adamczak:bednorz:2015},  and  the condition of the Markov chain being trace class.

Let $H: L^2(\pi) \rightarrow L^2(\pi)$ be a Hilbert-Schmidt integral operator (an integral operator whose eigenvalues are square summable) defined through a symmetric (in argmuents) kernel $h(\cdot, \cdot)$ as:
\begin{equation} \label{general_H_def}
(Hg)(x) = \int_{\X} h(x, x') g(x')\:d\Pi(x'),
\end{equation}
and interest lies in obtaining $\sp(H)$, the spectrum of $H$. In general, there does not exist any method of evaluating $\sp(H)$ for arbitrary $H$. However, \citet{koltchinskii:gine:2000} suggest a novel, elegant and simple approach of \emph{estimating} $\sp(H)$ via random matrices.  Let $X_0, \cdots, X_{m-1}$ denote an IID sample of size $m$  $(\geq 1)$ from the distribution $\Pi$. Then the authors show that a (strong) consistent estimator of $\sp(H)$ is given by the set of eigenvalues of the random matrix
\[
H_{m} = \frac{1}{m} \left((1-\delta_{jj'}) h(X_j, X_{j'})\right)_{0\leq j,j' \leq m-1}
\]
for large $m$, where $\delta_{jj'}$ denotes the Dirac delta function. The strength of the result lies in the fact that it works for \emph{any} Hilbert Schmidt operator, irrespective of the dimension and structure of $\X$, as long as an IID sample from $\Pi$ can be drawn. Unfortunately, IID simulations are not always feasible, especially in high dimensional settings (otherwise there would be no need for MCMC!), thus limiting the applicability of the method. \citet{adamczak:bednorz:2015} generalize \citet{koltchinskii:gine:2000}'s result by allowing $X_0, \cdots, X_{m-1}$ to be an MCMC sample (i.e., realizations of  a Markov chain with equilibrium distribution $\Pi$), and prove consistency of the resulting estimates. 

Let $K$ denote a positive self-adjoint trace class Markov operator as defined in (\ref{K_defn}). Of course $K$ is Hilbert Schmidt (eigenvalues are summable \emph{implies} they are square summable), and $h(x, x') = k(x, x')/\pi(x')$ is symmetric in its argument due to reversibility of the associated Markov chain. Thus  by expressing $K$ in the form (\ref{general_H_def})  $\sp(K)$ can potentially be  estimated by \citet{adamczak:bednorz:2015}'s method, which only requires an MCMC sample from $\Pi$. The resulting method, which uses the same random data generated during the original run of the Markov chain in the recipe proposed in \citet{adamczak:bednorz:2015} to estimate the spectrum, will be called the \emph{Random Matrix Approximation (RMA)} method henceforth, and is described below.

\begin{algorithm}[h]
	
	\begin{enumerate}[label = Step \arabic*:, start = 0, leftmargin=1.5cm]
		\item Given a starting point $X_0$, draw realizations $X_1, X_2, \dots, X_{m-1}$ from the Markov chain $\mcx$ associated with $K$.
		
		\item Given $X_0, \dots, X_{m-1}$, for each pair $(j, j')$ with $0 \leq j , j' \leq m-1$, compute the Markov transition densitys $k(X_j, X_{j'})$ and the kernels $h(X_j, X_{j'}) = k(X_j, X_{j'})/\pi(X_{j'})$, and construct the matrix
		\begin{equation} \label{Hmn}
		H_{m} = \frac{1}{m} \left((1-\delta_{jj'})\: h(X_j, X_{j'})\right)_{0 \leq j,  j' \leq m-1} 
		\end{equation}
		where $\delta_{jj'} = \one(j = j')$ is the Dirac delta function.
		
		\item Calculate the eigenvalues $\hat \lambda_0 \geq  \hat \lambda_1 \geq  \cdots \hat \lambda_{m-1}$ of $H_m$ and estimate $\sp(K)$ by $\hat \sp(K) = \sp(H_m):= \{\hat \lambda_0, \hat \lambda_1, \dots, \hat \lambda_{m-1} \}$.
	\end{enumerate}	
	\caption{Random Matrix Approximation (RMA) method of estimating $\sp(K)$ for a Markov operator $K$ with  Markov transition density $k(\cdot, \cdot)$ and stationary density $\pi(\cdot)$ available in closed form.}
	\label{algo_erma}
\end{algorithm}

Sacrificing independence and identicalness of the random sample in \citet{adamczak:bednorz:2015}'s RMA method, however, comes at a price (as compared to \cite{koltchinskii:gine:2000}'s method, which uses IID samples). In particular, to ensure strong consistency in \citet{adamczak:bednorz:2015}'s method, an additional regularity condition is required to be satisfied by the Markov operator $K$, namely, a $L^2(\pi)$ function $F : \X \rightarrow \R$ needs to exist for which $|h(x,x')|  \leq F(x)F(x')$ for all $x, x' \in X$.  Interestingly, as we show in the following lemma (Lemma~\ref{lemma_eqv}), this condition for $K$ is equivalent to that of $K$ being trace-class in the current setting. The proof of Lemma~\ref{lemma_eqv} is provided in Section~\ref{supp_sec_rma_exact} of the Appendix. At the core of the proof, the following two alternative characterizations of trace class and Hilbert Schmidt operators (see, e.g., \citet{jorgens:1982}) are used. The operator $K$ as defined in (\ref{K_defn}) is trace class if and only if
\begin{equation} \label{tr_cond}
\int_{\X} k(x, x) \: d\nu(x) = \int_{\X} \frac{k(x, x)}{\pi(x)} \: d\Pi(x) < \infty 
\end{equation}
whereas it is Hilbert Schmidt if and only if  
\begin{equation} \label{hs_cond}
\int_{\X} \int_{\X} k(x, x')^2\: \frac{\pi(x)}{\pi(x')} \: d\nu(x) \:d\nu(x') =  \int_{\X} \int_{\X} \: \left[\frac{k(x, x')}{\pi(x')}\right]^2  \: d\Pi(x) \:d\Pi(x') < \infty. 
\end{equation}

\begin{lemma} \label{lemma_eqv}	
	Consider a reversible Markov operator $K$ as defined in (\ref{K_defn}). Define $h(x, x') = k(x, x')/\pi(x')$ for  $x, x' \in \X$. Then the following two conditions are equivalent:
	\begin{enumerate}[label = (\roman*)]
		\item there exists $F : \X \rightarrow \R$ such that $\pi F^2 < \infty$ and $|h(x,x')|  \leq F(x)F(x')$ for all $x, x' \in X$.
		
		\item $K$ is trace class.
	\end{enumerate}
	
\end{lemma}

As a consequence of Lemma~\ref{lemma_eqv}, we are now in a position to adapt the consistency result from \citet{adamczak:bednorz:2015} for the RMA method described above in Algorithm~\ref{algo_erma}. Before stating the result, we introduce required notations from \citet{koltchinskii:gine:2000} and \citet{adamczak:bednorz:2015}. Recall that for any operator $A$ (finite or infinite) we use the notation $\sp(A)$ to denote its spectrum. Thus, for a finite matrix $A$, $\sp(A)$ will denote the set of its eigenvalues.  Since the Markov operators we consider are trace class (and therefore, Hilbert Schmidt), their spectra can be identified with the sequences $(\lambda_m)_{m = 0}^\infty \in \ell_2$ of eigenvalues, where $\ell_2$ is the Hilbert space of all square summable real sequences. Because our goal is to approximate the (possibly infinite) spectrum of an integral operator by the finite spectrum of a matrix, we will identify the latter with an element of $\ell_2$, by appending an infinite sequence of zeros to it. 
As in \citet{koltchinskii:gine:2000}, the metric we use for comparing spectra is the $\delta_2$ metric, which is defined for $x, y \in \ell_2$ as 
\begin{equation} \label{delta2_defn}	
\delta_2(x,y) = \inf_{\zeta \in \mathcal{P}} \left[\sum_{m=0}^\infty (x_m - y_{\zeta(m)})^2\right]^{1/2}
\end{equation}
where $\mathcal{P}$ is the set of all permutations of natural numbers.  Note that for any two points on $\ell_2$, the above metric can be expressed as an $\ell_2$ distance of the sorted versions of the two points, as explained below. Following \citet{koltchinskii:gine:2000}, for any $x = (x_m)_{m=0}^\infty \in \ell_2$, we set $x = x_{+} + x_{-}$, where $x_{+} = (\max\{x_m, 0\})_{m=0}^\infty$ and $x_{-} = x - x_{+}$. We denote by $x_{+}^\downarrow$ ($x_{-}^\uparrow$) the point in $\ell_2$ with the same coordinates of $x_{+}$ ($x_{-}$), but arranged in non-increasing (non-decreasing) order, and set $x^{\uparrow \downarrow} = x_{-}^{\uparrow} \oplus x_{+}^{\downarrow}$, where  $u \oplus v = (u_0, \cdots, u_m, \cdots, v_0, \cdots, v_m, \cdots) \in \ell_2$ for $u = (u_m)_{m=0}^\infty$ and  $v = (v_m)_{m=0}^\infty \in \ell_2$. Then
\begin{equation} \label{delta2_l2}
\delta_2(x,y) = \left\| x^{\uparrow\downarrow} - y^{\uparrow\downarrow} \right\|_{\ell_2}.
\end{equation}

\noindent From the Hoffman–Wielandt inequality (\citet{hoffman:wielandt:1953}, \citet[Theorem 2.2]{koltchinskii:gine:2000}), it follows that for normal operators $A$ and $B$, 
\begin{equation} \label{hoff_wiel_ineq}
\delta_2(\sp(A),\sp(B)) \leq \|A-B\|_\hs.
\end{equation}
where $\|A\|_\hs$ denotes the Hilbert Schmidt norm of an operator $A \in L^2(\pi)$ defined by 
\[
\|A\|_\hs =  \left(\sum_{m=0}^\infty \|A\varphi_m\|^2\right)^{1/2},
\] 
for any orthonormal basis $(\varphi_m)_{m=0}^\infty$ of $L^2(\pi)$. Note that if $A$ is finite (i.e., a matrix), say $A = (a_{ij})$, then  $\|A\|_\hs = \|A\|_\fb$ where $\|A\|_\fb$ denotes the Frobenious norm of $A$ defined as
\[
\|A\|_\fb = \left(\sum_i \sum_j a_{ij}^2 \right)^{1/2}.
\]

The following theorem, a rephrasing of Theorem 2.1 from \citet{adamczak:bednorz:2015} adapted to the current setting using Lemma~\ref{lemma_eqv},  establishes (strong) consistency of the spectrum estimator obtained by RMA method for a positive self-adjoint trace class Markov operator.

\begin{theorem} \label{thm_rma_consis}
	Let $\mcx$ = $(X_n)_{n \geq 0}$ be a  reversible Markov chain with Markov transition density $k(\cdot, \cdot)$, invariant measure $\Pi$, and suppose the associated Markov operator $K$ as given in (\ref{K_defn}) is positive and trace class. Let  $\Phi_m = \{X_0, \dots, X_{m-1}\}$ denote the first $m$ realizations of the Markov chain, and given $\Phi_m$, construct the matrix $H_m$ as given in (\ref{Hmn}).
	Then, for every initial measure $\nu_0$ for the chain $\mcx$, with probability one, as $m \rightarrow \infty$,
	\[
	\delta_2 (\sp(H_m),\sp(K)) \rightarrow 0.
	\]
\end{theorem}

\section{A Novel Monte Carlo Based Random Matrix Approximation Method for DA Markov Chains} \label{sec_mcrma}

As we see in Section~\ref{sec_rma_exact}, the RMA method of \citet{adamczak:bednorz:2015} requires evaluation of the ratio $k(X_j, X_{j'})/\pi(X_{j'})$ for every pair $(j, j')$. Unfortunately, as mentioned in the introduction, one or both of 
$k(\cdot, \cdot)$ and $\pi(\cdot)$ are often intractable and do not have closed form expressions in many statistical applications. This is particularly true in the context of the Data Augmentation (DA)  algorithm, where along with the variable of interest $X$, one introduces a latent variable $Z$ such that generations from the conditional distributions of $X \mid Z$ ($X$ given $Z$) and $Z \mid X$ are possible. Then, given a starting point $X_0$, at each iteration $m \geq 1$, one first simulates $z$ from the distribution of $Z \mid X=X_{m-1}$ and then generates $X_{m}$, from the distribution of $X \mid Z = z$. The $X_m$'s generated in this method are retained and used as the required MCMC sample. Hence, the Markov transition density can be written as 
\begin{equation} \label{mtd_da_general}
k(x, x') = \int_{\Z} f_{X \mid Z}(x' \mid z) f_{Z \mid X}(z \mid x)\:d\zeta(z).
\end{equation} 
Here $f_{Z \mid X}$ and  $f_{X \mid Z}$ are conditional densities with respect to the measures $\xi$ and $\nu$ respectively, and are simple and easy to sample from in a typical DA algorithm. A DA Markov chain is necessarily reversible, which means the associated Markov operator is self-adjoint. The operator is also positive with a positive spectrum, as shown in \citet{liu:wong:kong:1994}.

However, the integral in (\ref{mtd_da_general}) providing  the Markov transition density of a DA Markov chain often does not have a closed form expression, and cannot be efficiently approximated via deterministic numerical integration (usually due to high dimensionality). Intractability of the integral precludes applicability of \citet{adamczak:bednorz:2015}'s RMA method of estimating spectrum  (Algorithm~\ref{algo_erma}) in such cases. In this section we propose a Monte Carlo based random matrix approximation (MCRMA)  algorithm to estimate the spectrum of DA algorithms with intractable transition densities (Algorithm~\ref{algo_mcrma}). To contrast with MCRMA, we shall call the RMA method of \citet{adamczak:bednorz:2015}  (Algorithm~\ref{algo_erma})  the exact RMA. Consistency of MCRMA spectrum estimates is established in Theorem~\ref{thm_mcrma_consist}. 

Often, in addition to the intractability of the transition density, the stationary density is also available only up to an unknown normalizing constant (which is again hard to estimate in many modern applications as the stationary density is supported on a high-dimensional space). We adapt our algorithm to this situation (Algorithm~\ref{algo_mcrma_ext}), and establish consistency of the resulting estimates as well (Theorem~\ref{thm_mcrma_ext_consist}).

\subsection{Monte Carlo Random Matrix Approximation (MCRMA) Method} \label{sec_mcrma_w_pi}

In this section, we will present a method to estimate the spectrum of a DA Markov operator where the transition density in (\ref{mtd_da_general}) is intractable, but the staionary density $\pi$ is available in closed form. Given  $m$ realizations $\Phi_m = \{X_0, X_1, \dots, X_{m-1} \}$ of the positive trace class reversible  Markov chain $\mcx$  with transition density in the form (\ref{mtd_da_general}), the key idea is to approximate $k(X_j, X_{j'})$ for each pair $(j,j')$ using  classical Monte Carlo technique, and then construct an analogue of the RMA estimator that uses the  approximate kernels instead of the original. The details of the method are provided in Algorithm~\ref{algo_mcrma} below. 

\begin{algorithm}[h]
	\DontPrintSemicolon 
	\begin{enumerate}[label = Step \arabic*:, start = 0]
		\item Given a starting point $X_0$, draw realizations $X_1, X_2, \dots, X_{m-1}$ from the associated Markov chain $\mcx$. Call $\Phi_m = \{X_0, \dots, X_{m-1}\}$.
		
		\item Given $\Phi_m$, for each $j=0,1,\cdots,m-2$, generate generate $N = N(m)$ IID observations $Z_1^{(j)}, \dots, Z_N^{(j)}$ from the density $f_{Z \mid X}(\cdot\mid X_j)$.
		
		\item For each pair $(j, j')$ with $0 \leq j < j' \leq m-1$, construct the Monte Carlo estimate 
		\[
		\hat{k}_N(X_j, X_{j'}) = \frac{1}{N} \sum_{l=1}^N f_{X \mid Z} \left(X_{j'}|Z_l^{(j)}\right),
		\]
		define the estimated kernel 
		\[
		\hat h_N(X_j, X_{j'}) = 
		\begin{cases}
		\frac{\hat k_N \left( X_{j}, X_{j'} \right)}{ \pi(X_{j'})} & \text{ if } j < j' \\
		0 & \text{ if } j = j'\\
		\hat h_N(X_{j'}, X_{j}) & \text{ if } j >j' 	
		\end{cases},
		\] 
		and construct the matrix 
		\begin{equation} \label{Hmnhat}
		\Hmn = \frac{1}{m} \left((1-\delta_{jj'})\: \hat h_N(X_j, X_{j'})\right)_{0 \leq j,  j' \leq m-1}. 
		\end{equation}
		where $\delta_{jj'} = \one(j = j')$ is the Dirac delta function. 
		Observe that $\Hmn$ is  symmetric by construction, with zero diagonal entries. 
		
		\item Calculate the eigenvalues $\hat \lambda_0 \geq  \hat \lambda_1 \geq  \cdots \geq \hat \lambda_{m-1}$ of $\Hmn$ and estimate $\sp(K)$ by $\hat \sp(K) = \sp(\Hmn):= \{\hat \lambda_0, \hat \lambda_1, \dots,  \hat \lambda_{m-1} \}$.
	\end{enumerate}
	
	\caption{Monte Carlo Random Matrix Approximation (MCRMA) method of estimating $\sp(K)$ for a positive reversible Markov operator $K$ with associated Markov transition density $k(x, x') = \int_{\Z} f_{X \mid Z}(x'\mid z) f_{Z \mid X}(z \mid x)\:d\zeta(z)$}
	\label{algo_mcrma}
\end{algorithm}

It is to be noted that for Step 1 in the MCRMA algorithm to be feasible, the density $f_{Z \mid X}$ should be easy to sample from. This is typically true for DA algorithms that are used in practice. In fact, the major motivation for using a DA algorithm is that the conditional densities $f_{X \mid Z}$ and $f_{Z \mid X}$ are easy to sample from, whereas it is hard to directly generate samples from $\pi$. For Step 2 to be feasible, we need $f_{X \mid Z}$ to be available in closed form. Again, this is true in most statistical applications, where $f_{X \mid Z}$ is typically a standard density such as multivariate normal, gamma etc. Another crucial thing to note, from a computational point of view, is that the rows of the matrix $\Hmn$ can be constructed in an embarrassingly parallel fashion (since no relationship is assumed among the elements of $\Hmn$), thereby reducing the running time of the algorithm significantly.

Note that the MCRMA algorithm (Algorithm~\ref{algo_mcrma}) provides a coarser approximation to the spectrum of $K$ as compared to the the exact RMA algorithm  (Algorithm~\ref{algo_erma}). This is because we use the additional Monte Carlo based approximation $\hat k_N$ (for $k$) in the constructed random matrices. An obvious and important question is: does an analog of the consistency result for the RMA algorithm (Theorem~\ref{thm_rma_consis}) hold in this more complex setting of the MCRMA algorithm? We state a consistency result below, and the proof is provided in Section~\ref{supp_sec_mcrma} of the Appendix.

\begin{theorem}\label{thm_mcrma_consist}
	Let $\mcx$ = $(X_n)_{m \geq 0}$ be a positive, reversible Markov chain with transition density $k(\cdot, \cdot)$ in the form (\ref{mtd_da_general}), invariant measure $\Pi$ and associated Markov operator $K$ as given in (\ref{K_defn}). Let  $\Phi_m = \{X_0, \dots, X_{m-1}\}$ denote the first $m$ realizations of the Markov chain, and given $\Phi_m$, construct the matrix $\Hmn$ as given in (\ref{Hmnhat}). Then the following  hold:
	\begin{enumerate}[label=(\Roman*)]
		\item \label{case_finite} If $\X$ is finite, then 	(strong consistency) for every initial measure $\nu_0$ of the chain $\mcx$, as $m \rightarrow \infty$ and $N \rightarrow \infty$, 
		\[
		\delta_2\left(\sp\left(\Hmn\right), \sp(K)\right) \rightarrow 0 \text{ almost surely.}
		\]
		
		\item \label{case_infinite} If $\X$ is infinite (countable or uncountable), and
		\begin{enumerate}[label=(\Alph*)]
			\item \label{trace_condn} $K$ is trace class, and
			
			\item \label{var_condn}  (variance condition)
			\begin{align*}
			\sup_{m \geq 1} \: \max_{0 \leq j < j' \leq m-1} \int_{\X} \int_{\X} \int_{\Z}  & \left( \frac{ f_{X \mid Z}\left( x_{j'}\mid z \right) }{\pi(x_{j'})} \right)^2  f_{Z \mid X}\left(z \mid x_{j} \right) \\ & \qquad q_{jj'}(x_j, x_{j'}) \:d\zeta(z) \:d\nu(x_j)\: d\nu(x_{j'}) < \infty.
			\end{align*}
			where $q_{j_1 j_2 \cdots j_k}$ denotes the joint density of $X_{j_1}, \dots, X_{j_k}$, $0 \leq j_1 < \cdots < j_k \leq m-1$, $1 \leq k \leq m$,
		\end{enumerate}
		then 
		\begin{enumerate}[label = (\roman*)]
			\item \label{erma_consist_cond_i} (weak consistency) if $\frac{1}{N(m)} \rightarrow 0$ as $m \rightarrow \infty$, then $\delta_2\left(\sp\left(\Hmn\right), \sp(K)\right) \xrightarrow{P} 0$,
			\item \label{erma_consist_cond_ii} (strong consistency) if $\sum_{m=0}^\infty \frac{1}{N(m)} < \infty$, then $\delta_2\left(\sp\left(\Hmn\right), \sp(K)\right) \rightarrow 0$ almost surely.
		\end{enumerate}
	\end{enumerate}
\end{theorem}

\begin{remark}\label{rem_z_finite}
	Let $K^*$ denote the Markov operator associated with the $Z$ chain (the Markov chain of the generated latent data), defined for $f \in L^2(\pi^*)$ as
	\begin{equation} \label{K*_defn}
	(K^*f)(z) = \int_{\Z} \: k^*(z, z')  \: f(z') \:d\zeta(z') = \int_{\X} \: \frac{k^*(z, z')}{\pi^*(z')} \: f(z') \:d\Pi^*(z')
	\end{equation}
	where $k^*(\cdot, \cdot)$ denotes the Markov transition density of the $Z$ chain, $\pi^*$ denotes the stationary density for $Z$ (associated with $k^*$) and $\Pi^*$ is the probability measure associated with $\pi^*$. Then  \citet{khare:hobert:2011} show that $\sp(K) = \sp(K^*)$, which implies that instead of estimating $\sp(K)$, one can equivalently estimate $\sp(K^*)$. Note that,  
	\[
	k^*(z, z') = \int_{\X} f_{Z \mid X}(z'|x) f_{X \mid Z}(x\mid z) \:d\nu(x)
	\]  
	with $f_{X \mid Z}$ and $f_{Z \mid X}$ being the same conditional densities as before. Therefore, an analogous MCRMA algorithm for estimating $\sp(K^*)$ can be similarly formulated. Here, given the realizations $Z_0$,  $Z_1$,  $\cdots, Z_{m-1}$, one first finds the Monte Carlo approximates of $k^*(z, z')$ (via IID samples generated from $f_{X \mid Z}(\cdot\mid z)$) at every paired realization $(Z_j, Z_{j'})$, then defines a random matrix with the ratios $k^*(z,z')/\pi^*(z')$ (times an adjustment factor $1/m$), similar to the MCRMA with $X$ observations, and finally evaluates eigenvalues of the resulting random matrix. Consequently, an analogous consistency theorem will also hold for the resulting algorithm.
	
	Because the $Z$ chain is automatically generated as a by-product of the DA algorithm, from a practitioner's point of view, using $Z$ instead of $X$ makes little difference in MCRMA.  However, substantial simplifications on the regularity conditions may be achieved by using $Z$. This is particularly true in cases where the latent variable space $\Z$ is finite (however large). In such cases, no regularity condition is required to be satisfied (case~\ref{case_finite} in Theorem~\ref{thm_mcrma_consist}) to achieve strong consistency. See Section~\ref{sec_illus_finite} for examples.
\end{remark}

\begin{remark} \label{rem_var_condn}
The variance condition \ref{var_condn} is more restrictive than the trace class condition \ref{trace_condn} because of the square term $\left\{ {f_{X \mid Z}( x_{j'}\mid z) }/{\pi(x_{j'})} \right\}^2$. These types of second moment conditions are often necessary to guarantee good behavior  of eigenvalue estimators; a somewhat similar second moment condition appears in \citet[equation 14 and Theorem 2]{qin:hobert:khare:2017} to ensure finite variance of their second largest eigenvalue estimator. Proofs of Theorems~\ref{tr_thm_psw} and \ref{var_thm_psw}  in Section~\ref{supp_sec_illus} of the Appendix provide illustrations on how the integrals in conditions \ref{trace_condn} and \ref{var_condn} can be handled.
\end{remark}

\begin{remark} \label{rem_thm_mcrma_consis}
	When $\X$ is finite, strong consistency is guaranteed as long as $m \rightarrow \infty$ and $N \rightarrow \infty$ (no relationship between the rate of growth of $m$ and $N$ is necessary).  When $\X$ is infinite and the conditions \ref{trace_condn} and \ref{var_condn} hold, the conditions \ref{case_infinite}\ref{erma_consist_cond_i} and \ref{erma_consist_cond_ii} on in Theorem~\ref{thm_mcrma_consist} are required to justify weak and strong consistency respectively. These  conditions on $N$ and $m$, are however, not very demanding.  For example,  when $N(m) = O(m)$ or even, $N(m) = O(\log m)$, \ref{case_infinite}\ref{erma_consist_cond_i} is satisfied, and weak convergence holds. On the other hand  when $N(m) = O(m^{1+\delta})$, for some $\delta > 0$, condition~\ref{erma_consist_cond_ii} is satisfied, ensuring strong convergence. In practice, as long as both $N$ and $m$ are sufficiently large, reasonable results can be expected.
\end{remark}

\subsection{MCRMA  with $\pi$  Specified Only up to a Constant}
Note that Step 2 of MCRMA method requires construction of a symmetric matrix whose $(j,j')$th entry has $\pi(X_{j'})$ in the denominator. This is clearly not feasible in cases where $\pi$ is known up to a constant, i.e., $\pi$ is of the form $\pi(\cdot) = \eta(\cdot)/c$, where $c \in (0, \infty)$ is an unknown constant, and the functional form of $\eta(\cdot)$ is completely known.  In this section, we propose a simple strategy that adapts  Algorithm~\ref{algo_mcrma} for such cases. The basic idea, displayed formally in Algorithm~\ref{algo_mcrma_ext}, is to follow the steps of Algorithm~\ref{algo_mcrma} but now with $\eta(\cdot)$ in the denominator of the random matrix instead of $\pi(\cdot)$, and then simply rescale the eigenvalues so that the largest eigenvalue is 1. Clearly, this nullifies any estimation/evaluation of the normalizing constant. Theorem~\ref{thm_mcrma_ext_consist} establishes consistency for the resulting estimator by exploiting the fact that the largest eigenvalue of any Markov operator is 1. A detailed proof is given in the Appendix (Section~\ref{supp_sec_mcrma}).

\begin{algorithm}[h]
	
	\begin{enumerate}[label = Step \arabic*:, start = 0]
		\item Given a starting point $X_0$, draw realizations $X_1, X_2, \dots, X_{m-1}$ from the associated Markov chain $\mcx$. Call $\Phi_m = \{X_0, \dots, X_{m-1}\}$.
		
		\item Given $\Phi_m$, for each $j=0,1,\cdots,m-1$, generate generate $N =N(m)$ IID observations $Z_1^{(j)}, \dots, Z_N^{(j)}$ from the density $f_2(\cdot\mid X_j)$.
		
		\item For each pair $(j, j')$ with $0 \leq j < j' \leq m$, construct the Monte Carlo estimate 
		\[
		\hat{k}_N(X_j, X_{j'}) = \frac{1}{N} \sum_{l=1}^N f_1 \left(X_{j'}|Z_l^{(j)}\right),
		\]
		define the estimated kernel 
		\[
		\hat s_N(X_j, X_{j'}) = 
		\begin{cases}
		\frac{\hat k_N \left( X_{j}, X_{j'} \right)}{ \eta(X_{j'})} & \text{ if } j < j' \\
		0 & \text{ if } j = j'\\
		\hat s_N(X_{j'}, X_{j}) & \text{ if } j >j' 	
		\end{cases},
		\] 
		and construct the matrix  
		\begin{equation} \label{Smnhat}
		\Smn = \frac{1}{m+1} \left((1-\delta_{jj'}) \:\hat s_N(X_j, X_{j'})\right)_{0 \leq j,  j' \leq m-1} 
		\end{equation}
		
		\item Calculate the eigenvalues $\hat \kappa_0 \geq \hat \kappa_1 \geq \cdots \geq \hat \kappa_{m-1}$ of $\Smn$, and estimate $\sp(K)$ by $\sp\left( \Smn \right)/  \spmax\left( \Smn \right) := \{1, \hat \kappa_1/\hat \kappa_0, \dots, \hat \kappa_{m-1} / \hat \kappa_0 \}$, where $\spmax\left( \Smn \right) = \hat \kappa_0$ is the largest eigenvalue of $\Smn$.	
		
	\end{enumerate}
	
	\caption{MCRMA estimation of $\sp(K)$ for a trace class Markov operator $K$ when $\pi(\cdot) \propto \eta(\cdot)$, and the functional form for $\eta(\cdot)$ is known} 	
	\label{algo_mcrma_ext}    
\end{algorithm}

\begin{theorem} \label{thm_mcrma_ext_consist}
	Let $\mcx$ = $(X_n)_{m \geq 0}$ be a positive, reversible Markov chain with transition density $k(\cdot, \cdot)$ in the form (\ref{mtd_da_general}), invariant measure $\Pi$ and associated Markov operator $K$ as given in (\ref{K_defn}). Further, suppose that $\pi(\cdot) = \eta(\cdot)/c$, where $c \in (0, \infty)$ is a possibly unknown normalizing constant, and the functional form for $\eta(\cdot)$ is known. Let  $\Phi_m = \{X_0, \dots, X_{m-1}\}$ denote the first $m$ realizations of the  chain. Given $\Phi_m$, construct the matrix $\Smn$ as given in (\ref{Smnhat}). Then 
	
	\begin{enumerate}[label=(\Roman*)]
		\item \label{case_finite_ext} if $\X$ is finite, then 	(strong consistency) for every initial measure $\nu_0$ for the chain $\mcx$, as $m \rightarrow \infty$ and $N \rightarrow \infty$, 
		\[
		\textstyle \delta_2\left(\frac{\sp\left(\Smn\right)}{\spmax \left(\Smn\right)}, \sp(K)\right) \rightarrow 0 \text{ almost surely.}
		\]
		
		\item \label{case_infinite_ext} if $\X$ is infinite (countable or uncountable), and condition \ref{trace_condn} and \ref{var_condn} in Theorem~\ref{thm_mcrma_consist} hold, then 
		\begin{enumerate}[label = (\roman*)]
			\item \label{mcrma_ext_consist_cond_i} (weak consistency) if $\frac{1}{N(m)} \rightarrow 0$ as $m \rightarrow \infty$, then $\delta_2\left(\frac{\sp\left(\Smn\right)}{\spmax \left(\Smn\right)}, \sp(K)\right) \xrightarrow{P} 0$,
			\item \label{mcrma_ext_consist_cond_ii} (strong consistency) if $\sum_{m=0}^\infty \frac{1}{N(m)} < \infty$, then $\delta_2\left(\frac{\sp\left(\Smn\right)}{\spmax \left(\Smn\right)}, \sp(K)\right) \rightarrow 0$ almost surely.
		\end{enumerate}
	\end{enumerate}

\end{theorem}

\begin{remark}
The quantity $\hat\kappa_0 = \max \sp (\Smn)$, obtained as a by-product of Algorithm~\ref{algo_mcrma_ext}, is in fact a consistent estimator of the normalizing constant $1/c$ (see the proof of Theorem~\ref{thm_mcrma_ext_consist} in Section~\ref{supp_sec_mcrma} of the Appendix). This estimator is implicitly used during spectrum estimation in Algorithm~\ref{algo_mcrma_ext}, and nullifies the need of any separate estimation of the normalizing constant. It is to be noted that estimation of the constant $1/c$ is an interesting problem on its own, and appears in many statistical and machine learning problems; one notable example being marginal likelihood estimation in Bayesian model selection. Clearly, $\hat \kappa_0$ can be used for estimating $1/c$ outside the context of eigenvalue estimation, where the only goal is to estimate the normalizing constant; consistency of the estimator is however guaranteed only when the assumptions in Theorem~\ref{thm_mcrma_ext_consist} are met. Comparative assessment of the estimator $\hat \kappa_0$ with other estimators of the normalizing constant, such as the bridge sampling estimator \citep{bennett:1976, meng:1996}, is a topic of further research.
\end{remark}

\section{Illustrations} \label{sec_illus}
The purpose of this section is to illustrate the applicability and usefulness of the MCRMA algorithm in practical settings. We shall consider two separate examples, one with a finite, and one with an infinite state space. However, before proceeding to these applications, we will start with a toy normal-normal DA Markov chain to understand/illustrate the performance of the MCRMA algorithm in a setting where the entire spectrum is already known. All computations in this section are done in R \citep{R} with some parts written in C++ to speed up computation. We used the R packages \texttt{Rcpp} \citep{eddelbuettel:2013} to call C++ functions inside R, and \texttt{ggplot2} \citep{wickham:2016} and \texttt{reshape2} \citep{wickham:2007} to create the plots.

\subsection{Toy Example: The Normal-Normal DA Chain} \label{sec_toy_normal}

In this section we consider a toy normal-normal DA Markov chain considered in \cite{diaconis:khare:saloff:2008} and then in \cite{qin:hobert:khare:2017} with known eigenvalues to illustrate the performance of the MCRMA method. Here, given a starting point $x_0$, one iterates between the following two steps:
\begin{enumerate}[label=(\roman*)]
	\item generate $z$ from $\N(x/2, 1/8)$,
	\item generate $x$ from $\N(z, 1/4)$.
\end{enumerate}
Of course, the stationary density of $x$ is just $\N(0, 1/2)$, and there is no practical need for this MCMC algorithm. However, the spectrum of the corresponding Markov operator $K$ has been studied thoroughly in \citet{diaconis:khare:saloff:2008} and therefore it can be used as a nice toy example to exhibit the performance of MCRMA. It is easy to see that both the trace class condition \ref{trace_condn} and the variance condition \ref{var_condn} hold for the operator $K$ (since all the full conditional densities are just normal densities). From \citet{diaconis:khare:saloff:2008} it follows that the eigenvalues of $K$ are given by $(\lambda_n)_{n = 0}^\infty$ with $\lambda_n = 1/2^n$.

\begin{figure}
	\centering	
	\subfloat[The largest 11 eigenvalues. There are 10 curves, each corresponding to the choices $m = 1000, \cdots, 10,000$ in the MCRMA algorithm. The true eigenvalues are shown as red dots.]{\includegraphics[height = 0.25\textheight]{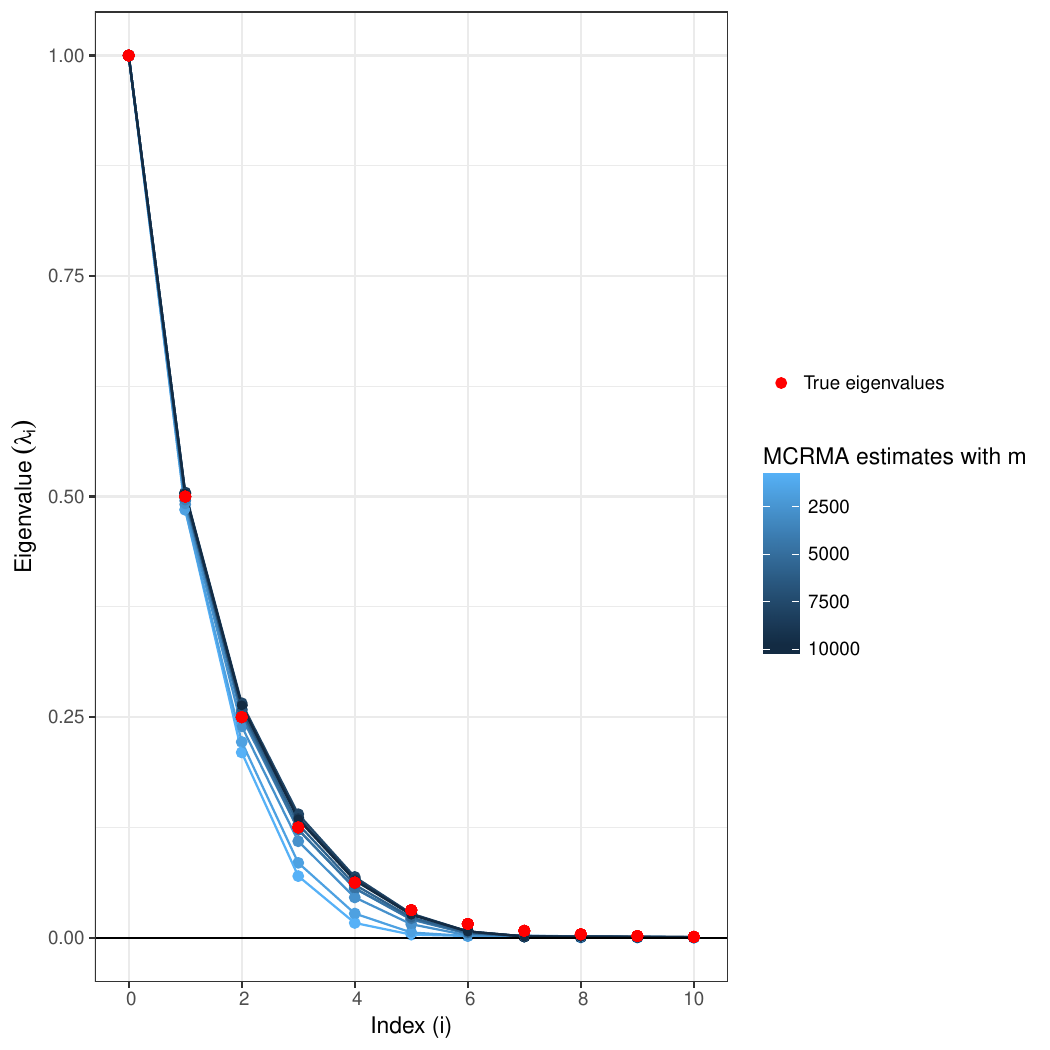}\label{toy_normal_all}} 
	\qquad 
	\subfloat[Second largest eigenvalue as a function of iterations $m$ in the MCRMA algorithm. The true eigenvalue is shown as a horizontal line.]{\includegraphics[height = 0.25\textheight]{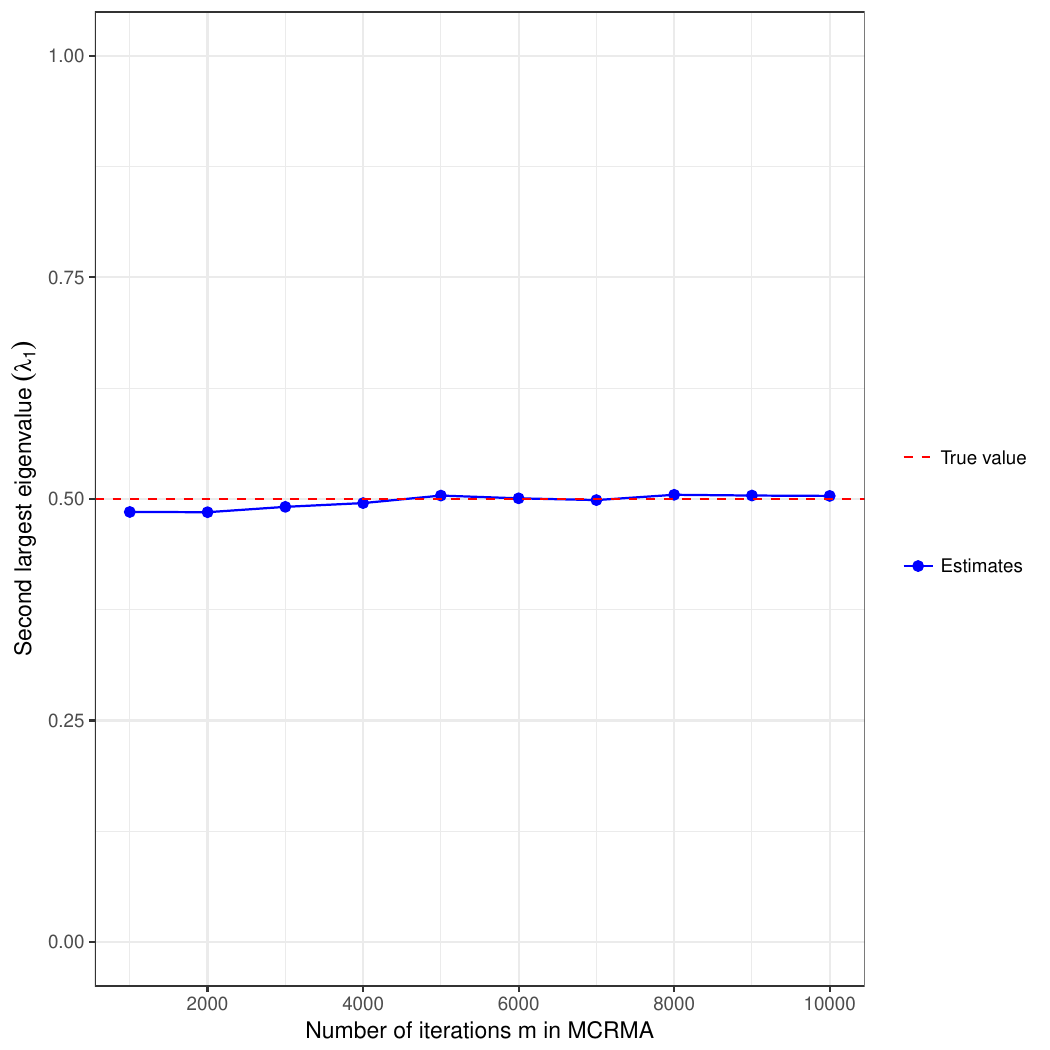} \label{toy_normal_lambda_1}} \\
	\subfloat[Third largest eigenvalue as a function of iterations $m$ in the MCRMA algorithm. The true eigenvalue is shown as a horizontal line.]{\includegraphics[height = 0.25\textheight]{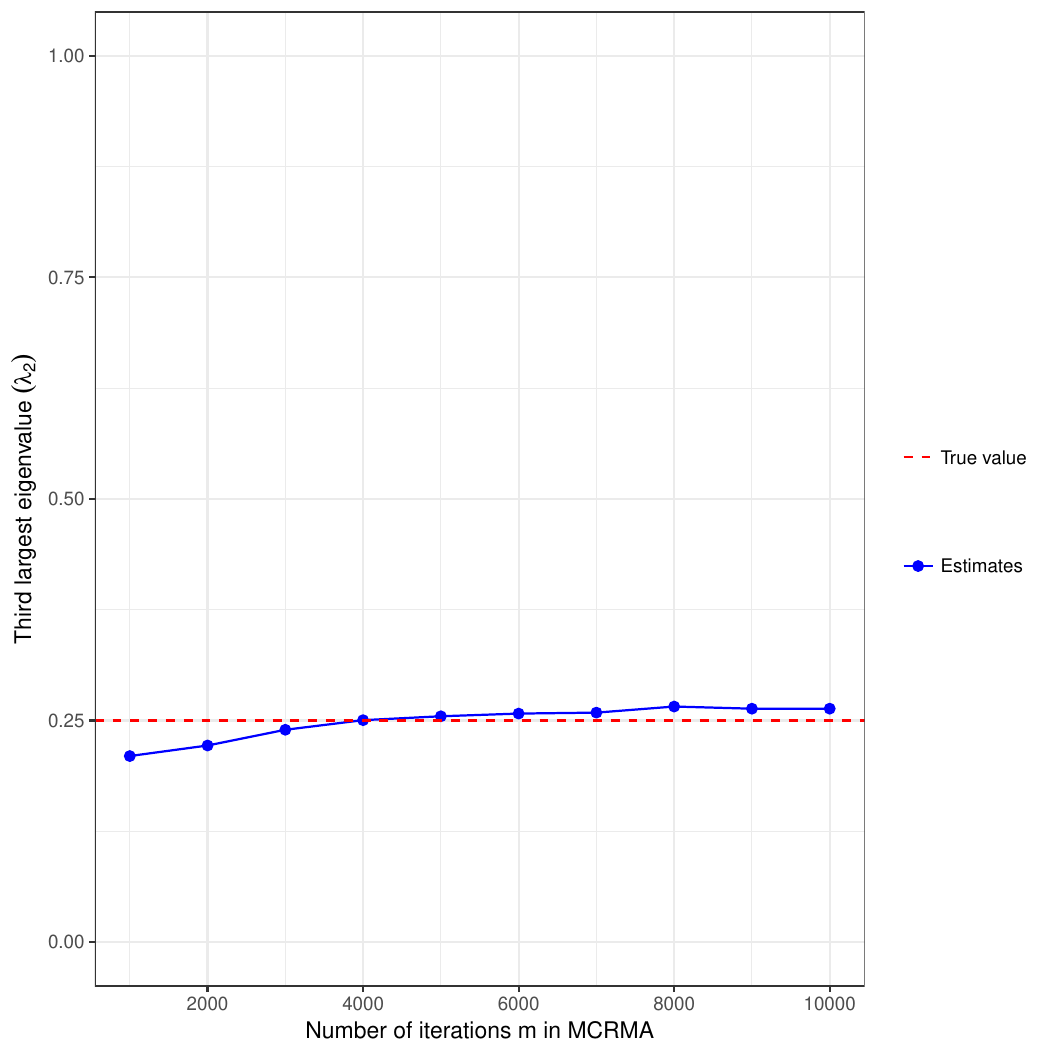} \label{toy_normal_lambda_2}}
	\qquad
	\subfloat[Fourth largest eigenvalue as a function of iterations $m$ in the MCRMA algorithm. The true eigenvalue is shown as a horizontal line.]{\includegraphics[height = 0.25\textheight]{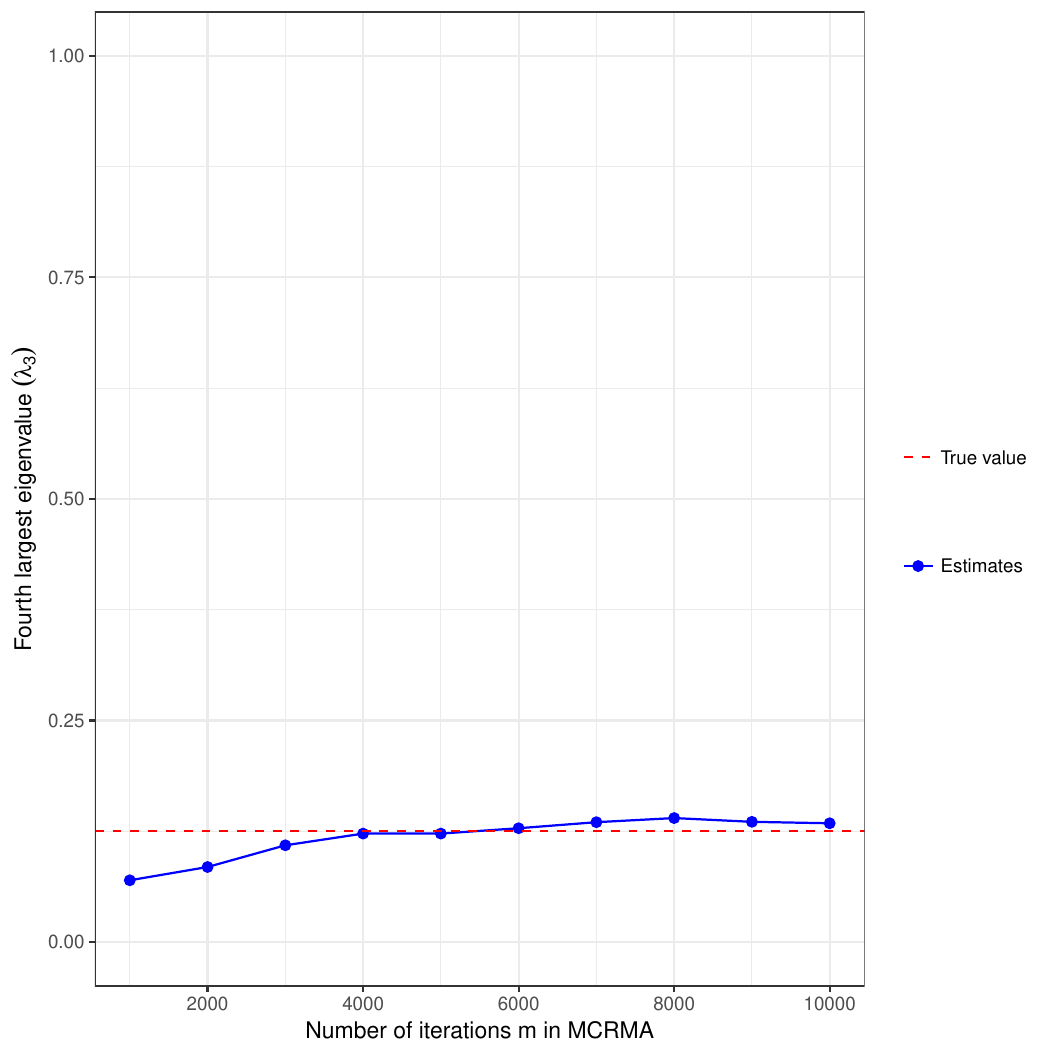} \label{toy_normal_lambda_3}} \\ 
	\subfloat[Fifth largest eigenvalue as a function of iterations $m$ in the MCRMA algorithm. The true eigenvalue is shown as a horizontal line.]{\includegraphics[height = 0.25\textheight]{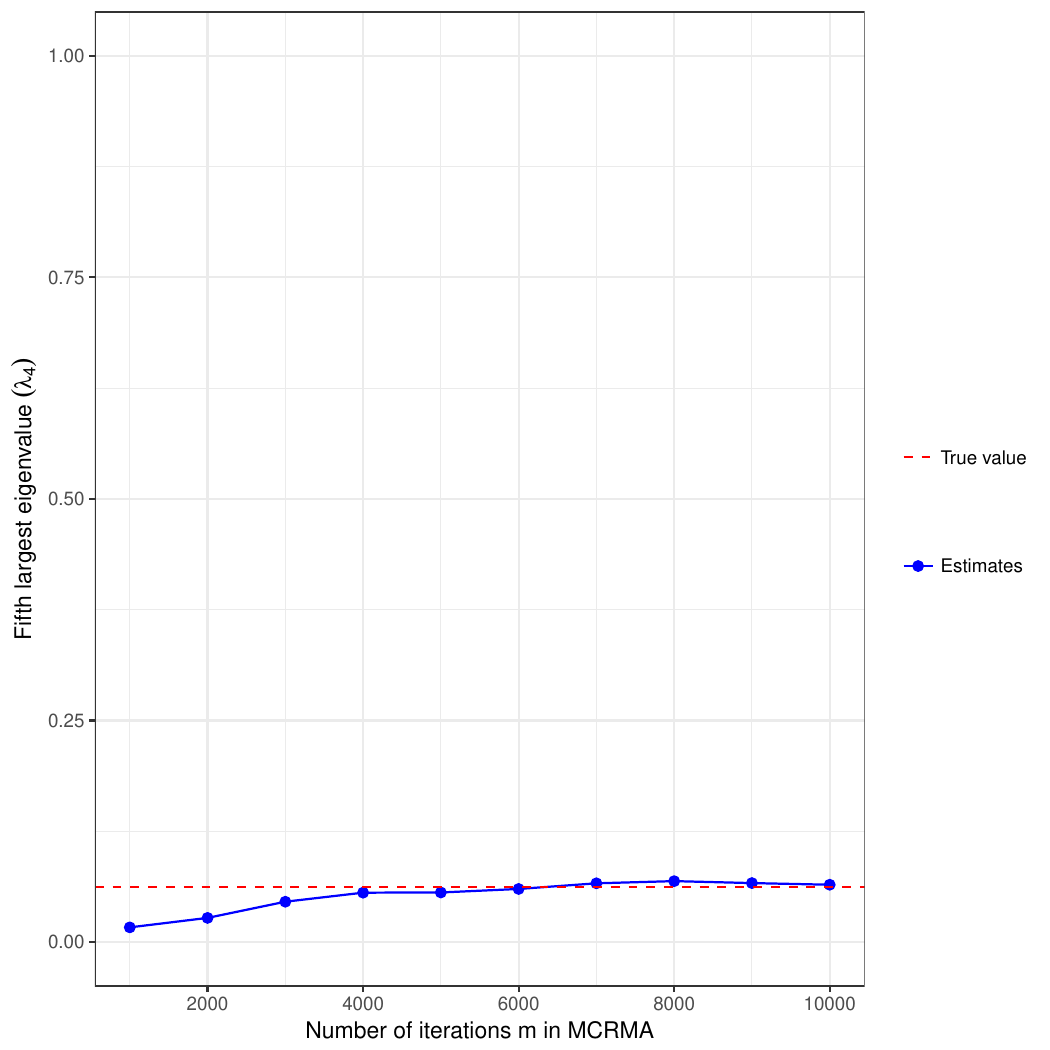} \label{toy_normal_lambda_4}}
	\qquad
	\subfloat[Sixth largest eigenvalue as a function of iterations $m$ in the MCRMA algorithm. The true eigenvalue is shown as a horizontal line.]{\includegraphics[height = 0.25\textheight]{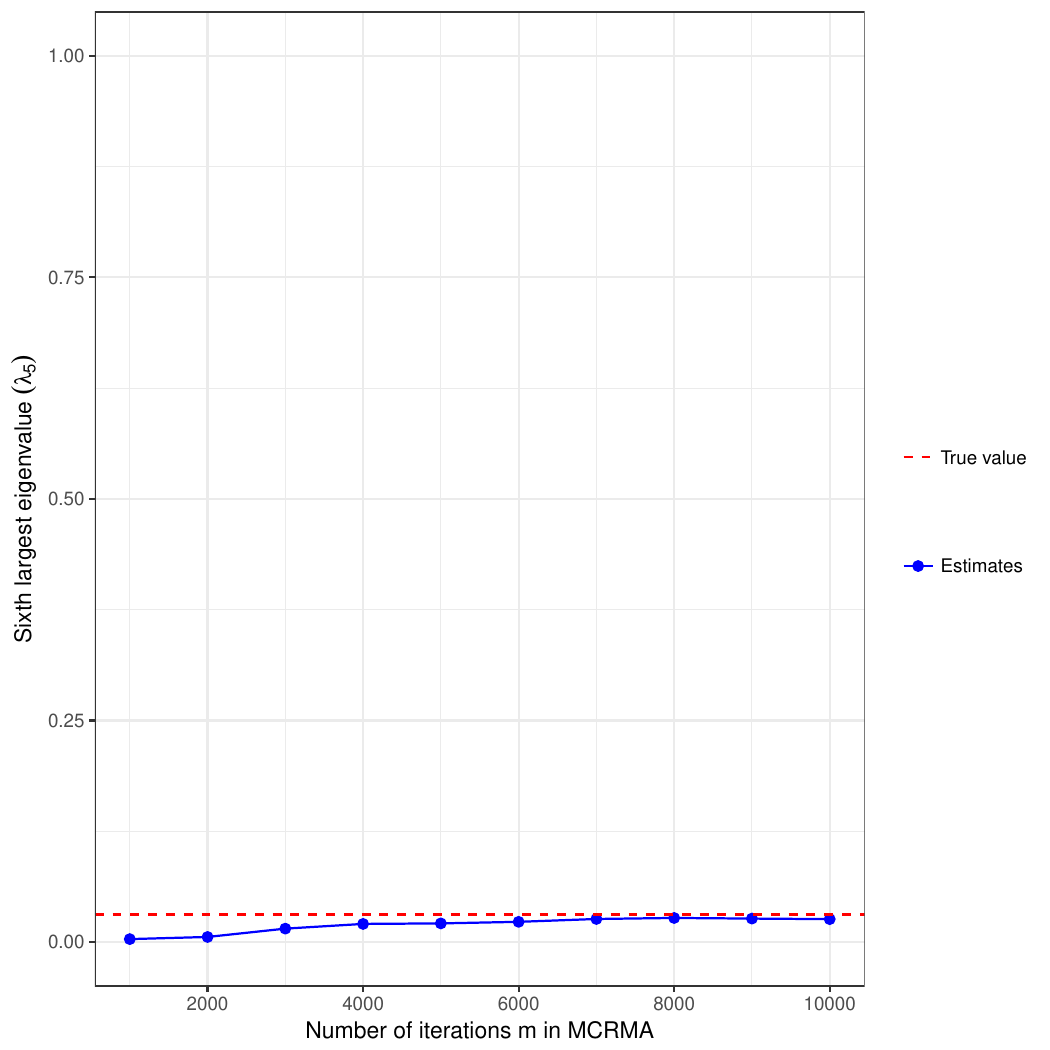} \label{toy_normal_lambda_5}}
	\\
	\caption{Eigenvalue estimates for the toy normal-normal DA Markov chain using the MCRMA algorithm.}
	\label{toy_normal_plots}
\end{figure}

Starting from $x = 0$, we first generate 10,000 realizations of the above Markov chain, after discarding a burn-in of size 10,000, and then extract the $x$ chain. Then we run 10 instances of the MCRMA algorithm \ref{algo_mcrma_ext} with $m = 1000, 2000, \cdots, 10,000$ (by using the first $1000, 2000, \cdots, 10,000$ iterations of the already generated Markov chain), and $N = N(m) = \left\lceil{m^{1+10^{-6}}}\right\rceil$, where for a real number $x$, $\lceil x \rceil$ denotes the ``ceiling'' of $x$, i.e., the smallest integer bigger than $x$. Then we look at the largest 11 eigenvalues (including $\lambda_0 = 1$) obtained from each instance of MCRMA. Note that the choice of $m$ and $N$ used here ensures strong consistency of the MCRMA estimates; weak consistency only requires $N(m) \rightarrow \infty$ as $m \rightarrow \infty$ (see Remark~\ref{rem_thm_mcrma_consis}). To understand the accuracy of MCRMA, the estimated eigenvalues are compared to the true eigenvalues, by displaying the estimates and truths on the same plot.  The resulting plots are shown in Figure~\ref{toy_normal_plots}. Figure~\ref{toy_normal_all} displays all 11 eigenvalues obtained from each of the 10 MCRMA instances (shown as 10 curves, one for each MCRMA instance), along with the true eigenvalues (shown as red dots). The second, third, fourth, fifth and sixth largest estimated eigenvalues, viewed  as functions of the MCRMA iteration size $m$, are shown separately in Figures~\ref{toy_normal_lambda_1} through \ref{toy_normal_lambda_5}, with the dotted line indicating the corresponding the true eigenvalue. The noticeable similarity between the truth and the estimates, (especially for the instances with $m \geq 5000$, where the estimates show satisfactory signs of convergence),  illustrates accuracy of the MCRMA method.

In spectral based diagnostics of MCMC algorithms, interest often lies in comparing the convergence rates of the estimated spectra, and that of the estimated functionals of interest (such as the posterior mean in Bayesian statistical analysis).  Here we consider the estimated (ergodic averages) mean of associated second Hermite polynomial for $x$, $H_2(x) = \frac{1}{\sqrt{2}}(2x^2 - 1)$, and compare its convergence to $\pi H_2 = 0$ with the convergence of MCRMA estimate of the second largest eigenvalue $\lambda_1$ to 0.5. In particular, for $m = 1000, 2000, \cdots, 10,000$, we compute
\begin{enumerate}[label = (\roman*)]
	\item $\hat{\pi H_2}{(m)} := m^{-1} \sum_{i=0}^{m-1} H_2(x_i)$  using the already generated Markov chain realizations $\{x_i: i = 0, \cdots, 9999\}$
	
	\item $\hat \lambda_1 (m) - 0.5$, where $\hat \lambda_1(m)$ is the estimated second largest eigenvalue obtained from the MCRMA instance ran with iteration size $m$
\end{enumerate}
and plot $|\hat{\pi H_2}{(m)}|$ and $|\hat \lambda_1 (m) - 0.5|$, both as functions of $m$, in the same diagram. The resulting plots are shown in Figure~\ref{toy_normal_cvg}, which shows that the convergence rates of the estimated spectrum and the estimated mean of second Hermite polynomial are comparable when $m \geq 5000$, and neither convergence is strictly than the other.  

\begin{figure}
	\centering	
	\includegraphics[width=0.7\linewidth]{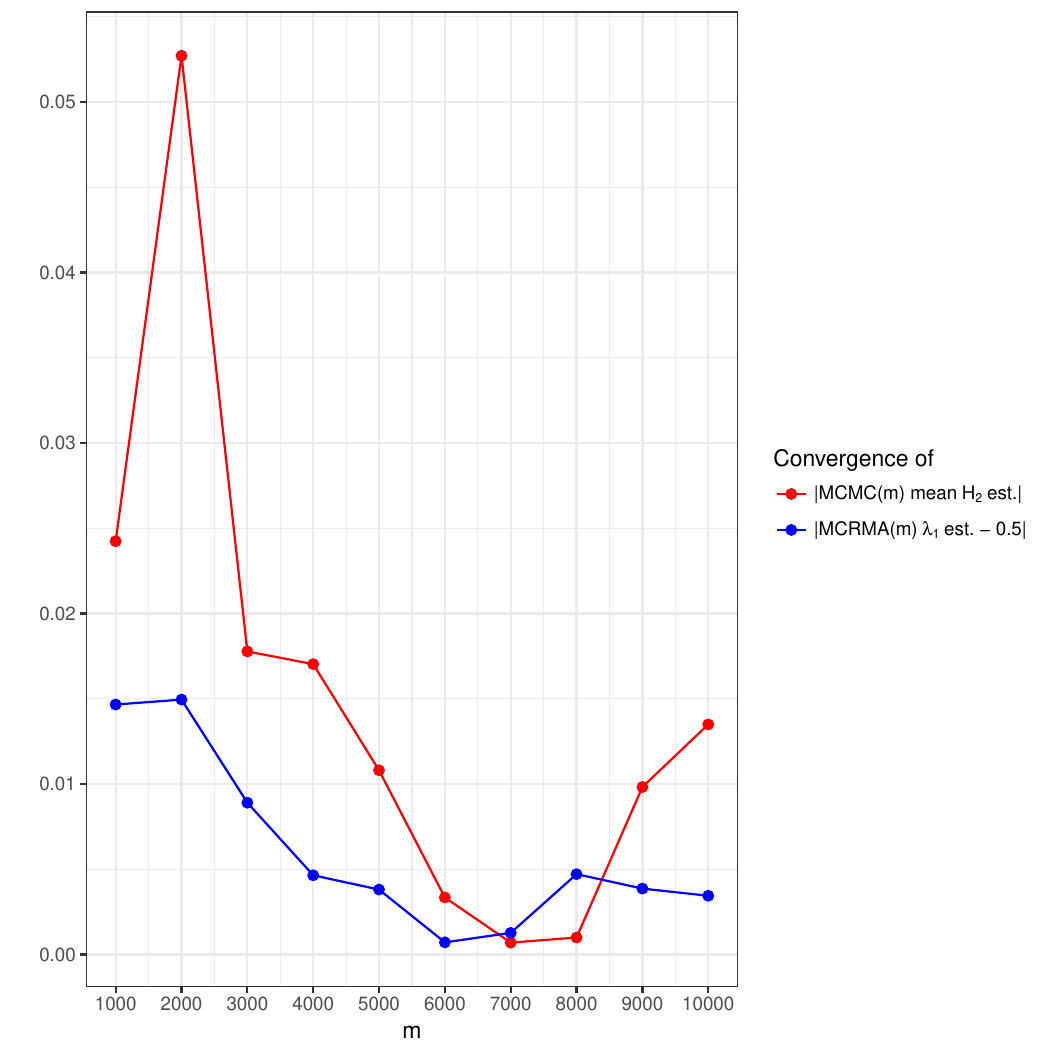} 
	\caption{Convergences of the MCMC estimate of mean second Hermite polynomial and MCRMA estimate of second largest eigenvalue, both viewed as functions of iteration size $m$. The absolute estimated means $|\hat{\pi H_2}{(m)}|$ are displayed as red dots, and the absolute differences $|\hat{\lambda}_{1}{(m)} - 0.5|$ are shown as blue dots.} \label{toy_normal_cvg}
\end{figure}

We end this example by comparing the MCRMA estimates of $\lambda_1$ to the estimates obtained using the power sum estimation technique of \cite{qin:hobert:khare:2017}, which we briefly describe in the following.  For a positive integer $r$, define the power sum $s_r := \sum_{i=0}^\infty \lambda_i^r$ of  eigenvalues $(\lambda_i)_{i \geq 0}$ of the associated trace class Markov operator $K$. Then for any $r \geq 1$
\[
l_r := \frac{s_r - 1}{s_{r-1} - 1} \leq  \lambda_1 \leq (s_r - 1)^{1/r} =: u_r  
\]
(with $s_0 = \infty$), and in addition, \citet[Proposition 1]{qin:hobert:khare:2017} show that as $1 \leq r \rightarrow \infty$, $l_r \uparrow \lambda_1$ and $u_r \downarrow \lambda_1$. For DA Markov operators, the authors provide an IID Monte Carlo based estimation technique for $s_r$, which in turn, provides estimates of $l_r$ and $u_r$, thus providing asymptotically consistent interval estimates of $\lambda_1$. The authors note that $r$ is not required to be very large in practice (in fact, very large $r$ cause instability in estimation, see \citet[Section~6]{qin:hobert:khare:2017}), and they recommend using $r$ large enough so that the difference between estimated $s_r$ and 1 is small.

\begin{figure}
	\centering	
	\includegraphics[width=0.7\linewidth]{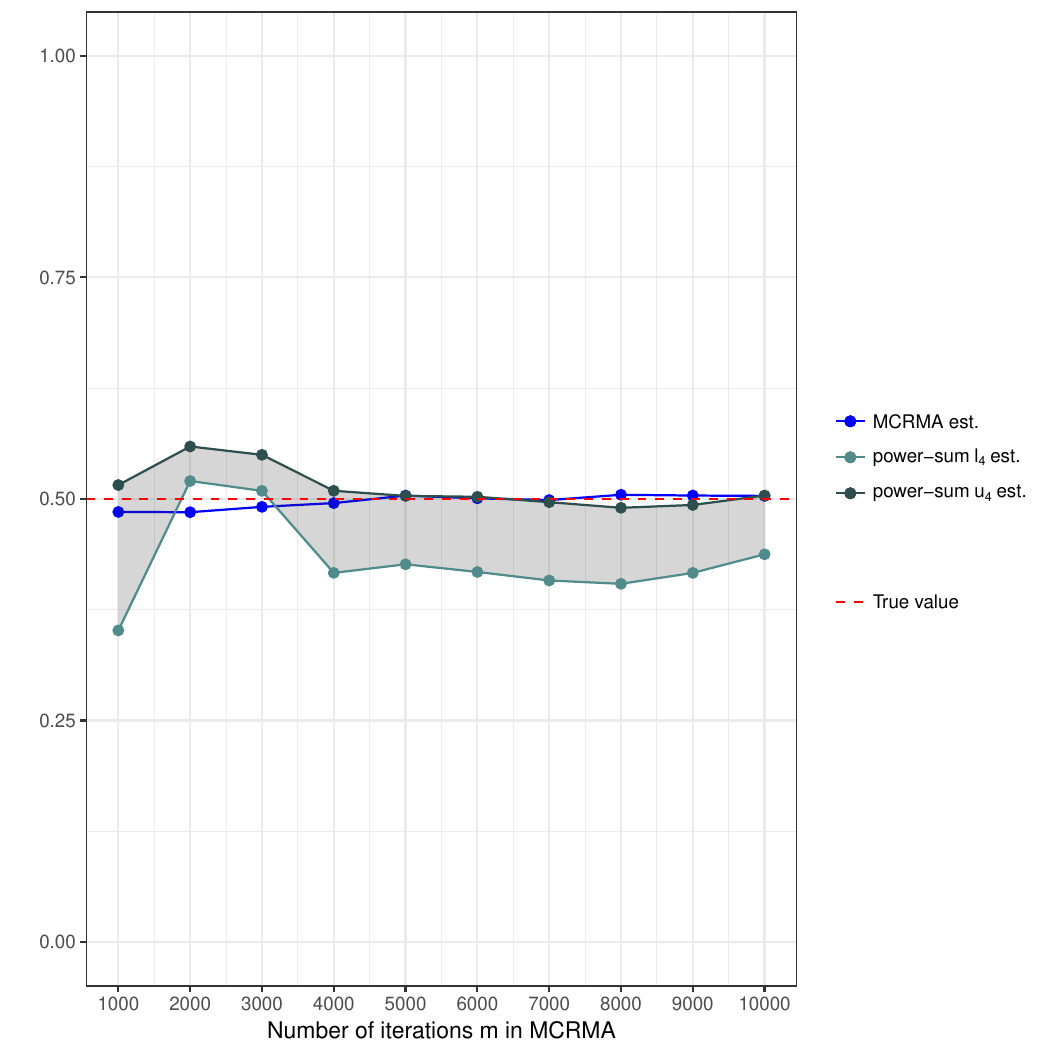} 
	\caption{Comparing MCRMA estimates with power-sum estimates of $\lambda_1 = 0.5$. The 10 MCRMA estimates with $m = 1000, \cdots, 10,000$ and $N(m) = \left\lceil{m^{1+10^{-6}}}\right\rceil$ are shown as blue dots, the light and dark gray dots are the power-sum estimates of $l_r$ and $u_r$ with $r = 4$ and Monte Carlo sample size $N(m)$, and the shaded gray region provides power-sum interval estimates of $\lambda_1$. The true $\lambda_1$ is shown as a horizontal red line.} \label{plot_mcrma_tracesum_compare}
\end{figure}   

The key step in the power-sum estimation method of \citet{qin:hobert:khare:2017} is the step of Monte Carlo estimation of $s_r$. To aid comparability with MCRMA, we set the associated Monte Carlo sample size to be the same $N(m) = \left\lceil{m^{1+10^{-6}}}\right\rceil$, and run 10 instances of  trace sum estimation method with $m = 1000, \cdots, 10,000$ and $r=4$. The estimated $(l_r, u_r)$ are then plotted as functions of $m$ together with the MCRMA estimates of $\lambda_1$, and displayed in Figure~\ref{plot_mcrma_tracesum_compare}, which shows that MCRMA gives slightly better estimates in the current settings.

\subsection{Infinite State Space Application: Polson and Scott DA Gibbs Sampler} \label{sec_illus_infinite}
We consider the Data Augmentation algorithm for Bayesian logistic regression proposed in \citet{polson:scott:windle:2013}. Let $Y_1, Y_2, 
\cdots, Y_n$ be independent Bernoulli random variables with $P(Y_i = 1 \mid \betab) = F(\ub_i^T \betab)$.  Here $\ub_i \in \R^{p}$ is a vector of known covariates associated with $Y_i$, $i = 1,\dots, n$, $\betab \in \R^p$ is a vector of unknown regression coefficients, and $F : \R \rightarrow [0,1] : t \mapsto e^t /(1 + e^t)$ is the distribution function of a standard logistic distribution.  For $y_i \in \{0,1\}$, the likelihood function under this model is given by:
\[
P(Y_1 = y_1, \dots, Y_n = y_n|\betab) = \prod_{i=1}^n \left[F(\ub_i^T\betab)\right]^{y_i} \left[1- F(\ub_i^T\betab)\right]^{1- y_i}
\]
The objective is to make inferences on the regression parameter $\betab$ and we intend to adopt a Bayesian approach, which requires a  prior density for $\betab$ to be specified. To keep parity with the literature, in this section we shall slightly abuse our notation by using $\betab$ (not $X$) to denote the parameter of interest, $U$ to denote the non-stochastic design matrix with $i$th row $\ub_i^T$, and $\pi(\betab)$ to denote the prior density for $\betab$.  Note that here our target distribution is not the prior density $\pi(\cdot)$, but the posterior density $\pi(\cdot \mid \yb)$ given the data $\yb = (y_1, \dots, y_n)^T$, which is given by
\begin{equation} \label{post_beta_inc}
\pi(\betab \mid \yb) = \frac{1}{c(\yb)}\: \pi(\betab) \left(\prod_{i=1}^n \left[F(\ub_i^T\betab)\right]^{y_i} \left[1- F(\ub_i^T\betab)\right]^{1- y_i} \right)
\end{equation}
where
\begin{equation} \label{norm_const}
c(\yb) = \int_{\R^p} \pi(\betab) \left(\prod_{i=1}^n \left[F(\ub_i^T\betab)\right]^{y_i} \left[1- F(\ub_i^T\betab)\right]^{1- y_i} \right) d\betab
\end{equation}
is the normalizing constant dependent of $\yb$ only.  We shall consider a proper $\N_p(\bbo, B)$ prior for $\betab$, as in \citet{choi:hobert:2013}. Note that the posterior density $\pi(\betab \mid \yb)$ is intractable; it does not have a closed form, and IID sampling is very inefficient even for moderately large $p$. \citet{polson:scott:windle:2013} proposed a data augmentation Gibbs sampling algorithm for approximate sampling  from  $\pi(\betab \mid \yb)$, which only requires random generation from easy to sample univariate distributions. In the following, we borrow notations from \citet{choi:hobert:2013} where the uniform ergodicity of the Markov chain produced by the \psw DA algorithm is proved.  Let $\R_+ = (0, \infty)$, and for fixed $\wb = (w_1, \dots, w_n)^T \in \R^n_+$
\begin{align*} 
\Omega(\wb) &= \diag(w_1, \dots, w_n), \\
\Sigma(\wb) &= \left(U^T \Omega(\wb) U + B^{-1}\right)^{-1}, \\
\text{and }\mub(\yb) = \mub &= U^T\left(\yb - \frac12 \bm{1}_n\right) + B^{-1}\bbo.
\end{align*}

\noindent Then the Polson, Scott and Windle DA Gibbs sampler (Algorithm~\ref{algo_psw_gs}) for generating MCMC samples from the posterior distribution $\pi(\betab \mid \yb)$  is obtained by iteratively sampling independent $w_i$ from (univariate) $\pg \left(1, \left|\ub_i^T \betab  \right|\right)$ distribution, for $i = 1, \cdots, n$, and then sampling $\betab$ from $\N_p \left(\Sigma(\wb) \mub, \Sigma(\wb) \right)$. Here $\pg(1,c)$ denotes the Polya-Gamma distribution with parameters 1 and $c$, which is defined as follows. Let $\displaystyle (E_k)_{k\geq 1}$ be a sequence of IID standard Exponential random variables, and let
\begin{equation} \label{w_from_E}
W = \frac{2}{\pi^2} \sum_{l=1}^\infty \frac{E_l}{(2l-1)^2}
\end{equation}
which has density
\begin{equation} \label{gt_defn}
\gt(w) = \sum_{l=1}^\infty (-1)^l \: \frac{(2l + 1)}{\sqrt{2\pi w^3}}\: e^{-\frac{(2l + 1)^2}{8w}}; \; w \geq 0.
\end{equation} 
Then the Polya-Gamma family of densities $\{\gt_c : c \geq 0\}$ is obtained through an exponential tilting of the density $\gt$: 
\[
\gt_c(x) = \cosh(c/2)\: e^{-\frac{c^2x}{2}} \: \gt(x),
\]
and a random variable is said to have a $\pg(1, c)$ distribution if it has density $\gt_c$. (Recall that $\cosh(t) = (e^t + e^{-t})/2$.) An efficient data generating algorithm from $\pg(1,c)$ is provided in \citet{polson:scott:windle:2013}.

\begin{algorithm}  [h]  
	Given a starting value $\betab_0$, iterate between the following two steps:	
	\begin{enumerate}[label = (\roman*)]
		\item Draw independent $w_1, \cdots, w_n$ with 	$\displaystyle w_i \sim \pg \left(1, \left|\ub_i^T \betab  \right|\right),\; i = 1,\cdots,n$, and define
		\begin{align*} 
		\Omega(\wb) &= \diag(w_1, \dots, w_n), \\
		\Sigma(\wb) &= \left(U^T \Omega(\wb) U + B^{-1}\right)^{-1}, \\
		\text{and }\mub(\yb) = \mub &= U^T\left(\yb - \frac12 \bm{1}_n\right) + B^{-1}\bbo.
		\end{align*}.
		
		\item Draw $\displaystyle \betab \sim \N_p \left(\Sigma(\wb) \mub, \Sigma(\wb) \right)$.
	\end{enumerate}
	\caption{The Polson, Scott and Windle DA Gibbs Sampler}
	\label{algo_psw_gs}  
\end{algorithm}

\noindent From the \psw Gibbs sampler (Algorithm~\ref{algo_psw_gs}), it follows that
\begin{enumerate}
	\item For $i = 1, \cdots, n$, the (full) conditional posterior distribution of $w_i$ given $\betab$ is independent $\pg \left(1, \left|\ub_i^T \betab  \right|\right)$, so that the  conditional joint density of $\wb = (w_1, \dots, w_n)^T$  given $\betab, \yb$ is given by 
	\begin{equation} \label{pi_w_given_beta}
	\pi(\wb \mid\betab, \yb) = \prod_{i=1}^{n} \left\lbrace  \cosh\left(\frac{|\ub_i^T\betab \mid }{2}\right) \: \exp\left[-\frac12 (\ub_i^T\betab)^2 w_i\right] \: \gt(w_i) \right\rbrace
	\end{equation}
	where $\gt$ is as given in (\ref{gt_defn}).
	
	\item The full conditional distribution of $\betab$ given $\wb, \yb$ is $\N_p \left(\Sigma(\wb) \mub, \Sigma(\wb) \right)$ with density
	\begin{align*} \label{pi_beta_given_w}
	\pi(\betab \mid \wb, \yb) &= \left({2\pi}\right)^{-p/2} \: \left|U^T \Omega(\wb) U + B^{-1}\right|^{1/2} \\
	& \qquad \qquad \times    \exp \left[-\frac12 \left(\betab - \Sigma(\wb) \mub \right)^T \Sigma(\wb)^{-1} \left(\betab - \Sigma(\wb) \mub \right) \right]. \numbereqn
	\end{align*}
\end{enumerate}

\noindent Note that the transition density of the associated Markov chain $\Phi$ for $\betab$ is given by
\[
k(\betab, \betab') = \int_{\R_+^n} \pi(\betab'|\wb, \yb) \:\pi(\wb \mid\betab, \yb) \:d\wb
\]
where $\pi(\betab \mid \wb, \yb)$ and $\pi(\wb \mid\betab, \yb)$ are as given in (\ref{pi_beta_given_w}) and (\ref{pi_w_given_beta}) respectively. It is clear that this transition density cannot be evaluated in closed form. Moreover, a closed form expression for the normalizing constant $c(\yb)$ in (\ref{norm_const}) is not available, which means the posterior density $\pi(\betab \mid \yb)$ in (\ref{post_beta_inc}) can only be specified up to a constant factor. Thus, exact RMA cannot be applied in this example. However, by letting  $\wb$ play the role of the augmented data $z$, $f_{Z|X}(\cdot \mid \cdot)$  the conditional density $\pi(\wb \mid\betab, \yb)$ (from which random sampling is easy due to the efficient simulation algorithm from $\pg(1, c)$ proposed in \citet{polson:scott:windle:2013}), and  $f_{X|Z}(\cdot \mid \cdot)$ the simple multivariate normal density $\pi(\betab \mid \wb, \yb)$, we can use the extended MCRMA method (Algorithm~\ref{algo_mcrma_ext}). Since the state space of $\betab$ (and $\wb$) is infinite, in order to ensure consistency of the MCRMA estimates, we need \ref{trace_condn},  and \ref{var_condn} in Theorem~\ref{thm_mcrma_consist} to hold. The following two theorems (Theorem~\ref{tr_thm_psw} and Theorem~\ref{var_thm_psw}) show that the \psw Markov chain does indeed satisfy these two conditions, thus guaranteeing consistency of MCRMA estimates in this case. Proofs of Theorem~\ref{tr_thm_psw} and Theorem~\ref{var_thm_psw} are provided in Section~\ref{supp_sec_illus} of the Appendix.

\begin{theorem} \label{tr_thm_psw}
	For any choice of the (proper multivariate normal) prior parameters $\bbo$ and $B$, the Markov operator associated with \psw Markov chain $\Phi$ is trace class.
\end{theorem}

\begin{theorem} \label{var_thm_psw}
	Let the initial distribution $\nu_0$ of $\betab$ be such that $\exp\left[\frac12 \: \sum_{i=1}^{n} |\ub_i^T\betab \mid  \right]$ is $\nu_0$-integrable for all $i=1,\cdots, n$. Then the operator $K$ associated with the \psw algorithm satisfies the variance condition \ref{var_condn}. 
\end{theorem}

\begin{remark}
	Note that the integrability condition assumed on the initial measure $\nu_0$ in Theorem~\ref{var_thm_psw} is not very restrictive, and can be easily ensured in practice for a number of families of distribution. For example, if the initial distribution of $\betab$ is Gaussian, integrability of $\exp\left[\frac12 \sum_{i=1}^{n} |\ub_i^T\betab \mid \right]$ is immediate.
\end{remark}

\subsubsection{Simulation Results}
For simulation, we used the the R package \texttt{BayesLogit} \citep{polson:scott:windle:2013} to efficiently draw random samples from the Polya Gamma density. We generated a \psw Markov chain on the  \texttt{nodal} dataset from the R package \texttt{boot} \citep{canty:ripley:2017, davidson:hinkley:1997}. The dataset consists of 53 observations on 5 binary predictors (aged,  stage, grade, xray and  acid) and one response which indicates whether cancer has spread from prostate to surrounding lymph nodes. Taking the maximum likelihood estimate as the starting value, we first generated 30,000 iterations of the \psw  chain for  the regression coefficient $\betab$ ($\in \R^6$, includes one intercept coefficient). We discarded the first 20,000 iterations as burn-in, and kept the remaining 10,000 as the MCMC sample. Then we ran 10 instances of the MCRMA algorithm with the MCMC sample already generated and with $m = 1000, 2000, \cdots, 10,000$ and $N(m) = \left\lceil{m^{1+10^{-6}}}\right\rceil$ (to ensure strong consistency), and recorded the 30 largest eigenvalues. Then we created plots similar to the toy normal-normal DA example, except, the true eigenvalues were of course unknown in this case. The resulting plots are shown in Figure~\ref{psw_plots}. 

\begin{figure}[htpb]
	\centering	
	\subfloat[The largest 30 eigenvalues. There are 10 curves, each corresponding to the choices $m = 1000, \cdots, 10,000$ in the MCRMA algorithm.]{\includegraphics[height = 0.25\textheight]{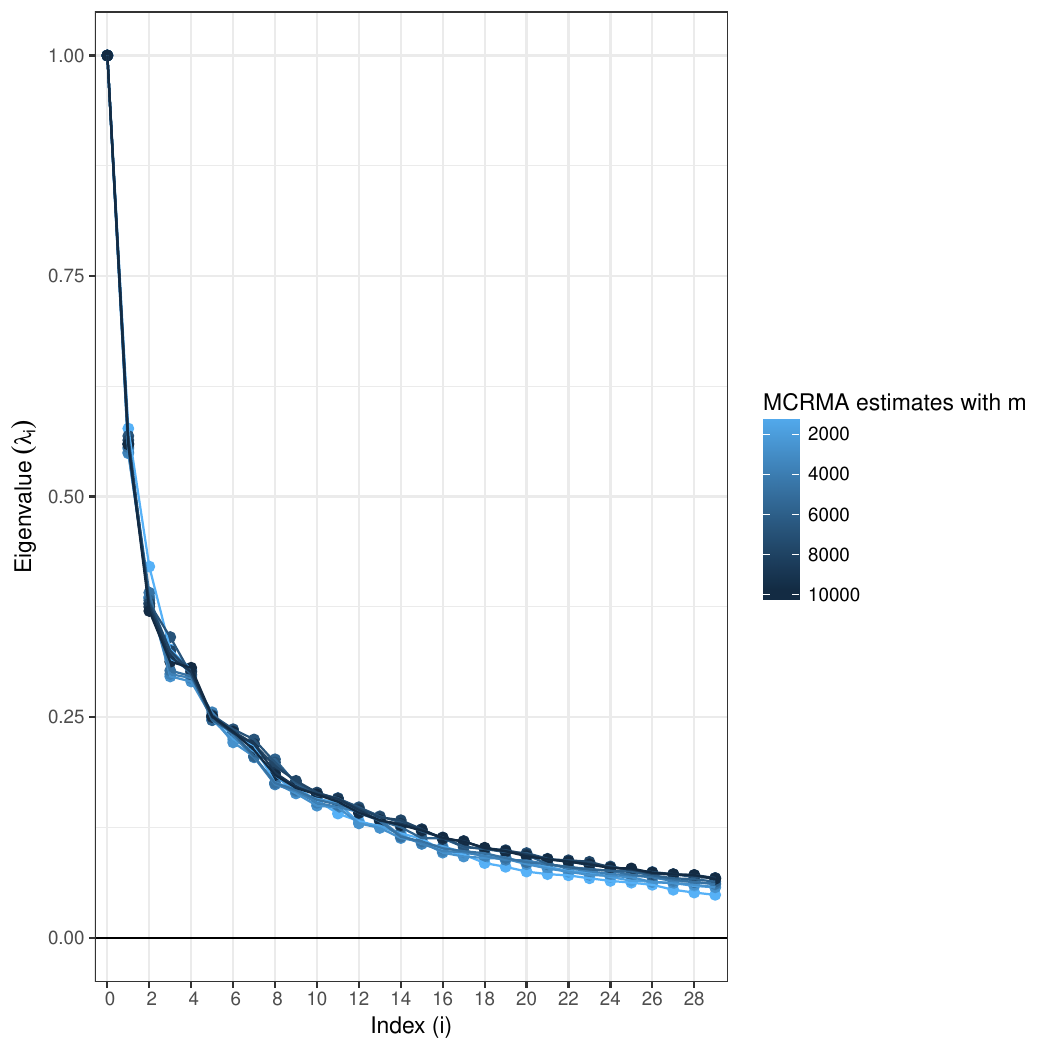} \label{psw_all}}
	\qquad
	\subfloat[Second largest eigenvalue as a function of iterations $m$ in the MCRMA algorithm.]{\includegraphics[height = 0.25\textheight]{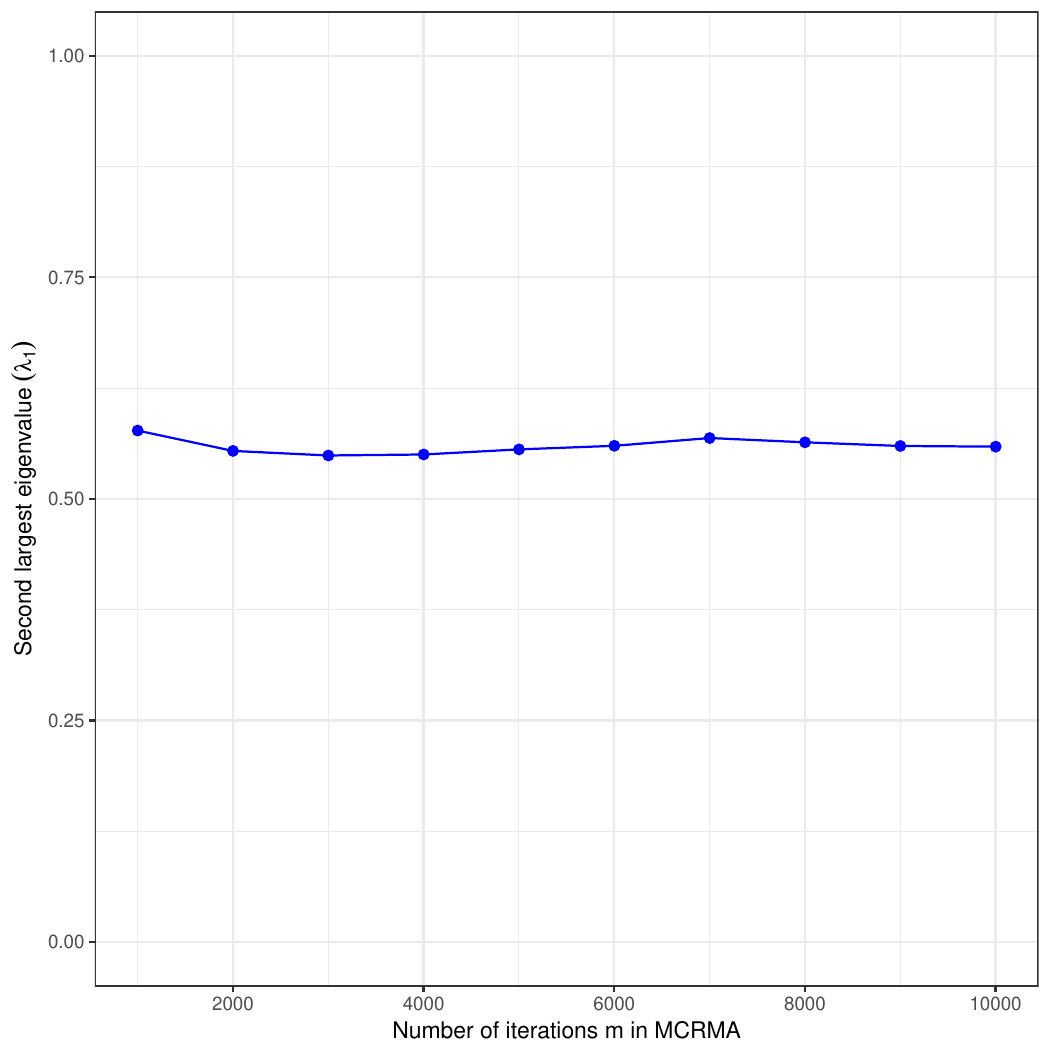} \label{psw_lambda_1}}\\
	\subfloat[Third largest eigenvalue as a function of iterations $m$ in the MCRMA algorithm.]{\includegraphics[height = 0.25\textheight]{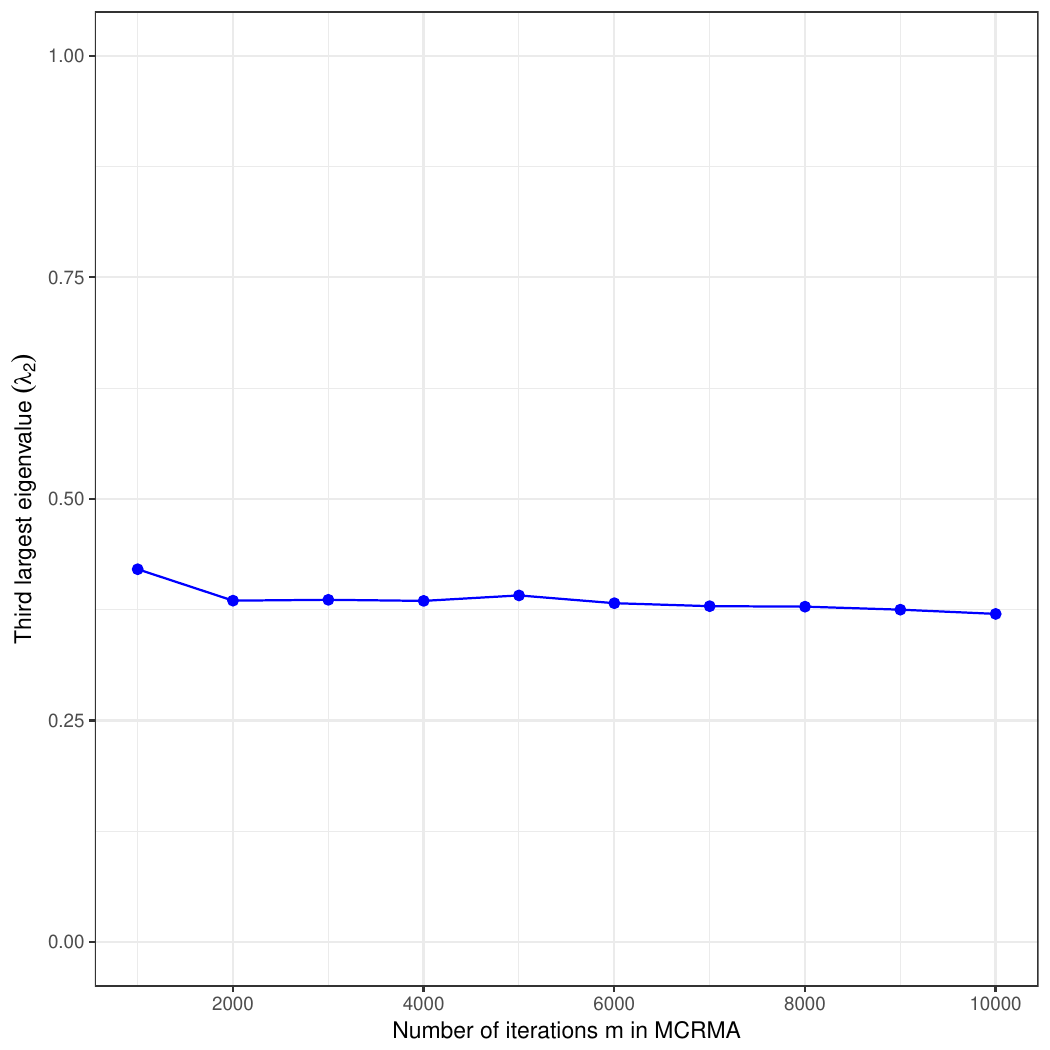} \label{psw_lambda_2}}
	\qquad
	\subfloat[Fourth largest eigenvalue as a function of iterations $m$ in the MCRMA algorithm.]{\includegraphics[height = 0.25\textheight]{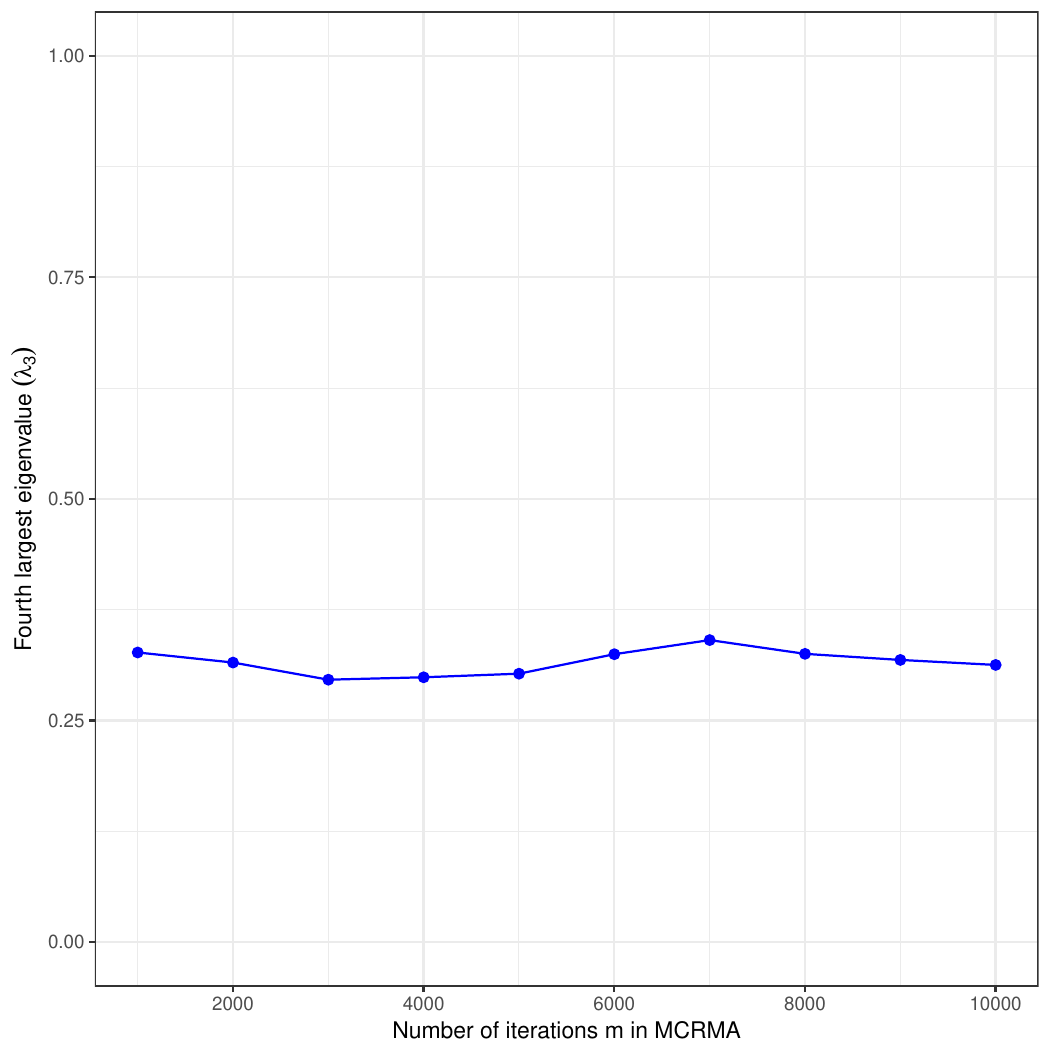} \label{psw_lambda_3}} \\ 
	\subfloat[Fifth largest eigenvalue as a function of iterations $m$ in the MCRMA algorithm.]{\includegraphics[height = 0.25\textheight]{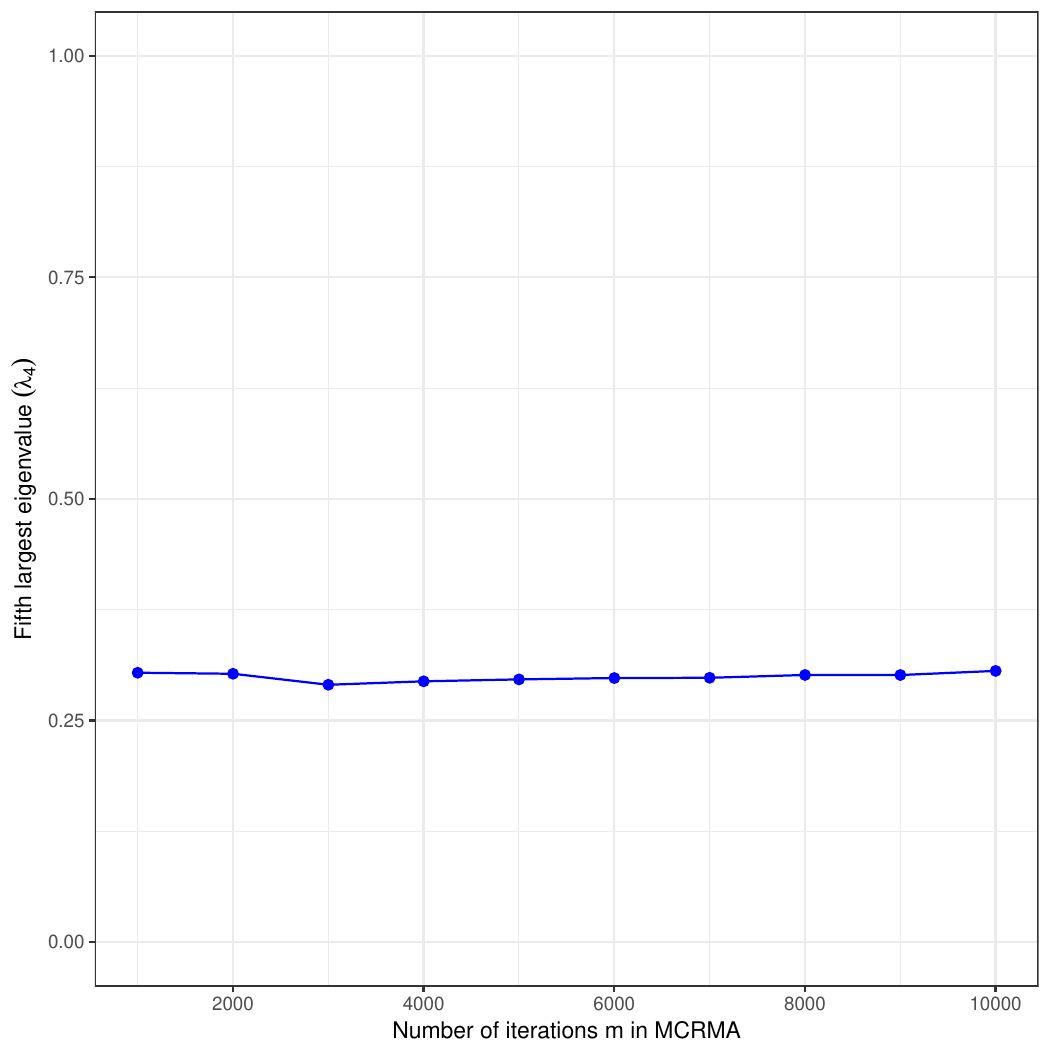} \label{psw_lambda_4}}
	\qquad
	\subfloat[Sixth largest eigenvalue as a function of iterations $m$ in the MCRMA algorithm.]{\includegraphics[height = 0.25\textheight]{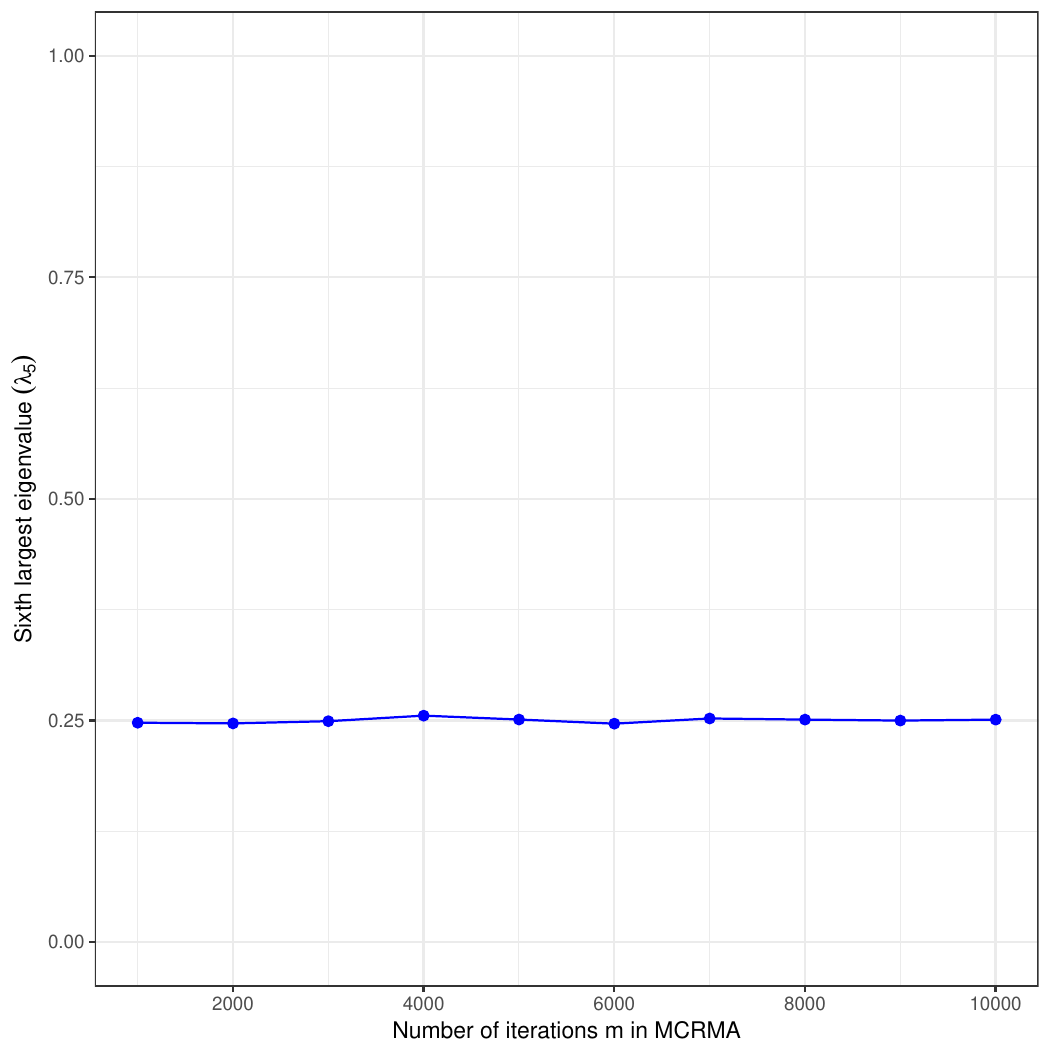} \label{psw_lambda_5}}
	\\
	\caption{Eigenvalue estimates for the \psw DA Markov chain using the MCRMA algorithm.}
	\label{psw_plots}
\end{figure}

Figure~\ref{psw_all} shows all 30 eigenvalues obtained from each of the 20 MCRMA instances, plotted as 20 curves, one for each MCRMA instance. The second, third, fourth, fifth and sixth largest estimated eigenvalues, viewed  as functions of the MCRMA iteration size $m$, are shown in Figures~\ref{psw_lambda_1} through \ref{psw_lambda_5}. As is clear from the plots, the MCRMA spectrum estimates for the \psw chain show adequate signs of convergence when $m \geq 5000$, thereby providing confidence on the accuracy of estimation.


\subsection{Finite State Space Application: Two Component Normal Mixture} \label{sec_illus_finite}
In this section we consider the problem of Bayesian finite mixture modeling with two components. Let $\yb = (y_1, \cdots, y_n)$ be a random sample from the two component equal variance mixture normal density 
\begin{equation} \label{2_comp_normmix}
f(y|\mub, p) = p \frac{1}{\tau} \phi \left(\frac{y-\mu_1}{\tau}\right) + (1-p) \frac{1}{\tau} \phi \left(\frac{y-\mu_2}{\tau}\right)
\end{equation}
where $p \in [0,1]$ is the mixing proportion, $\mub = (\mu_1 ,\mu_2 ) \in \R^2$ is the vector of component means and $\tau^2 > 0$ is the \emph{known} variance for both components, and $\phi(\cdot)$ is the standard normal density function. The objective is to make inferences on the unknown parameter vector $\thetab = (\mub, p)$ through the data $\yb$, and we adopt a Bayesian approach. The prior density for $\thetab$ is taken to be of the form $\pi(\thetab) = \pi(p) \pi(\mu_1) \pi(\mu_2)$, with $\pi(p)$ being the Uniform$(0,1)$ density, and $\pi(\mu_j)$ being the $\N(0, \tau^2)$ density. 
Then the posterior density for $\thetab$ is given by 
\begin{equation} \label{post_y_normmix} 
\pi(\thetab \mid \yb) =  \pi(\mub,  p \mid \yb) = \frac{1}{c(\yb)} \prod_{i=1}^n \left\lbrace p \frac{1}{\tau} \phi \left(\frac{y_i-\mu_1}{\tau}\right) + (1-p) \frac{1}{\tau} \phi \left(\frac{y_i-\mu_2}{\tau}\right) \right \rbrace \pi(\thetab)
\end{equation}
where $c(\yb)$ is the normalizing constant that makes (\ref{post_y_normmix}) a proper density. It is clear that $\pi(\thetab \mid \yb)$ is intractable, which makes evaluation of posterior mean or IID  simulation infeasible. We therefore resort to approximate sampling via MCMC. A slightly general version of this problem (with \emph{unknown} and \emph{different} component variances $\tau_1^2$ and $\tau_2^2$) is considered in \citet[section 6.2]{hobert:roy:robert:2011} and the authors consider two different Gibbs sampling algorithms, namely, the mixture Data Augmentation algorithm (MDA algorithm, or simply, MDA) and the \fs  algorithm (FS algorithm, or simply, FS),  to generate MCMC samples from the posterior. 

\subsubsection{MDA Algorithm}
Let us introduce latent component indicators $z_1, \cdots, z_n$, with $z_i = j$ indicating that the $i$th observation $y_i$ is coming from the $j$th component $\N(\mu_j, \tau^2)$ for $j=1,2$.  Then,
\begin{enumerate}
	\item the full conditional posterior distribution of the components of $\thetab$ given $\zb$ are independent, with 
	$p$ being Beta$(c_1+1, c_2+1)$ and $\mu_j$ being $\N\left(\frac{c_j}{c_j+1} \bar y_j, \frac{\tau^2}{c_j+1}\right)$ with $c_j = \sum_{i=1}^n \one_{\{j\}}\left(z_i\right)$ and $\bar y_j = c_j^{-1} \sum_{i=1}^n y_i \one_{\{j\}}\left(z_i\right) $ for $j = 1,2$. We shall denote the corresponding density of $\thetab$ by $\pi(\thetab \mid \zb, \yb)$.
	
	\item the full conditional posterior density (\emph{mass}, with respect to the counting measure $\zeta$) of $\zb$ given $\thetab$ is given by 
	\begin{align*}
	\pi(\zb \mid \thetab, \yb) \propto \prod_{i=1}^n \left(\pt_i{\one_{\{1\}}(z_i)} +  (1-\pt_i){\one_{\{2\}}(z_i)}\right) \numbereqn 
	\label{pi_z_given_theta}
	\end{align*}
	where 
	\[
	\pt_i = \frac{ p \: \phi \left(\frac{y_i-\mu_1}{\tau}\right)}{ p \: \phi \left(\frac{y_i-\mu_1}{\tau}\right) + (1-p) \: \phi \left(\frac{y_i-\mu_2}{\tau}\right)}.
	\]
\end{enumerate}
The MDA algorithm entails iterative generation of $\zb$ from $\pi(\zb \mid \thetab, \yb)$, and $\thetab$ from $\pi(\thetab \mid \zb, \yb)$. The resulting Gibbs sampler is formally displayed in Algorithm~\ref{algo_mda_gs}.  

\begin{algorithm}  [h]  
	Given a starting value $(\mub_0, p_0)$ for the parameter vector $\thetab = (\mub, p)$, iterate between the following two steps:
	\begin{enumerate}[label = (\roman*)]
		\item 
		Draw independent $z_1, \cdots, z_n$ with $z_i$ having a categorical probability distribution with categories 1 and 2, and   
		\[
		P(z_i = j) = \begin{cases}
		\frac{ p \: \phi \left(\frac{y_i-\mu_1}{\tau}\right)}{ p \: \phi \left(\frac{y_i-\mu_1}{\tau}\right) + (1-p) \: \phi \left(\frac{y_i-\mu_2}{\tau}\right)} & \text{ if } j = 1 \\
		\frac{ (1-p) \: \phi \left(\frac{y_i-\mu_2}{\tau}\right)}{ p \: \phi \left(\frac{y_i-\mu_1}{\tau}\right) + (1-p) \: \phi \left(\frac{y_i-\mu_2}{\tau}\right)} & \text{ if } j = 2
		\end{cases}.
		\]
		
		\item Compute $c_j = \sum_{i=1}^n \one_{\{j\}}\left(z_i\right)$ and $\bar y_j = c_j^{-1} \sum_{i=1}^n y_i \one_{\{j\}}\left(z_i\right) $. Then independently generate:
		\begin{enumerate}
			\item $p$ from Beta($c_1+1, c_2+1$)
			\item $\mu_j$ from $\displaystyle \N\left(\frac{c_j}{c_j+1} \bar y_j, \frac{\tau^2}{c_j+1}\right)$ for $j = 1, 2.$
		\end{enumerate}
	\end{enumerate}
	\caption{The Mixture DA (MDA) Gibbs Sampler}
	\label{algo_mda_gs}  
	
\end{algorithm}

Note that, although the parameter  vector $\thetab = (\mub, p)$ in the MDA algorithm lives on the infinite space $\X = \R^2 \times [0,1]$, the latent data $\zb = (z_1, \cdots, z_n)$ lives on the finite state space $\Z = \{1,2\}^n$. We shall, therefore, study the spectrum of the Markov operator $K^*$ associated with the latent data $\zb$ (see Remark~\ref{rem_z_finite}). The Markov transition density associated with the operator $K^*$ is given by 
\begin{equation} \label{mtd_nmda}
k^*(\zb, \zb') = \int_\X \pi(\zb' \mid \thetab, \yb) \: \pi(\thetab \mid \zb, \yb) \: d\thetab
\end{equation}
which is, of course, not available in closed form, because of the denominators of the product terms in $\pi(\zb' \mid \thetab, \yb)$. However, $\pi(\zb' \mid \thetab, \yb)$ is avaiable in closed form and $\pi(\thetab \mid \zb, \yb)$ is easy to sample from.  Thus, the MCRMA method can be applied here, and the estimates are guaranteed to be strongly consistent (Remark~\ref{rem_z_finite}). Recall that MCRMA requires evaluation of the stationary density $\pi(\zb \mid \yb)$ for $\zb$. Straight-forward calculations show that 
\begin{equation} \label{pi_z_mda}
\pi(\zb \mid \yb) \propto B(c_1+1, c_2+1) \: \prod_{j=1}^2 \left[(1+c_j)^{-\frac12}\: \exp\left(\frac{c_j^2 \: \bar y_j^2}{2 \tau^2 (1+c_j) } \right) \right].
\end{equation}
Since the normalizing constant that makes (\ref{pi_z_mda}) a density is not available in closed form, we shall, therefore use Algorithm~\ref{algo_mcrma_ext}.

\subsubsection{FS Algorithm} \label{sec_fs_algo}
Along with MDA, \citet{hobert:roy:robert:2011} consider another Gibbs sampling algorithm, called the \fs   (FS) algorithm \citep{fruhwirth:sylvia:2001}, which is obtained by inserting an intermediate random label switching step in between the two steps of MDA. The key idea here is to randomly permute the labels of the latent variable $\zb$ obtained in the first step of MDA, before moving on to the second step. That is, after generating $\zb$ from the conditional distribution of $\zb \mid \thetab$, instead of drawing the next state of $\thetab$ directly from $\thetab \mid \zb$, here one first randomly permutes the labels of components in the mixture model, and switches the labels of $\zb$ according to that random permutation to get $\zb'$. The next state of $\thetab$ is then generated from the conditional distribution of $\thetab \mid \zb'$. In the context of two component mixture models, the intermediate step $\zb \rightarrow \zb'$ simply entails performing a Bernoulli experiment with probability of success 0.5. One then takes $\zb' = \bar \zb$ or $\zb' = \zb$ according as whether the Bernoulli experiment results in a success or a failure, where $\bar \zb$ denotes $\zb$ with its 1's and 2's flipped.

The computationally inexpensive label switching step in the FS algorithm is introduced to force movement between the symmetric modes of the posterior density $\pi(\thetab \mid \yb)$. This makes the  FS algorithm superior to the MDA algorithm in terms of  convergence and mixing. The FS algorithm is in fact a member of a wide class of so-called \emph{sandwich algorithms}, where one inserts an inexpensive intermediate \emph{meat} step inside the two \emph{bread} steps of a DA algorithm to achieve better convergence and mixing. In fact, when the operator associated with a Markov chain is trace class, the spectrum of a sandwich chain is guaranteed to be bounded above by that of the parent DA chain, with at least one strict inequality \citep{khare:hobert:2011}. In the current setting, since the MDA Markov chain is trace class (the latent state space is finite), the FS chain is therefore guaranteed to be better mixing than the DA chain. To visualize or quantify the improvement, however, information on the actual spectra of the two chains is needed. Clearly, the spectrum of the FS chain, similar to the MDA chain, can neither be evaluated analytically nor can be estimated in  exact RMA method of \cite{adamczak:bednorz:2015}, since the associated Markov transition density is not available in closed form. Instead, we make use of MCRMA estimation, as described in the following.

The  usual sandwich representation (with three steps - two \emph{bread} steps similar to MDA and one additional \emph{meat} step) of the FS algorithm does not furnish a Markov transition density in the form (\ref{mtd_da_general}); however, following \citet[Section 5.2]{hobert:roy:robert:2011}, one can represent the algorithm as a DA algorithm with different joint (and hence, different full conditional), but same marginal posterior distributions as MDA. In particular, the DA representation of the FS algorithm  entails iterative random generation of $\zb$ from the conditional density $\pit(\zb \mid \thetab, \yb)$, and  $\thetab$ from $\pit(\thetab \mid \zb, \yb)$, where
\begin{gather*}
\pit(\thetab \mid \zb, \yb) = \int_{\Z} \pi(\thetab \mid \zb', \yb) \: r(\zb, \zb') \:d\zeta(\zb') \\
\text{and } \pit(\zb \mid \thetab, \yb) =  \frac{\pi(\zb \mid \yb)}{\pi(\thetab \mid \yb)} \int_{\Z} \pi(\thetab \mid \zb') \: r(\zb, \zb') \:d\zeta(\zb')
\end{gather*} 
and $r(\zb, \zb')$ is the transition density (with respective to the counting measure $\zeta$) associated with the intermediate \emph{meat} step $\zb \rightarrow \zb'$ in the sandwich representation of FS. Since the intermediate \emph{meat} step is that of  random label switchings, we have
\[
r(\zb, \zb') = \frac12 \one_{\{\zb\}}(\zb') + \frac12 \one_{\{\bar \zb\}}(\zb'),
\]
where $\bar \zb$ is $\zb$ with its 1's and 2's flipped, and therefore,
\begin{align*}
\pit(\thetab \mid \zb, \yb) = \frac12 \pi(\thetab \mid \zb, \yb) + \frac12 \pi(\thetab \mid \bar \zb, \yb) \numbereqn  \label{pi_tilde_theta_given_z}
\end{align*}
and
\begin{align*}
\pit(\zb \mid \thetab, \yb) = \frac12 \pi(\zb \mid \thetab, \yb) + \frac12 \pi(\bar \zb \mid \thetab, \yb) \numbereqn  \label{pi_tilde_z_given_theta}
\end{align*}
with (\ref{pi_tilde_z_given_theta}) being a consequence  of (\ref{pi_z_mda}). Note that $\pit(\thetab \mid \zb, \yb)$ and $\pit(\zb \mid \thetab, \yb)$ are just half-half mixtures of standard densities, and  can be easily sampled. The DA form of the FS algorithm is formally displayed in Algorithm~\ref{algo_fs_gs}.   

\begin{algorithm}  [h]  
	Given a starting value $(\mub_0, p_0)$ for the parameter vector $\thetab = (\mub, p)$, iterate between the following two steps:	
	\begin{enumerate}[label = (\roman*)]

		\item 	Draw independent $z_1', \cdots, z_n'$ with $z_i'$ having a categorical probability distribution with categories 1 and 2, and   
		\[		
		P(z_i' = j) = \begin{cases}
		\frac{ p \: \phi \left(\frac{y_i-\mu_1}{\tau}\right)}{ p \: \phi \left(\frac{y_i-\mu_1}{\tau}\right) + (1-p) \: \phi \left(\frac{y_i-\mu_2}{\tau}\right)} & \text{ if } j = 1 \\
		\frac{ (1-p) \: \phi \left(\frac{y_i-\mu_2}{\tau}\right)}{ p \: \phi \left(\frac{y_i-\mu_1}{\tau}\right) + (1-p) \: \phi \left(\frac{y_i-\mu_2}{\tau}\right)} & \text{ if } j = 2
		\end{cases},
		\]
		and call $\zb' = (z_1', \cdots, z_n')$.	Now perform a Bernoulli experiment with probability of success 0.5. If the experiment results in a success, define $\zb = \zb'$, or else define $\zb = \bar{\zb'}$, where $\bar{\zb'}$ is $\zb'$ with its 1’s and 2’s flipped.
		
		\item Perform another Bernoulli experiment with probability of success 0.5. Define $\zb^* = \zb$ if the experiment results in a success, and $\zb^* = \bar \zb$ otherwise. Compute $c_j = \sum_{i=1}^n \one_{\{j\}}\left(z_i^*\right)$ and $\bar y_j = c_j^{-1} \sum_{i=1}^n  y_i \one_{\{j\}}\left(z_i^*\right)$ for $j=1,2$. Then independently generate:
		\begin{enumerate}
			\item $p$ from Beta($c_1+1, c_2+1$)
			\item $\mu_j$ from $\displaystyle \N\left(\frac{c_j}{c_j+1} \bar y_j, \frac{\tau^2}{c_j+1}\right)$ for $j = 1, 2.$
			
		\end{enumerate}
	\end{enumerate}
	\caption{The \fs (FS) Gibbs Sampling algorithm (in the DA form)}
	\label{algo_fs_gs}  
	
\end{algorithm}

Similar to the MDA case, the spectrum of the Markov operator associated with the $\thetab$ sub-chain of an FS Markov chain can be studied through that of the Markov operator $\tilde K^*$ corresponding to the $\zb$ sub-chain. From the DA representation described in the previous paragraph, it follows that the Markov transition density associated with $\tilde K^*$ can be written as
\begin{equation} \label{mtd_fsa}
\tilde k^*(\zb, \zb') = \int_\X \pit(\zb' \mid \thetab, \yb) \: \pit(\thetab \mid \zb, \yb) \: d\thetab.
\end{equation} 
Owing to the above representation and the facts that $\pit(\zb \mid \thetab, \yb)$ is available in closed form, and $\pit(\thetab \mid \zb, \yb)$ is easy to sample from, one can use the MCRMA method to estimate the spectrum of $\tilde K^*$.

\subsubsection{Simulation Study}

To illustrate the performance of the MCRMA method in estimating the spectra of $K^*$ and $\tilde K^*$ (the MDA and FS Markov operators respectively), we consider a simulated dataset with sample size $n = 20$ from the mixture density $(\ref{2_comp_normmix})$, with $\mu_1 = 0$, $\mu_2 = 0.1$, $p = 0.5$ and fixed $\tau = 0.1$. Then, with $k$-means estimates taken as the starting values, we separately generate 10,000 realizations of MDA and FS Markov chains  after discarding first 20,000 realizations as burn-in from each chain. Then we extract the $\zb$ sub-chains from the two MCMC samples and use them in the MCRMA method to estimate their spectra. Note that the  latent space $\Z$ in both algorithms consist of $2^{20} = 1048576$ states, which means, each of the associated Markov operators corresponds to a  $1048576 \times 1048576$ matrix of transition probabilities. Hence, in order to find the true eigenvalues, one needs to compute the eigenvalues of $1048576 \times 1048576$ matrices, which is practically infeasible even though the state space is finite. However, the MCRMA method can still be applied here to provide estimates, as we discuss in the following.

For each of the two Markov chains, we run 10 separate instances of MCRMA, with number of Markov chain iterations $m = 1000, 2000, \cdots, 10,000$, and  Monte Carlo sample size $N = 5000$, to estimate the eigenvalues, and then create plots similar to Figure~\ref{psw_plots}. Note that, because the latent state space $\Z$ is finite, strong consitency of the MCRMA estimator is automatically ensured, and no relationship between the rate of growth of $N$ and $m$ is required. For each of the two chains, and for each of the 10 MCRMA instances, we record the first 21 estimated eigenvalues  (including the trivial eigenvalue $\lambda_0 = 1$) and plot them in Figure~\ref{n_20_mixnorm_plots}. Figure~\ref{n_20_mixnorm_all} shows all 21 eigenvalues obtained from each of the 10 MCRMA instances and for each Markov chain, plotted as 20 curves. The second, third, fourth, fifth and sixth largest estimated eigenvalues, viewed  as functions of the MCRMA iteration size $m$, are shown in Figures~\ref{n_20_mixnorm_lambda_1} through \ref{n_20_mixnorm_lambda_5}. From these plots, it appears that the MCRMA estimates for the MDA chain show some instability. Most of these estimates eventually stabilize, but it is interesting to note that the behavior of FS spectrum estimates is much more stable than the corresponding MDA spectrum estimates, even for smaller $m$'s. This is due to the fact that the FS chain is better mixing than the MDA chain, which in turn, is a consequence of the theoretically proven fact that the true spectrum of the MDA chain dominates that of the FS chain (see Section~\ref{sec_fs_algo}). As clearly displayed by the plots, the MCRMA estimates also exhibit this dominance, and provides us a visual idea of  the gains achieved in the FS algorithm in terms of convergence and mixing. 

\begin{figure}[htpb]
	\centering	
	\subfloat[The largest 21 eigenvalues of the MDA and FS chains. There are 10 curves for each Markov chain, each corresponding to the choices $m = 1000, \cdots, 10,000$ in the MCRMA algorithm.]{\includegraphics[height= 2.1in, width = 2.1in]{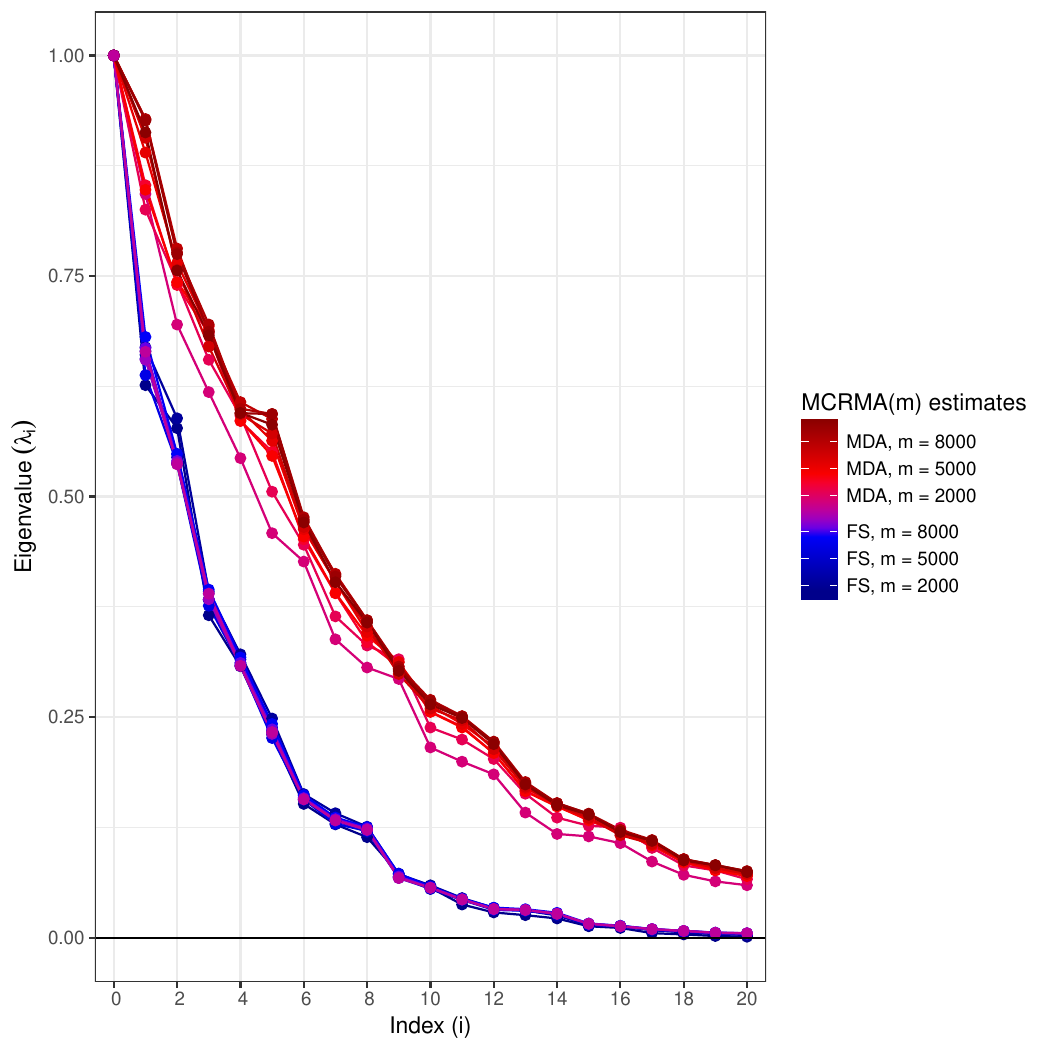} \label{n_20_mixnorm_all}}
	\qquad
	\subfloat[Second largest eigenvalues of the MDA and FS algorithm as functions of iterations $m$ in the MCRMA algorithm.]{\includegraphics[height= 2.1in, width = 2.1in]{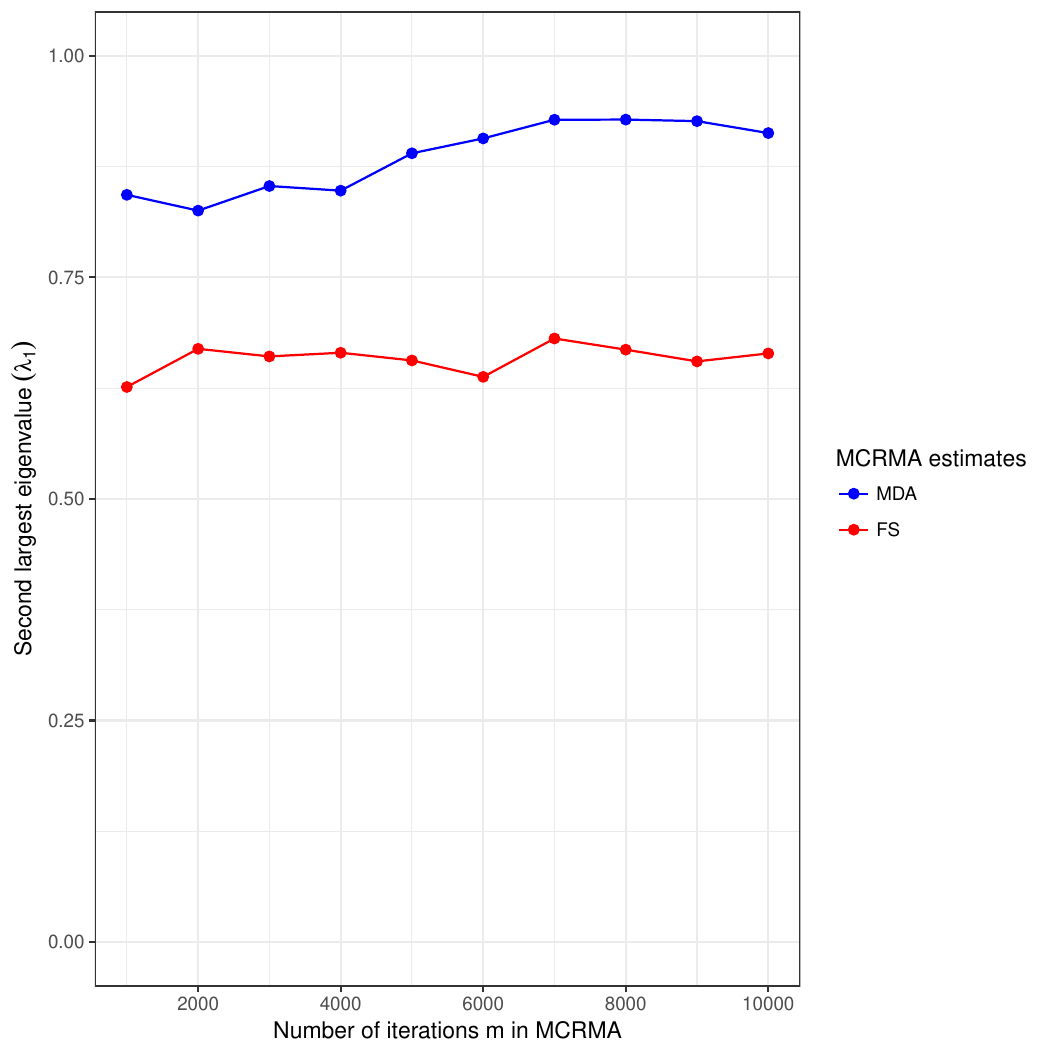} \label{n_20_mixnorm_lambda_1} }\\
	\subfloat[Third largest eigenvalues of the MDA and FS algorithm as functions of iterations $m$ in the MCRMA algorithm.]{\includegraphics[height= 2.1in, width = 2.1in]{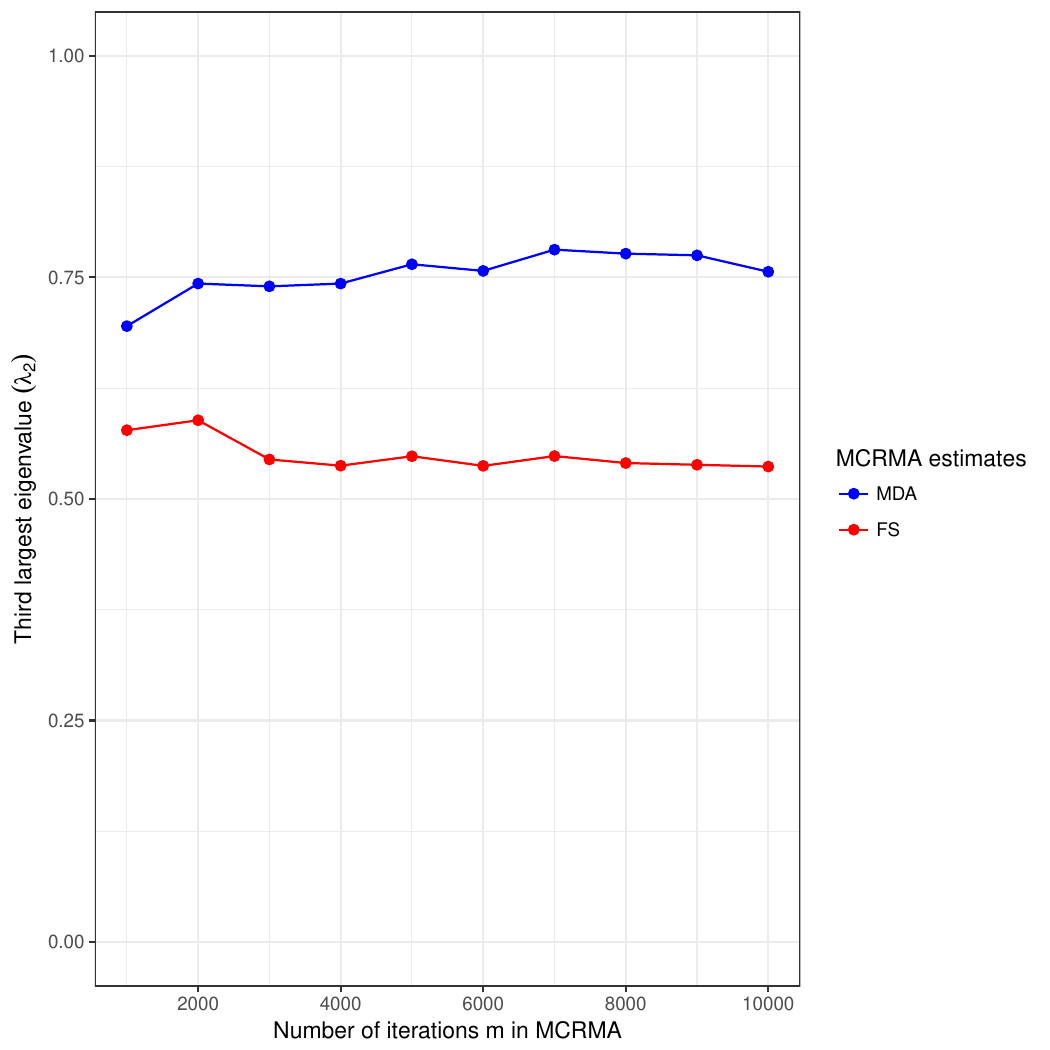} \label{n_20_mixnorm_lambda_2}}
	\qquad
	\subfloat[Fourth largest eigenvalues of the MDA and FS algorithm as functions of iterations $m$ in the MCRMA algorithm.]{\includegraphics[height= 2.1in, width = 2.1in]{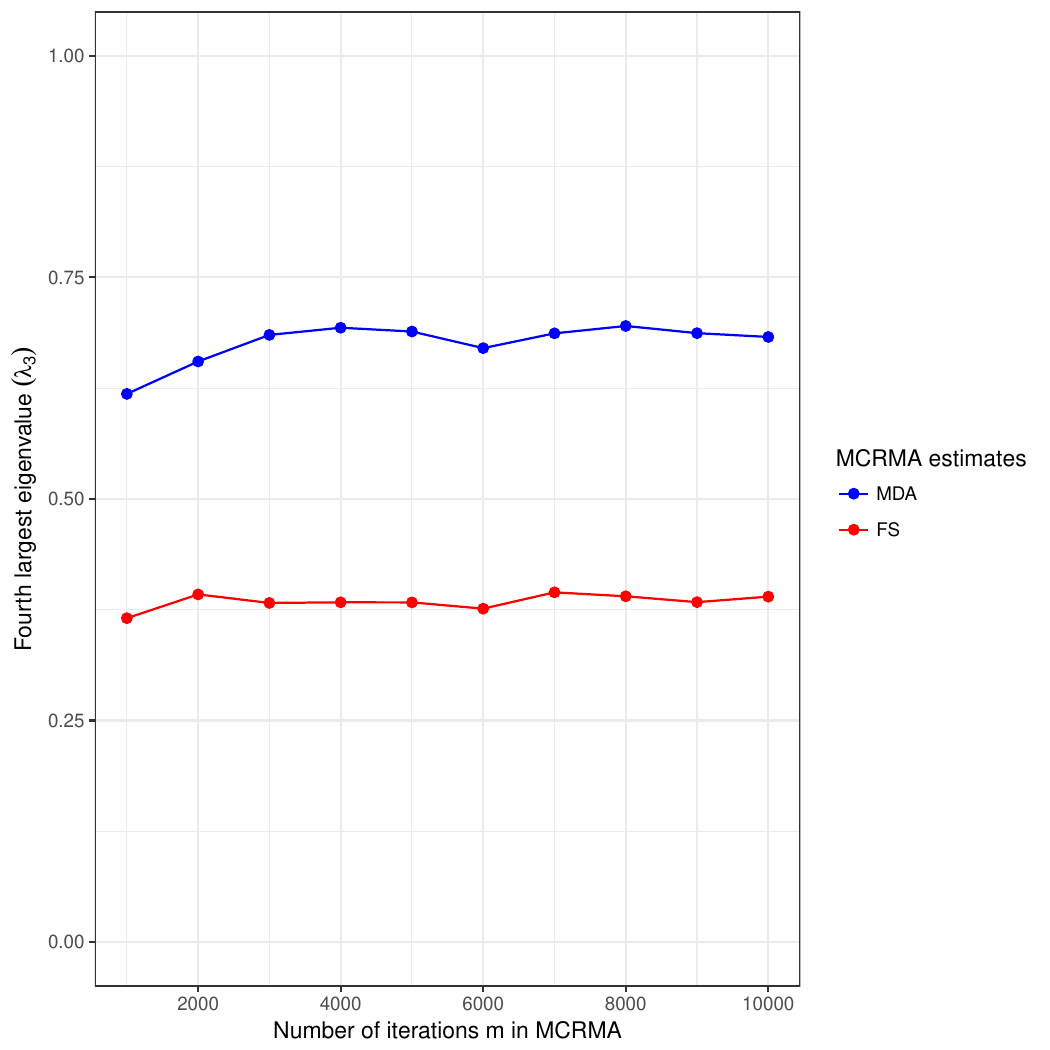} \label{n_20_mixnorm_lambda_3}} \\ 
	\subfloat[Fifth largest eigenvalues of the MDA and FS algorithm as functions of iterations $m$ in the MCRMA algorithm.]{\includegraphics[height= 2.1in, width = 2.1in]{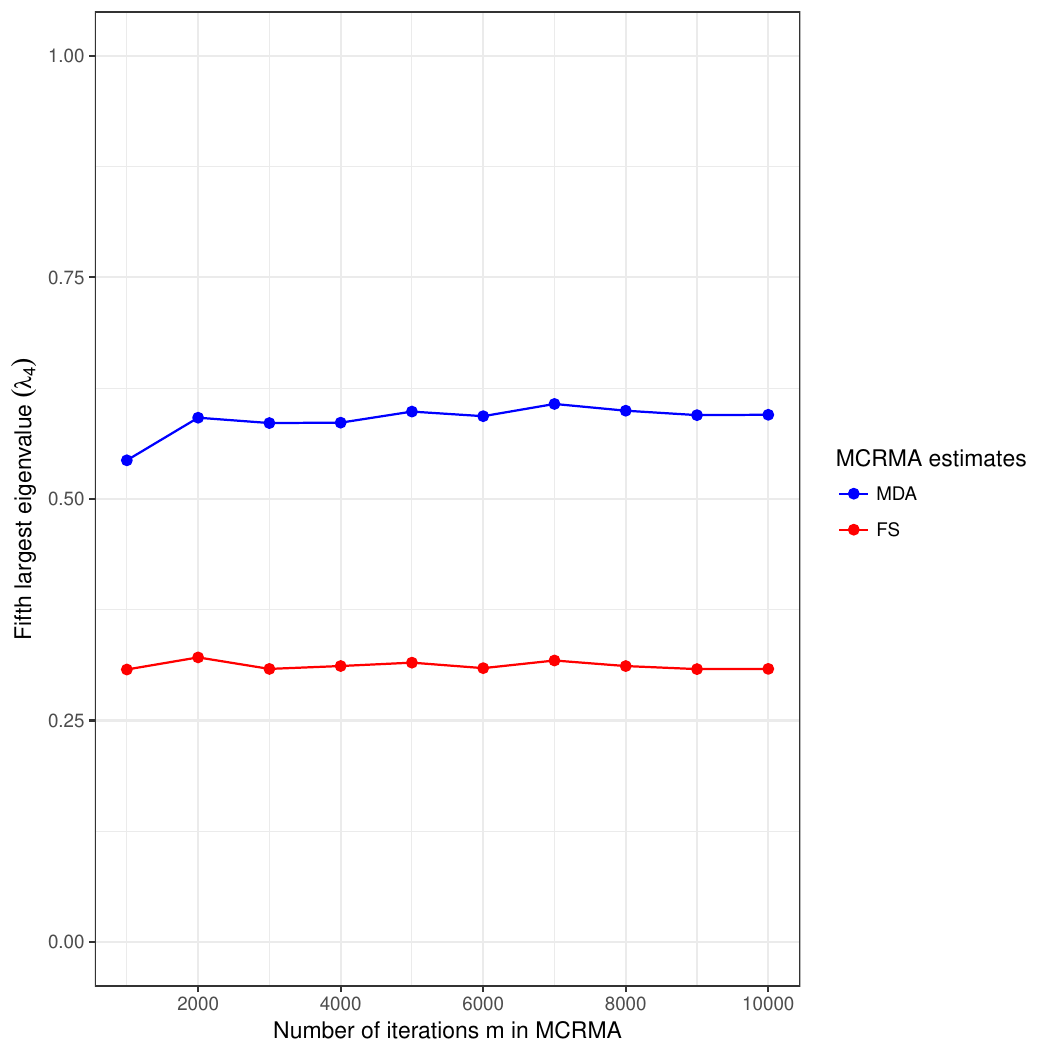} \label{n_20_mixnorm_lambda_4}}
	\qquad
	\subfloat[Sixth largest eigenvalues of the MDA and FS algorithm as functions of iterations $m$ in the MCRMA algorithm.]{\includegraphics[height= 2.1in, width = 2.1in]{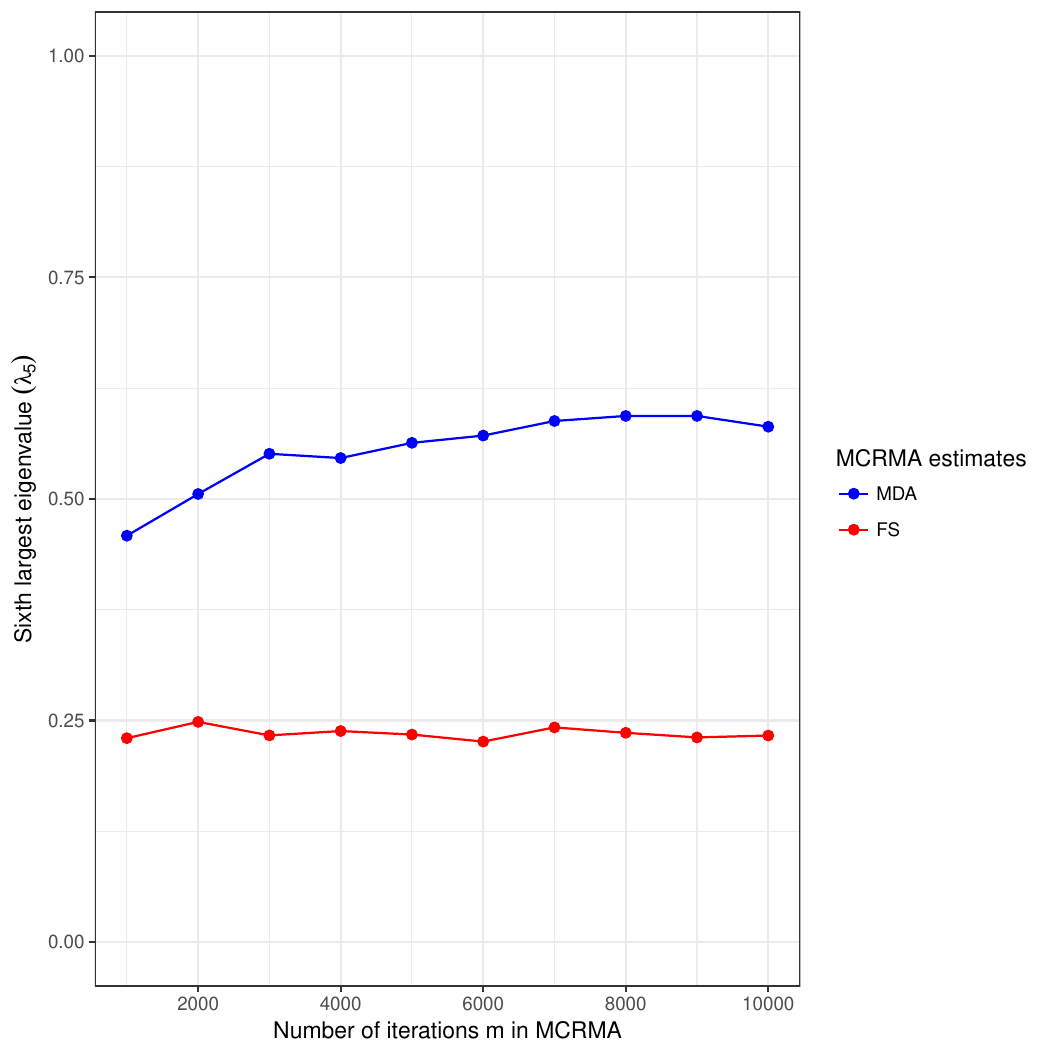} \label{n_20_mixnorm_lambda_5}}
	\\
	\caption{Eigenvalue estimates for the MDA and FS Markov chains using the MCRMA algorithm.}
	\label{n_20_mixnorm_plots}
\end{figure}

\begin{remark}
	It should be noted that the performance of MCRMA can be poor when the state space of the Markov chain is \emph{finite, but extremely large}, especially if the chain is poorly mixing. Although the spectrum estimates are guaranteed to converge to the truth for any Markov chain with a finite state space, in practice however, the value of $m$ required for a reasonable approximation can be too large to handle (recall that we need to find the eigenvalues of an $m \times m$ matrix to obtain the eigenvalue estimates).  In our case, we tried running the MCRMA algorithm for MDA and FS chains with $n = 30$ (more than a billion states), but the estimates did not show enough signs of convergence with $m \leq 10,000$. 
\end{remark}

\section{Discussion} \label{sec_discuss}

\noindent
As stated in the introduction, while bounding or estimating the spectral gap (or equivalently, the second largest eigenvalue) has received a lot of attention over the past three decades, very few methods have been proposed for accurately estimating the entire spectrum of Markov chains arising in modern applications. Building on the work of \citet{koltchinskii:gine:2000}, \citet{adamczak:bednorz:2015} develop an elegant method to estimate the spectrum of a trace class 
Markov operator using random matrix approximations. However, this method requires closed form expressions for the Markov transition density (and the stationary density), which is often unavailable in practice. We consider the general class of Markov chains arising from trace class 
Data Augmentation algorithms, where the transition density can typically only be expressed as an intractable integral. We develop a Monte Carlo based random matrix approximation method to consistently estimate the entire spectrum of the corresponding DA Markov operators. 

The particular integral form of the DA transition density in (\ref{mtd_da_general}) was critical in the development of our method. This form enables us to provide Monte Carlo based approximations for the intractable Markov transition density. We are able to show in Theorem \ref{thm_mcrma_consist} that the eigenvalues of the subsequently constructed random matrix still consistently estimates the desired spectrum. Methods to approximate general intractable transition densities, which may not necessarily have the integral form in (\ref{mtd_da_general}), have been proposed in the literature, see for example \citet{athreya:atuncar:1998}. The next obvious question in this line of research is: if the intractable 
transition densities appearing in the random matrix approximation of \citet{adamczak:bednorz:2015} are replaced by approximations based on these methods, does that still lead to consistent estimates of the desired spectrum? This is a challenging question, and will be investigated in future research.

\bibliographystyle{plainnat}
\bibliography{eigvalest_citations}

\pagebreak

\section*{Appendix}
\renewcommand{\theequation}{A\arabic{equation}}

This Appendix provides proofs of theorems and lemmas introduced in the original text. The referred equations in this Appendix are labeled as (A1), (A2) etc., whereas labels such as (1), (2), etc. refer to equations from the main text. A proof of an assertion from the main text ends with a $\square$, whereas a $\blacksquare$ marks the end of a proof of an assertion introduced in this Appendix. Proofs are organized by the sections in which they appear in the main text.

\begin{appendix}

\section{Proofs for  Section~\ref{sec_rma_exact}} \label{supp_sec_rma_exact}

\begin{proof}[Proof of Lemma~\ref{lemma_eqv}]
    Note that $h$ is necessarily non-negative and measurable and  symmetric in its arguments (since $K$ is self-adjoint). We prove the two implications $(i) \implies (ii)$ and $(ii) \implies (i)$ separately.
    
    \subsubsection*{$(i) \implies (ii)$:}
    Let there exist $F : \X \rightarrow \R$ such that $\pi F^2 < \infty$ and $h(x,x')  \leq F(x)F(x')$ for all $x, x' \in \X$. This means $h(x, x) \leq F(x)^2$ for all $x \in \X$. Therefore, 
    \[
    \int_{\X} \frac{k(x,x)}{\pi(x)}\:d\Pi(x) =\int_{\X} h(x,x)\:d\Pi(x)  \leq \int_{\X} F(x)^2 \:d\Pi(x) = \pi F^2 < \infty,
    \] 
    which, from (\ref{tr_cond}), implies that $K$ is trace class.
    
    \subsubsection*{$(ii) \implies (i)$:}
    Let $K$ be trace class. This means $K$ is also Hilbert Schmidt, and therefore from (\ref{hs_cond}),
    \[
    \int_{\X} \int_{\X} \: h(x, x')^2  \: d\Pi(x) \:d\Pi(x') = \int_{\X} \int_{\X} \: \left[\frac{k(x, x')}{\pi(x')}\right]^2  \: d\Pi(x) \:d\Pi(x')   < \infty.
    \]
    We shall prove the existence of $F$ by construction. Let us first denote by $(\lambda_m)_{m = 0}^\infty$  the sequence of eigenvalues of $K$. Then by the trace class property of $K$,  $\sum_{m = 0}^\infty \lambda_m = \tr K < \infty$,  and  by spectral theorem (see, e.g. \citet{jorgens:1982}), $h(x, x') = \sum_{m = 0}^\infty \lambda_m \varphi_m(x) \varphi_m(x')$ for all $x, x' \in \X$, where $(\varphi_m)_{m = 0}^\infty$ is an orthonormal basis of $L^2(\pi)$. Hence, for all $x, x' \in \X$
    \[
    h(x, x') = \sum_{m = 0}^\infty \left(\sqrt{\lambda_m} \varphi_m(x)\right) \: \left(\sqrt{\lambda_m} \varphi_m(x')\right) \leq \sqrt{ \sum_{m = 0}^\infty \lambda_m \varphi_m(x)^2} \: \sqrt{ \sum_{m = 0}^\infty \lambda_m \varphi_m(x')^2} = F(x) F(x'), 
    \]
    where $F := \sqrt{\sum_{m = 0}^\infty \lambda_m \varphi_m^2}$ and the inequality is due to Cauchy-Schwarz. The proof is completed by noticing that $\pi F^2 = \pi( \sum_{m = 0}^\infty \lambda_m \varphi_m^2) = \sum_{m = 0}^\infty \lambda_m \: \pi \varphi_m^2  = \sum_{m = 0}^\infty \lambda_m < \infty$.
\end{proof}

\section{Proofs for  Section~\ref{sec_mcrma}} \label{supp_sec_mcrma}

\begin{proof}[Proof of Theorem~\ref{thm_mcrma_consist}]
    On the outset, note that, by the triangle inequality and then (\ref{hoff_wiel_ineq}), we have 
    \begin{align} \label{del2_upper}
    \delta_2\left(\sp(\Hmn), \sp(K)\right)  \leq \left\|\Hmn - H_m \right\|_\hs + \delta_2\left(\sp(H_m), \sp(K)\right)
    \end{align}
    Since $\delta_2\left(\sp(H_m), \sp(K)\right) \rightarrow 0$ almost surely (Theorem~\ref{thm_rma_consis}) as $m \rightarrow \infty$, therefore, we only need to show the almost sure or in probability convergence (to zero) of the first term on the right hand side of the above inequality, as $N \rightarrow \infty$ and $m \rightarrow \infty$. We shall prove this convergence separately for the cases where $\X$ is finite and infinite.
    
    \subsubsection*{$\ref{case_finite}:$ $\X$ is finite}
    Note that, for any matrix $A = (a_{ij})\in \R^{m \times m}$,  $\|A\|_{\hs} = \left(\sum_i \sum_j a_{ij}^2\right)^{1/2} \leq m \max_{i,j} |a_{ij}|$. Therefore,
    \begin{align*}
    &\sup_{m \geq 1} \left\|\Hmn - H_m \right\|_\hs \\
    \leq & \sup_{m \geq 1} \: \max_{0 \leq j \neq j' \leq m-1} \frac{\left|\hat{k}_N(X_j, X_{j'}) - k(X_j, X_{j'})\right|}{\pi(X_{j'})} \\
    \leq & \max_{x, x' \in \X} \frac{\left|\hat{k}_N(x, x') - k(x, x')\right|}{\pi(x')} 
    \leq \sum_{x \in \X} \sum_{x' \in \X} \frac{\left|\hat{k}_N(x, x') - k(x, x') \right|}{\pi(x')} \rightarrow 0 \text{, almost surely} 
    \numbereqn \label{finite_hhat_minus_h}
    \end{align*}
    as $N \rightarrow \infty$, since the (double) summation includes a finite number of terms, with each term converging almost surely to zero (Monte Carlo convergence). Thus, combining (\ref{del2_upper}), (\ref{finite_hhat_minus_h}) and Theorem~\ref{thm_rma_consis}, we have
    \[
    \delta_2\left(\sp\left(\Hmn\right), \sp(K)\right) \rightarrow 0 \text{ almost surely as } m \rightarrow \infty \text{ and } N \rightarrow \infty.
    \]

    \subsubsection*{$\ref{case_infinite}:$ $\X$ is infinite:}
    By the variance condition~\ref{var_condn}, we have,
    \[
    M := \sup_{m \geq 1} \: \max_{0 \leq j < j' \leq m-1} \int_{\X} \int_{\X} \int_{\Z}   \left( \frac{ f_{X \mid Z}\left( x_{j'}\mid z \right) }{\pi(x_{j'})} \right)^2  f_{Z \mid X}\left(z \mid x_{j} \right) q_{jj'}(x_j, x_{j'}) \:d\zeta(z) \:d\nu(x_j)\: d\nu(x_{j'}) < \infty.
    \]
    
    \noindent Let $\epsilon > 0$ be arbitrary. Since the Hilbert Schmidt norm for matrices are the same as the Frobenius norm, and since $\hat{h}_N(\cdot, \cdot)$ $h(\cdot, \cdot)$ are both symmetric in its arguments, therefore, we have
    \begin{align*}
    \left\|\Hmn - H_m \right\|_\hs ^2 &= \frac{2}{m^2} \underset{0 \leq j < j' \leq m-1}{\sum \sum} \left(\hat{h}_N(X_j, X_{j'}) - h(X_j, X_{j'}) \right)^2 \\
    & = \frac{2}{m^2} \underset{0 \leq j < j' \leq m-1}{\sum \sum} \left( \frac{\hat{k}_N(X_j, X_{j'}) - k(X_j, X_{j'})}{\pi(X_{j'})} \right)^2.
    \end{align*}
    
    \noindent Therefore, by Markov inequality,
    \begin{align*}
    & \quad P\left(\left\|\Hmn - H_m \right\|_\hs > \epsilon \right) \\
    & \leq \cfrac{\E \left\|\Hmn - H_m \right\|_\hs ^2}{\epsilon^2} \\
    &= \frac{2}{m^2 \: \epsilon^2} \underset{0 \leq j < j' \leq m-1}{\sum \sum} \E \left( \frac{\hat{k}_N(X_j, X_{j'}) - k(X_j, X_{j'})}{\pi(X_{j'})} \right)^2 \\
    &= \frac{2}{m^2 \: \epsilon^2} \underset{0 \leq j < j' \leq m-1}{\sum \sum} \E \left\lbrace  \frac{\E  \left[ \left. \left( \hat{k}_N(X_j, X_{j'}) - k(X_j, X_{j'}) \right)^2  \right|  \Phi_m \right]}{\pi(X_{j'})^2} \right\rbrace  \\
    &= \frac{2}{m^2 \: \epsilon^2} \underset{0 \leq j < j' \leq m-1}{\sum \sum} \E \left\lbrace  \frac{\var  \left( \left.  \hat{k}_N(X_j, X_{j'}) \right|  \Phi_m \right)}{\pi(X_{j'})^2} \right\rbrace  \\
    &= \frac{2}{m^2 \: \epsilon^2} \underset{0 \leq j < j' \leq m-1}{\sum \sum} \E \left\lbrace  \frac{\var  \left( \left.  f_2\left( X_{j'}|Z^{(j)}_1 \right) \right|  \Phi_m \right)}{N\:\pi(X_{j'})^2} \right\rbrace  \\
    &\leq \frac{2}{N m^2 \: \epsilon^2  } \underset{0 \leq j < j' \leq m-1}{\sum \sum} \E \left\lbrace  \frac{\E  \left( \left.  f_2\left( X_{j'}|Z^{(j)}_1 \right)^2 \right|  \Phi_m \right)}{\pi(X_{j'})^2} \right\rbrace  \\
    &\leq \frac{2}{N m^2 \: \epsilon^2  } \: \frac{m(m-1)}{2} \: \max_{0 \leq j < j' \leq m-1} \E \left\lbrace  \frac{\E  \left( \left.  f_2\left( X_{j'}|Z^{(j)}_1 \right)^2 \right|  \Phi_m \right)}{\pi(X_{j'})^2} \right\rbrace  \\
    &\leq \frac{1}{N \: \epsilon^2} \: \sup_{m \geq 1} \: \max_{0 \leq j < j' \leq m-1} \int_{\X} \int_{\X} \int_{\Z}  \left( \frac{ f_{X \mid Z} \left( x_{j'}\mid z \right) }{\pi(x_{j'})} \right)^2 \\
    & \qquad \qquad \qquad \qquad \qquad \qquad \qquad \qquad  \qquad  f_{Z \mid X} \left(z \mid x_{j} \right) q_{jj'}(x_j, x_{j'}) \:d\zeta(z) \:d\nu(x_j)\: d\nu(x_{j'}) \\
    &= \frac{1}{N \: \epsilon^2}\: M = \frac{M}{\epsilon^2}\: \frac{1}{N(m)}.
    \end{align*}
    
    \noindent Therefore, 
    \begin{enumerate}[label = (\roman*)]
        \item if $\frac{1}{N(m)} \rightarrow 0$ as $m \rightarrow \infty$, then 
        \[
        \varlimsup_{m \rightarrow \infty} P\left(\left\|\Hmn - H_m \right\|_\hs  > \epsilon \right) \leq  \lim\limits_{m \rightarrow \infty} \frac{M}{\epsilon^2}\: \frac{1}{N(m)} = 0,
        \]
        since $M < \infty$. This means $\left\|\Hmn - H_m \right\|_\hs \xrightarrow{P} 0$, which, by (\ref{del2_upper}) implies 
        \[
        \delta_2\left(\sp(\Hmn), \sp(K)\right) \xrightarrow{P} 0 \text{ as } m \rightarrow \infty.
        \]
        
        \item if $\sum_{m=0}^\infty \frac{1}{N(m)} < \infty$ as $m \rightarrow \infty$, then
        \[
        \sum_{m=0}^\infty P\left(\left\|\Hmn - H_m \right\|_\hs  > \epsilon \right) \leq  \frac{M}{\epsilon^2}  \sum_{m=0}^\infty \frac{1}{N(m)} < \infty.
        \]
        which, by the Borel-Cantelli theorem, implies $\left\|\Hmn - H_m \right\|_\hs \rightarrow 0$ almost surely. Therefore, from (\ref{del2_upper}), 
        \[
        \delta_2\left(\sp(\Hmn), \sp(K)\right) \rightarrow 0 \text{ almost surely,  as } m \rightarrow \infty.
        \]
    \end{enumerate}
    
\end{proof}

\begin{proof}[Proof of Theorem~\ref{thm_mcrma_ext_consist}]
    Let us first define $\hat c_m = 1/ \spmax\left( \Smn \right)$ and $\Hmnc = \hat c_m \Smn$. Then 
    \[
    \sp\left(\Hmnc \right) = \frac{\sp\left( \Smn \right)}{\spmax\left( \Smn \right)},
    \]
    which means it will be enough to prove the convergence of $\delta_2\left(\sp\left( \Hmnc \right), \sp(K) \right)$ to zero. 
    
    \noindent As in (\ref{Hmnhat}), define $\hat h_N(X_j, X_{j'}) = \hat k_N \left( X_{j}, X_{j'} \right) / \pi(X_{j'})$ for $j < j'$ and set  $\hat{h}_N(X_j, X_{j'}) = \hat{h}_N(X_{j'}, X_{j})$ for $j > j'$, and construct the matrix 
    \[
    \Hmn = \frac{1}{m} \left((1-\delta_{jj'})\: \hat h_N(X_j, X_{j'})\right)_{0 \leq j,  j' \leq m-1}.
    \]
    Then, because $\pi(\cdot) = \eta(\cdot)/c$, we have $\Hmn = c \: \Smn$. This implies
    \[
    \hat c_m = \frac{1}{\spmax\left( \Smn \right)} = \frac{c}{\spmax\left(c\: \Smn \right)} = \frac{c}{\spmax \left(c\: \Hmn \right)}.
    \]
    Because $\max \sp \left(c\: \Hmn \right) \rightarrow \max \sp (K) = 1$ (applying continuous mapping theorem on the results of Theorem~\ref{thm_mcrma_consist}), therefore $\hat c_m \rightarrow c$, where the convergences are in almost sure sense under cases \ref{case_finite_ext} and \ref{case_infinite_ext}\ref{mcrma_ext_consist_cond_ii}, and in probability under condition \ref{case_infinite_ext}\ref{mcrma_ext_consist_cond_i}.  Now, by triangle inequality, 
    \[
    \delta_2\left(\sp\left( \Hmnc \right), \sp(K) \right) \leq \delta_2\left( \sp\left( \Hmnc \right), \sp\left( \Hmn \right) \right) + \delta_2\left(\sp\left( \Hmn \right), \sp(K) \right)
    \]
    Because $\delta_2\left(\sp\left( \Hmn \right), \sp(K) \right) \rightarrow 0$ [by Theorem~\ref{thm_mcrma_consist}; almost surely under condition \ref{case_finite_ext} and \ref{case_infinite_ext}\ref{mcrma_ext_consist_cond_ii}, and in probability under condition \ref{case_infinite_ext}\ref{mcrma_ext_consist_cond_i}], so we only need to show the convergence of  
    \[
    \delta_2\left( \sp\left( \Hmnc \right), \sp\left( \Hmn \right) \right).
    \] 
    Observe that
    \begin{align*}
    0 \leq \delta_2\left( \sp\left( \Hmnc \right), \sp\left( \Hmn \right) \right) &= \delta_2\left( \hat c_m \sp\left( \Smn \right), c\: \sp\left( \Smn \right) \right) \\
    & = \left\| \left(\hat c_m \sp\left( \Smn \right)\right)^{\uparrow\downarrow} - \left( c \sp\left( \Smn \right)\right)^{\uparrow\downarrow} \right\|_{\ell_2} \;(\text{from } (\ref{delta2_l2})) \\
    &= |\hat c_m - c| \: \left\| \sp\left( \Smn \right)^{\uparrow\downarrow} \right\|_{\ell_2} \\
    &= \frac{|\hat c_m - c|}{c} \: \delta_2\left( \sp\left( \Hmn \right), 0 \right) \\
    &\stackrel{(\star)}{\leq} \frac{|\hat c_m - c|}{c}\:  \left[\delta_2\left( \sp\left( \Hmn \right), \sp(K) \right) +  \delta_2\left( \sp(K), 0 \right) \right] \\
    &\stackrel{(\dagger)}{\leq}  \frac{|\hat c_m - c|}{c}\:  \left[\delta_2\left( \sp\left( \Hmn \right), \sp(K) \right) +  \left\| K \right\|_{\hs} \right]. \numbereqn \label{del2_hdiff_upper}
    \end{align*} 
    where $(\star)$ follows from the triangle inequality, and $(\dagger)$ follows from (\ref{hoff_wiel_ineq}). Since $K$ is trace class and hence Hilbert-Schmidt, $\left\| K \right\|_{\hs} < \infty$, and hence by Theorem~\ref{thm_mcrma_consist}, the  sum within the square brackets on the right hand side of (\ref{del2_hdiff_upper}) converges to $\left\| K \right\|_{\hs} < \infty$ as $m \rightarrow \infty$ [in almost sure sense under cases \ref{case_finite_ext} and \ref{case_infinite_ext}\ref{mcrma_ext_consist_cond_ii}, and in probability under condition \ref{case_infinite_ext}\ref{mcrma_ext_consist_cond_i}]. This, together with the convergence of $\hat c_m$ completes the proof. 
\end{proof}

\section{Proofs for  Section~\ref{sec_illus}} \label{supp_sec_illus}
\begin{proof}[Proof of Theorem~\ref{tr_thm_psw}]
    The Markov transition density of $\Phi$ is given by
    \begin{equation} \label{mtd_psw}
    k(\betab, \betab') = \int_{\R_+^n} \pi(\betab'|\wb, \yb) \:\pi(\wb \mid\betab, \yb) \:d\wb 
    \end{equation}
    where $\pi(\betab \mid \wb, \yb)$ and $\pi(\wb \mid\betab, \yb)$ are as given in (\ref{pi_beta_given_w}) and (\ref{pi_w_given_beta}) respectively. Our objective is to establish (\ref{tr_cond}), i.e., to prove
    \begin{equation} \label{tr_cond_psw}
    I = \int_{\R^p} k(\betab, \betab)\: d\betab = \int_{\R^p} \int_{\R_+^n} \pi(\betab \mid \wb, \yb) \:\pi(\wb \mid\betab, \yb) \:d\betab < \infty.
    \end{equation}
    
    \noindent From (\ref{pi_beta_given_w}) it follows that 
    \begin{align*}
    \pi(\betab \mid \wb, \yb) &= \left({2\pi}\right)^{-p/2} \: \left|U^T \Omega(\wb) U + B^{-1}\right|^{1/2} \: \exp \left[-\frac12 (\betab - \Sigma(\wb) \mub)^T \Sigma(\wb)^{-1}(\betab - \Sigma(\wb) \mub) \right] \\
    & \stackrel{(\star)}{\leq} C_1 \: \left(\sum_{i=1}^{n}w_i + 1\right)^{p/2} \: \exp\left[-\frac12 \left( \betab^T \Sigma(\wb)^{-1} \betab - 2\betab^T\mub + \mub^T \Sigma(\wb) \mub \right) \right] \\
    & \stackrel{(\star \star)}{\leq} C_1 \: \left(\sum_{i=1}^{n}w_i + 1\right)^{p/2}\: \exp\left(-\frac12 \betab^T B^{-1} \betab \right) \: \exp(\betab^T\mub)  \numbereqn \label{pi_beta_upper}
    \end{align*}
    where $C_1$ is a  constant,  $(\star)$ follows from Proposition~\ref{prop_det_ineq}, and $(\star \star)$ follows from the facts that $\Sigma(\wb)^{-1} = \left(U^T \Omega(\wb) U + B^{-1}\right) \geq B^{-1}$ and $\mub^T \Sigma(\wb) \mub \geq 0$.\\
    
    \noindent Again, from (\ref{pi_w_given_beta}) we get 
    \begin{align*}
    \pi(\wb \mid\betab, \yb) &= \prod_{i=1}^{n} \left\lbrace  \cosh\left(\frac{|\ub_i^T\betab \mid }{2}\right) \: \exp\left[-\frac12 (\ub_i^T\betab)^2 w_i\right] \: \gt(w_i) \right\rbrace \\
    & \leq \prod_{i=1}^{n} \left\lbrace  \cosh\left(\frac{|\ub_i^T\betab \mid }{2}\right) \: \gt(w_i)\right\rbrace  \\
    & \leq \prod_{i=1}^{n} \left\lbrace  \exp\left(\frac{|\ub_i^T\betab \mid }{2}\right) \: \gt(w_i)\right\rbrace  = \exp\left(\frac12 \sum_{i=1}^n \left|\ub_i^T\betab \right| \right) \prod_{i=1}^{n} \gt(w_i)
    \numbereqn \label{pi_w_upper}
    \end{align*} 
    where the last inequality follows from the fact that $\cosh(u) = \frac12 (e^u + e^{-u}) \leq e^u$, if $u \geq 0$.
    
    \noindent Therefore, from (\ref{mtd_psw}), (\ref{pi_beta_upper}) and (\ref{pi_w_upper}) we get
    \begin{align*}
    I
    &\leq \int_{\R^p} \int_{\R_+^n} C_1 \: \left(\sum_{i=1}^{n}w_i + 1\right)^{p/2}\: \exp\left(-\frac12 \betab^T B^{-1} \betab \right) \: \exp(\betab^T\mub)  \\
    &\qquad \qquad \:\exp\left(\frac12 \sum_{i=1}^n \left|\ub_i^T\betab \right| \right) \prod_{i=1}^{n} \gt(w_i) \:d\wb \:d\betab \\
    &= C_1 \int_{\R^p} \exp\left(-\frac12 \betab^T B^{-1} \betab \right) \: \exp(\betab^T\mub)  \:\exp\left(\frac12 \sum_{i=1}^n \left|\ub_i^T\betab \right| \right) \:d\betab \\
    &\qquad \qquad \times \int_{\R_+^n} \left(\sum_{i=1}^{n}w_i + 1\right)^{p/2} \prod_{i=1}^{n} \gt(w_i) \:d\wb\\
    &= C_1 \: I_1 \: I_2, \text{ say},
    \end{align*}
    where 
    \begin{align}
    I_1 &= \int_{\R^p} \exp\left(-\frac12 \: \betab^T B^{-1} \betab + \betab^T\mub + \frac12 \sum_{i=1}^n\left|\ub_i^T\betab\right| \right) \:d\betab \label{I_1_trace} \\
    \text{and } I_2 &= \int_{\R_+^n} \left(\sum_{i=1}^{n}w_i + 1\right)^{p/2} \prod_{i=1}^{n} \gt(w_i) \:d\wb. \label{I_2_trace}
    \end{align}
    
    \noindent Thus, to prove that $K$ is trace class, it is enough to show that $I_1 < \infty$ and $I_2 < \infty$. Now, from Proposition~\ref{prop_I1_intgrnd_upper}, it follows that the integrand of $I_1$ can be bounded above by 
    \[
    C_2 \: \exp\left[-\frac14 \: \betab^T B^{-1} \betab \right].
    \]
    for an appropriately chosen constant $C_2$. Note that the above upper bound is a constant (only dependent of $\yb$) multiple of a  multivariate normal density, and hence, integrable, which implies $I_1$ is finite. \\
    
    \noindent As for $I_2$, because existence of higher moments ensures that of lower moments, it will be enough to show that 
    \begin{equation} \label{I_3_trace}
    I_3 = \int_{\R_+^n} \left(\sum_{i=1}^{n}w_i + 1\right)^{p} \prod_{i=1}^{n} \gt(w_i) \:d\wb < \infty. 
    \end{equation}
    First observe that 
    \[
    \left(\sum_{i=1}^{n}w_i + 1\right)^{p} \leq (n+1)^{p} \left(\sum_{i=1}^{n}w_i^{p} + 1\right).
    \]
    Therefore
    \[
    I_3 \leq \int_{\R_+^n} (n+1)^{p} \left(\sum_{i=1}^{n}w_i^{p} + 1\right) \prod_{i=1}^{n} \gt(w_i) \:d\wb = (n+1)^{p} \left[\left(\sum_{i=1}^{n} \int_{\R_+} w_i^{p}\:\gt(w_i)\:dw_i\right) + 1\right],
    \]
    which means in order to prove $I_3 < \infty$, it is enough to show that
    \[
    \E\left(W^{p}\right) = \int_{\R_+} w^{p} \: \gt(w_i) \:dw < \infty,
    \]
    where $W$ is a random variable with density $\gt$. Now using the representation (\ref{w_from_E}), we can write  
    \[
    W = \frac{2}{\pi^2} \sum_{l=1}^\infty \frac{E_l}{(2l-1)^2}
    \]
    where $E_l, l \geq 1$ are IID standard exponential random variables.  Therefore,
    \begin{align*}
    \E\left(W^{p}\right) &= \left( \frac{2}{\pi^2}\right)^p \E\left(\sum_{l=1}^\infty \frac{E_l}{(2l-1)^2}\right)^p \\
    &= \left( \frac{2}{\pi^2}\right)^p \E\left(\sum_{l_1=1}^\infty \cdots \sum_{l_p=1}^\infty \frac{E_{l_1}}{(2{l_1}-1)^2} \cdots \frac{E_{l_p}}{(2{l_p}-1)^2}   \right) \\
    &= \left( \frac{2}{\pi^2}\right)^p \sum_{l_1=1}^\infty \cdots \sum_{l_p=1}^\infty \frac{\E\left(E_{l_1} \cdots E_{l_p} \right) } {(2{l_1}-1)^2 \cdots (2{l_p}-1)^2}. \numbereqn \label{expon_sum}
    \end{align*} 
    Note that, in the right hand side of (\ref{expon_sum}), the term $\E\left(E_{l_1} \cdots E_{l_p} \right)$ provides the  expected value of product of at most $p$ distinct (and independent) exponential random variables $E_{l_1}, \dots, E_{l_p}$, and each random variable can be repeated at most $p$ times. Therefore,
    \begin{equation} \label{E_prod_upper}
    \E\left(E_{l_1} \cdots E_{l_p} \right) \leq \left(\E E_{l_1}^p\right) \cdots \left(\E E_{l_p}^p\right) = (p!)^p.
    \end{equation}
    Hence, from (\ref{expon_sum}) and (\ref{E_prod_upper})
    \[
    \E\left(W^{p}\right) \leq \left( \frac{2\:p!}{\pi^2}\right)^p \left(\sum_{l=1}^\infty \frac{1}{(2l-1)^2}\right)^p < \infty.
    \]
    Therefore, $I_3$ and hence $I_2$ is finite, which, together finiteness of $I_1$, implies that $I < \infty$. This completes the proof.
\end{proof}

\begin{proof}[Proof of Theorem~\ref{var_thm_psw}]
    
    Let $\Phi_m = \{\betab_0, \dots, \betab_{m-1} \}$ denote the first $m$ realizations of the \psw Markov chain $\Phi$. If $q_{j_1 j_2 \cdots j_k}$ denotes the joint density of $\betab_{j_1}, \dots, \betab_{j_k}$, $0 \leq j_1 < \cdots < j_k \leq m-1$, $k = 1, \cdots, m$, then we need to show that 
    \[
    \sup_{m \geq 1}\: \max_{0 \leq j < j' \leq m-1} \int_{\R^p} \int_{\R^p} \left[\int_{\R_+^n} \left(\cfrac{\pi(\betab \mid \wb, \yb)}{\pi(\betab \mid \yb)}\right)^2 \: \pi(\wb \mid\betat, \yb)\: d\wb \right] q_{jj'}(\betat, \betab)\: d\betat \: d\betab \label{var_int_psw} < \infty.
    \]
    
    \noindent We shall prove that for any $m \geq 1$ and any $0\leq j < j' \leq m$,  
    \begin{equation} 
    \It_{jj'} = \int_{\R^p} \int_{\R^p} \left[\int_{\R_+^n} \left(\cfrac{\pi(\betab \mid \wb, \yb)}{\pi(\betab \mid \yb)}\right)^2 \: \pi(\wb \mid\betat, \yb)\: d\wb \right] q_{jj'}(\betat, \betab)\: d\betat \: d\betab < C^* \label{I_jj'}
    \end{equation} 
    where $C^*$ is a finite constant free of $j, j'$, and $m$, and that will complete the proof. To this end, We first find an upper bound for the squared ratio $({\pi(\betab \mid \wb, \yb)}/{\pi(\betab \mid \yb)})^2$. Note that from (\ref{pi_beta_upper}), we get
    \begin{align*}
    \pi(\betab \mid \wb, \yb) \leq & \;C_1 \left(\sum_{i=1}^{n}w_i + 1\right)^{p/2}  \exp \left[\betab^T U^T \left(\yb - \frac12 \bm 1_n \right) \right]\\
    &\qquad \times \exp\left[-\frac12 \left(\betab^T B^{-1} \betab  - 2 \betab^T B^{-1} \bbo \right) \right]  \numbereqn \label{pi_beta_upper_2}  
    \end{align*}
    
    \noindent Now the prior density for $\betab$ is given by
    \begin{align*}
    \pi(\betab) &= (2\pi)^{-p/2} |B|^{-1/2} \exp\left[-\frac12 (\betab - \bbo)^T B^{-1} (\betab - \bbo)\right] \\
    &= \Ct_0 \exp\left[-\frac12 \left(\betab^T B^{-1} \betab - 2 \betab^T B^{-1} \bbo \ \right)\right] 
    \end{align*}
    where $\Ct_0 = (2\pi)^{-p/2} |B|^{-1/2}$.  Therefore, from (\ref{post_beta_inc}), the complete posterior density for $\betab$ is obtained as
    \begin{align*}
    \pi(\betab \mid \yb) &= \frac{\Ct_0}{c(\yb)}   \left[\prod_{i=1}^n {\left\lbrace F(\ub_i^T \betab) \right\rbrace}^{y_i} {\left\lbrace 1 - F(\ub_i^T \betab) \right\rbrace}^{1-y_i} \right]  \exp\left[-\frac12 \left(\betab^T B^{-1} \betab - 2 \betab^T B^{-1} \bbo \ \right)\right] \numbereqn \label{post_beta}
    \end{align*}

    \noindent Therefore, from (\ref{pi_beta_upper_2}) and (\ref{post_beta}), we get 
    \begin{align*}
    \left(\frac{\pi(\betab \mid \wb, \yb)}{\pi(\betab \mid \yb)}\right)^2 &\leq  \left(\frac{c(\yb) \: C_1}{\Ct_0}\right)^2  \left(\sum_{i=1}^{n}w_i + 1\right)^{p}  \exp \left[2\betab^T U^T \left(\yb - \frac12 \bm 1_n \right) \right]    \\
    & \qquad \qquad \qquad\times  \left[\prod_{i=1}^n {\left\lbrace F(\ub_i^T \betab) \right\rbrace}^{y_i} {\left\lbrace 1 - F(\ub_i^T \betab) \right\rbrace}^{1-y_i} \right]^{-2}  \\
    & \leq \Ct_1  \left(\sum_{i=1}^{n}w_i + 1\right)^{p} \exp\left(\frac12 \sum_{i=1}^n\left|\ut_i^T\betab\right| \right)
    \numbereqn \label{ratio2_upper_bd}
    \end{align*} 
    where $\Ct_1$ is a constant and the last inequality follows from Proposition~\ref{prop_exp_F_upper}.  In the following, we use the symbol  $q_{j' \mid  j}(\cdot, \cdot)$ to denote the conditional density for $\betab_{j'}$ given $\betab_j$, with $q_{j \mid  j}(\cdot, \betab)$ being the degenerate density associated with the point measure $\one_{\{\betab\}}(\cdot)$, and $q_{j}(\cdot)$ to denote the marginal density for $\betab_j$. Therefore, an upper bound for the integral $\It_{jj'}$ in (\ref{I_jj'}) is obtained as follows:

    \begin{align*}
    & \quad \Ct_1 \int_{\R^p} \int_{\R^p} \int_{\R_+^n}  \: \left(\sum_{i=1}^{n} w_i + 1\right)^{p} \exp\left(\frac12 \sum_{i=1}^n \left|\ut_i^T\betab \right| \right) \exp\left(\frac12 \sum_{i=1}^n \left|\ub_i^T\betat \right| \right) \\
    & \qquad \qquad \qquad \qquad \qquad \qquad \prod_{i=1}^{n} \gt(w_i) \: q_{jj'}(\betat, \betab) \: d\wb \: d\betat \: d\betab \\
    & \quad \Ct_1 \int_{\R^p} \int_{\R^p} \int_{\R_+^n}  \: \left(\sum_{i=1}^{n}w_i + 1\right)^{p} \exp\left(\frac12 \sum_{i=1}^n \left|\ut_i^T\betab \right| \right) \exp\left(\frac12 \sum_{i=1}^n \left|\ub_i^T\betat \right| \right) \\
    & \qquad \qquad \qquad \qquad \qquad \qquad \prod_{i=1}^{n} \gt(w_i) \: q_{j' \mid  j}(\betab \mid  \betat) \: q_j(\betat) \: d\wb \: d\betat \: d\betab \\
    & \quad \Ct_1 \int_{\R^p} \int_{\R^p} \int_{\R_+^n} \int_{\R^p}  \: \left(\sum_{i=1}^{n}w_i + 1\right)^{p} \exp\left(\frac12 \sum_{i=1}^n \left|\ut_i^T\betab \right| \right) \exp\left(\frac12 \sum_{i=1}^n \left|\ub_i^T\betat \right| \right) \\
    & \qquad \qquad \qquad \qquad \qquad \qquad \prod_{i=1}^{n} \gt(w_i) \: k(\betas, \betab)\: q_{j'-1 \mid  j}(\betas \mid \betat) \: q_j(\betat) \: d\nut(\betas)  \: d\wb \: d\betat \: d\betab \\
    &= \Ct_1 \left(\int_{\R_+^n} \left(\sum_{i=1}^{n}w_i + 1\right)^{p} \prod_{i=1}^{n} \gt(w_i) \: d\wb \right) \\
    &\qquad \qquad \times \int_{\R^p} \int_{\R^p} \left(\int_{\R^p} \exp\left(\frac12 \sum_{i=1}^n \left|\ut_i^T\betab \right| \right) \: k(\betas, \betab) \: d\betab \right) \\
    & \qquad \qquad \qquad \qquad \times \exp\left(\frac12 \sum_{i=1}^n \left|\ub_i^T\betat \right| \right) \: q_{j'-1 \mid  j}(\betas \mid \betat) \: q_j(\betat) \: d\nut(\betas)  \: d\betat \\
    &= \Ct_1 I_3 \int_{\R^p} \int_{\R^p} \left(\int_{\R^p} \exp\left(\frac12 \sum_{i=1}^n \left|\ut_i^T\betab \right| \right) \: k(\betas, \betab) \: d\betab \right) \\
    & \qquad \qquad \qquad \qquad \times \exp\left(\frac12 \sum_{i=1}^n \left|\ub_i^T\betat \right| \right) \: q_{j'-1 \mid  j}(\betas \mid \betat) \: q_j(\betat) \: d\nut(\betas)  \: d\betat. \numbereqn \label{I_jj'_upper}
    \end{align*}
    Here $\nut$ denotes the Lebesgue measure on $\R^p$ if $j' - 1 > j$, and the counting measure on $\R^p$ if $j'-1 = j$, and $I_3$, as defined in (\ref{I_3_trace}), is finite  (see the proof of Theorem~\ref{tr_thm_psw}).
    To show that $\It_{jj'}$ is bounded, we first find an upper bound for the inner integral on the right hand side of (\ref{I_jj'_upper}). We have
    \begin{align*}
    &\quad \int_{\R^p} \exp\left(\frac12 \sum_{i=1}^n \left|\ut_i^T\betab \right| \right) k(\betas, \betab)  \: d\betab \\
    &=  \int_{\R^p} \int_{\R_+^n} \exp\left(\frac12 \sum_{i=1}^n \left|\ut_i^T\betab \right| \right)  \pi(\betab \mid \wb, \yb) \:\pi(\wb \mid\betas, \yb)  \: d\wb \: d\betab  \\
    &= \Ct_2(X)  \int_{\R^p} \int_{\R_+^n} \exp\left(\frac12 \sum_{i=1}^n \left|\ut_i^T\betab \right| \right)  \: \left|U^T \Omega(\wb) U + B^{-1}\right|^{1/2} \: \exp \left[\betab^T \mub \right] \\
    &\qquad \qquad \qquad \qquad \times  \exp \left[-\frac12 \:\betab^T \left(U^T \Omega(\wb) U + B^{-1}\right)\betab \right]    \:\pi(\wb \mid\betas, \yb)  \: d\wb \: d\betab  \\
    &= \Ct_2(X) \int_{\R^p} \int_{\R_+^n}  \left|U^T \Omega(\wb) U + B^{-1}\right|^{1/2} \exp\left(\frac12 \sum_{i=1}^n \left|\ut_i^T\betab \right| + \betab^T \mub \right)     \\
    &\qquad \qquad \qquad \qquad \times  \exp \left[-\frac12 \:\betab^T \left(U^T \Omega(\wb) U + B^{-1}\right)\betab \right]    \:\pi(\wb \mid\betas, \yb)  \: d\wb \: d\betab \\
    & \stackrel{(\dagger)}{\leq} \Ct_3(X,\Ut)  \int_{\R^p} \int_{\R_+^n} \left|U^T \Omega(\wb) U + B^{-1}\right|^{1/2}  \exp \left(\frac14\: \betab^T B^{-1} \betab \right) \\
    &\qquad \qquad \qquad \qquad \qquad \times  \exp \left[-\frac12 \:\betab^T \left(U^T \Omega(\wb) U + B^{-1}\right)\betab \right]    \:\pi(\wb \mid\betas, \yb)  \: d\wb \: d\betab \\
    & = \Ct_3(X, \Ut)  \int_{\R^p} \int_{\R_+^n} \left|U^T \Omega(\wb) U + B^{-1}\right|^{1/2}  \exp \left[-\frac14 \:\betab^T \left(2U^T \Omega(\wb) U + B^{-1}\right)\betab \right]  \\
    &\qquad \qquad \qquad \qquad \qquad \times   \pi(\wb \mid\betas, \yb)  \: d\wb \: d\betab  \\
    &\leq  \Ct_3(X, \Ut)  \int_{\R_+^n} \left( \int_{\R^p}  \left|U^T \Omega(\wb) U + B^{-1}\right|^{1/2}  \exp \left[-\frac14 \:\betab^T \left(U^T \Omega(\wb) U + B^{-1}\right)\betab \right] \: d\betab\right)  \\
    &\qquad \qquad \qquad \qquad \qquad \times   \pi(\wb \mid\betas, \yb)  \: d\wb   \\
    & \stackrel{(\dagger \dagger)}{=} \Ct_3(X, \Ut) (4\pi)^{p/2}  \int_{\R_+^n} \pi(\wb \mid\betas, \yb)  \: d\wb  = \Ct_3(X, \Ut) (4\pi)^{p/2}.
    \end{align*}
    Here  $\Ct_2(X)$ and $\Ct_3(X, \Ut)$ are constants, $\Ut^T = (\ut_1, \dots, \ut_n)$, $(\dagger)$ follows from Proposition~\ref{exp_prod_upper}, and $(\dagger \dagger)$ follows from the fact that the inner integrand is $(4\pi)^{p/2}$ times a  Gaussian density. Thus, for any $\betas$, 
    \begin{equation} \label{int_exp_k_upper}
    \int_{\R^p} \exp\left(\frac12 \sum_{i=1}^n \left|\ut_i^T\betab \right| \right) k(\betas, \betab)  \: d\betab \leq \Ct_4(X, \Ut)
    \end{equation} 
    where $\Ct_4(X, \Ut) := \Ct_3(X, \Ut) (4\pi)^{p/2}$ is a constant free of $\betas$. Therefore, using this upper bound from  (\ref{int_exp_k_upper}) into (\ref{I_jj'_upper}), we get
    \begin{align*}
    \It_{jj'} &\leq C_1\: I_3 \: \Ct_4(X, \Ut) \int_{\R^p} \int_{\R^p}  \exp\left(\frac12 \sum_{i=1}^n \left|\ub_i^T\betat \right| \right) \: q_{j'-1 \mid  j}(\betas \mid \betat) \: q_j(\betat) \: d\nut(\betas)  \: d\betat \\
    &= C_1 \: I_3 \: \Ct_4(X, \Ut) \int_{\R^p} \left(\int_{\R^p} q_{j'-1 \mid  j}(\betas \mid \betat)  \: d\nut(\betas)  \right) \exp\left(\frac12 \sum_{i=1}^n \left|\ub_i^T\betat \right| \right) \: q_j(\betat) \: d\betat \\
    &= C_1 \: I_3 \: \Ct_4(X, \Ut) \int_{\R^p}  \exp\left(\frac12 \sum_{i=1}^n \left|\ub_i^T\betat \right| \right) \: q_j(\betat) \: d\betat \\
    &= C_1 \: I_3 \: \Ct_4(X, \Ut) \: I_4(q_j), \numbereqn \label{I_jj'_upper_2}
    \end{align*}
    where
    \begin{equation*}
    I_4(q_j) := \int_{\R^p} \exp\left(\frac12 \sum_{i=1}^n \left|\ub_i^T\betat \right| \right) \: q_j(\betat) \: d\betat.
    \end{equation*}
    and we show that for all $0 \leq j \leq m-1$ and all $m = 1,2, \cdots$, $I_4(q_j)$ is bounded above by the same constant. To this end, first observe that $q_0$ is the density associated with the initial distribution $\nu_0$, and hence
    \begin{equation} \label{I_4_q_0_upper}
    I_4(q_0) = \int_{\R^p} \exp\left(\frac12 \sum_{i=1}^n \left|\ub_i^T\betat \right| \right) q_0(\betat) \: d\betat = \int_{\R^p} \exp\left(\frac12 \sum_{i=1}^n \left|\ub_i^T\betat \right| \right) \: d\nu_0(\betat)  < \infty
    \end{equation}
    by assumption. Now for $j\geq 1$,
    \begin{align*}
    I_4(q_j) &= \int_{\R^p} \exp\left(\frac12 \sum_{i=1}^n \left|\ub_i^T\betat \right| \right) q_j(\betat) \: d\betat \\
    &= \int_{\R^p} \int_{\R^p} \exp\left(\frac12 \sum_{i=1}^n \left|\ub_i^T\betat \right| \right) k(\betab, \betat) \: q_{j-1}(\betab) \: d\betab \: d\betat \\
    &= \int_{\R^p} \left(\int_{\R^p} \exp\left(\frac12 \sum_{i=1}^n \left|\ub_i^T\betat \right| \right) k(\betab, \betat) \: d\betat\right)  q_{j-1}(\betab) \: d\betab  \\
    & \leq \Ct_4(X, X) \int_{\R^p}  q_{j-1}(\betab) \: d\betab = \Ct_4(X, X) \numbereqn \label{I_4_q_j_upper},
    \end{align*}
    where the last inequality follows from (\ref{int_exp_k_upper}). Combining (\ref{I_4_q_0_upper}) and (\ref{I_4_q_j_upper}), we get for all $j=0, \dots, m-1$ and all $m\geq 1$
    \begin{equation} \label{I_4_upper}
    I_4(q_j) \leq \Ct_5(X) := \max\{I_4(q_0), \Ct_4(X, X) \}. 
    \end{equation}
    Therefore, from (\ref{I_jj'_upper_2}) and (\ref{I_4_upper}), we have,
    \[
    \It_{jj'} \leq C_1 \: I_3 \: \Ct_4(X, \Ut) \: \Ct_5(X) = C^*
    \]
    where $C^* = C_1 \: I_3 \: \Ct_4(X, \Ut) \: \Ct_5(X)$ is a constant independent of $j, j'$ and $m$. This completes the proof.

\end{proof}

\section{Proofs of some inequalities used in this Appendix} \label{appA}

\begin{prop} \label{prop_det_ineq}
    There exists a constant $a > 0$ such that	
    \[
    \left|U^T \Omega(\wb) U + B^{-1}\right| \leq a \left(\sum_{i=1}^{n}w_i + 1\right)^{p}. 
    \]
\end{prop} 
\begin{supplproof}
    We have,
    \begin{align*}
    U^T \Omega(\wb) U + B^{-1} &= \sum_{i=1}^n w_i \ub_i \ub_i^T + B^{-1} \\
    &\leq \sum_{i=1}^n w_i \left(\sum_{j=1}^n \ub_i \ub_i^T\right) + B^{-1}  \\
    &= \left(\sum_{i=1}^n w_i\right) U^TU   + B^{-1} \leq \left(\sum_{i=1}^n w_i + 1\right) \left(U^TU + B^{-1}\right). 
    \end{align*}
    Therefore,
    \[
    \left|U^T \Omega(\wb) U + B^{-1}\right| \leq a \left(\sum_{i=1}^{n}w_i + 1\right)^{p},\; \text{where } a = \left|U^TU + B^{-1}\right|.
    \]
\end{supplproof}

\begin{prop} \label{prop_exp_F_upper}
    There exist constants $M > 0$ and $\xi > 0$ with $\ut_i = \xi \ub_i$, $i = 1, \dots, n$ such that
    \[
    \exp \left[\betab^T U^T \left(\yb - \frac12 \bm 1_n \right) \right]  \left[\prod_{i=1}^n {\left\lbrace F(\ub_i^T \betab) \right\rbrace}^{y_i} {\left\lbrace 1 - F(\ub_i^T \betab) \right\rbrace}^{1-y_i} \right]^{-1} \leq M \: \exp\left(\frac14 \sum_{i=1}^n\left|\ut_i^T\betab\right| \right).
    \]
\end{prop}

\begin{supplproof}
    First note that, for all $i=1,\cdots,n$,
    \begin{align*}
    F(\ub_i^T \betab) = \frac{e^{\ub_i^T\betab}}{1+e^{\ub_i^T\betab}} = \left(1+e^{-\ub_i^T\betab}\right)^{-1}
    \text{ and } 1 - F(\ub_i^T \betab) &= \left(1 + e^{\ub_i^T\betab} \right)^{-1}.
    \end{align*}
    
    \noindent Therefore, 
    \begin{align*}
    &\quad \left[\prod_{i=1}^n {\left\lbrace F(\ub_i^T \betab) \right\rbrace}^{y_i} {\left\lbrace 1 - F(\ub_i^T \betab) \right\rbrace}^{1-y_i} \right]^{-1} = \; \prod_{i=1}^n \left(1+e^{-\ub_i^T\betab}\right)^{y_i} \left(1 + e^{\ub_i^T\betab} \right)^{1-y_i} \\
    &\leq \prod_{i=1}^n \left(1+e^{-\ub_i^T\betab}\right) \left(1 + e^{\ub_i^T\betab} \right)  = \prod_{i=1}^n \left(2 + e^{-\ub_i^T\betab} + e^{\ub_i^T\betab} \right) \\
    & \leq \prod_{i=1}^n \left(2 e^{\left| \ub_i^T\betab \right|} + e^{\left| \ub_i^T\betab \right|} + e^{\left| \ub_i^T\betab \right|} \right) = 4^n  \exp\left( \sum_{i=1}^n\left|\ub_i^T\betab\right| \right). \numbereqn \label{F_prod}
    \end{align*}
    
    \noindent Also,
    \[
    \betab^T U^T \left(\yb - \frac12 \bm 1_n \right) = \sum_{i=1}^n \ub_i^T\betab \left(y_i - \frac12\right) \leq  \sum_{i=1}^n \left|\ub_i^T\betab\right|  \left|y_i - \frac12\right| \leq \xi_0 \sum_{i=1}^n \left|\ub_i^T\betab\right|.
    \]
    where $\xi_0 = \max_{1 \leq i \leq n} |y_i - \frac12|$. This implies
    \begin{equation} \label{exp_term}
    \exp\left[ \betab^T U^T \left(\yb - \frac12 \bm 1_n \right) \right] \leq \left[ \xi_0 \sum_{i=1}^n \exp \left|\ub_i^T\betab\right| \right].
    \end{equation}
    Therefore, from (\ref{F_prod}) and (\ref{exp_term}) we have
    \begin{align*}
    &\quad \exp \left[\betab^T U^T \left(\yb - \frac12 \bm 1_n \right) \right]  \left[\prod_{i=1}^n {\left\lbrace F(\ub_i^T \betab) \right\rbrace}^{y_i} {\left\lbrace 1 - F(\ub_i^T \betab) \right\rbrace}^{1-y_i} \right]^{-1} \\
    &\leq 4^n \exp \left[(1+\xi_0) \sum_{i=1}^n  \left|\ub_i^T\betab\right|\right] = 4^n \exp \left[\frac14 \sum_{i=1}^n  \left|(\xi\ub_i)^T\betab\right|\right] \\ 
    &= M  \exp \left[\frac14 \sum_{i=1}^n  \left|\ut_i^T\betab\right|\right],
    \end{align*}
    with $M = 4^n$, $\ut_i = \xi \ub_i$ for all $i = 1, \dots, n$, and $\xi = 4(1+\xi_0)$.
\end{supplproof}

\begin{prop} \label{exp_prod_upper}
    For an appropriately chosen constant $a_0$,
    \[
    \exp\left(\frac12 \sum_{i=1}^n \left|\ub_i^T\betab \right| + \betab^T \mub \right)  \leq a_0\: \exp \left(\frac14\: \betab^T B^{-1} \betab \right).
    \]
\end{prop}
\begin{supplproof}
    Let $\epsilon > 0$ be arbitrary. Then, by AM-GM inequality,
    \begin{align*}
    &\frac{\left|\ub_i^T\betab \right|}{2} = \sqrt{\frac{\epsilon \left(\ub_i^T\betab \right)^2}{4}\cdot \frac{1}{\epsilon}} \leq \frac{\epsilon\: \frac{\left(\ub_i^T\betab \right)^2}{4} + \frac{1}{\epsilon}}{2} = \frac18\:\epsilon\: \betab^T \ub_i\ub_i^T \betab + \frac{1}{2\epsilon} \\
    \implies & \frac12 \sum_{i=1}^n \left|\ub_i^T\betab \right| \leq \frac18\:\epsilon\: \betab^T U^TU \betab + \frac{n}{2\epsilon} \\
    \implies & \exp \left(\frac12 \sum_{i=1}^n \left|\ub_i^T\betab \right|\right) \leq a' \exp \left[ \frac18\: \betab^T (\epsilon \: U^TU) \betab\right] \leq  a' \exp \left(\frac18\: \betab^T B^{-1} \betab\right) \numbereqn \label{upper_1}
    \end{align*}
    where $a' = \exp\left(\frac{n}{2\epsilon}\right)$, and $\epsilon$ is chosen to be sufficiently small such that $B^{-1} - \epsilon U^TU \geq 0$. Note that because $B$ and hence $B^{-1}$ is positive definite, and $U^TU$ is positive semi-definite, there always exists such an $\epsilon > 0$. (See, e.g., Proposition~A.3 in \citet{chakraborty:khare:2017}.) Now, letting $\betas =  \left(\frac12 B^{-1/2}\right) \betab$ and $\mustar = \left(2B^{1/2}\right)\mub$ yields
    \begin{align*}
    & \betab^T\mub = \betas^T \mustar  \stackrel{(\star)}{\leq} \sqrt{\left(\betas^T \betas\right) \left(\mustar^T\mustar\right) } \stackrel{(\star \star)}{\leq} \frac{\betas^T\betas + \mustar^T \mustar}{2} \\
    \implies & \exp \left(\betab^T\mub \right) \leq a'' \exp \left(\frac12 \: \betas^T\betas \right) = a'' \exp \left(\frac18 \betab^T B^{-1} \: \betab \right), \numbereqn \label{upper_2} 
    \end{align*}
    where $(\star)$ follows form Cauchy-Schwarz inequality, $(\star \star)$ follows from AM-GM inequality, and 
    \[
    a'' = \exp\left(\frac12\:\mustar^T \mustar\right).
    \]
    
    \noindent Therefore, from (\ref{upper_1}) and (\ref{upper_2}), and by letting $a_0 = a'\cdot a''$, we get
    \[
    \exp\left(\frac12 \sum_{i=1}^n \left|\ub_i^T\betab \right| + \betab^T \mub \right)  \leq  a_0\: \exp \left(\frac14\: \betab^T B^{-1} \betab \right).
    \]
    
\end{supplproof}

\begin{prop}\label{prop_I1_intgrnd_upper}
    For an appropriately chosen constant $a_0$,
    \[ 
    \exp\left(-\frac12 \: \betab^T B^{-1} \betab + \betab^T\mub + \frac12 \sum_{i=1}^n|\ub_i^T\betab \mid  \right) \leq a_0 \: \exp\left(-\frac14 \: \betab^T B^{-1} \betab \right).
    \]
\end{prop}
\begin{supplproof}
    This is result is a corollary to Proposition~\ref{exp_prod_upper}, which ensures the existence of a constant $a_0$ such that
    \[
    \exp\left(\frac12 \sum_{i=1}^n \left|\ub_i^T\betab \right| + \betab^T \mub \right)  \leq  a_0\: \exp \left(\frac14\: \betab^T B^{-1} \betab \right).
    \]
    
    \noindent The proof is completed by multiplying both sides of the above inequality by $\exp\left(-\frac12 \: \betab^T B^{-1} \betab\right)$. 
\end{supplproof}

\end{appendix}


\end{document}